\documentclass[UKenglish,cleveref,numberwithinsect,thm-restate,letter]{no-lipics-v2019}

\RequirePackage{doi}
\usepackage{hyperref}
\usepackage{blkarray}
\usepackage{bm}
\usepackage{rotating}
\usepackage{afterpage}
\usepackage{booktabs}
\allowdisplaybreaks

\theoremstyle{definition}

\newtheorem{question}[theorem]{Question}

\usepackage{amsmath, amsthm}
\usepackage{amssymb}
\usepackage{enumitem}
\usepackage{graphicx}
\usepackage{mathdots}
\usepackage{color}
\usepackage{hyperref,url}
\usepackage{tikz, tikz-cd}
\usetikzlibrary{matrix}
\usetikzlibrary{shapes}
\usetikzlibrary{arrows,decorations.markings, tikzmark}
\usepackage{mathtools}
\usepackage{ stmaryrd }
\usepackage[normalem]{ulem}
\usepackage[utf8]{inputenc}
\usepackage{bbm}
\usepackage{bbold}

\usepackage{xcolor}
\hypersetup{
    colorlinks,
    linkcolor={red!50!black},
    citecolor={blue!50!black},
    urlcolor={blue!80!black}
}

\newcommand{\lovasz}{Lov{\'{a}}sz}

\DeclareMathOperator{\GL}{GL}

\newcommand{\pr}{\mathrm{pr}}

\newcommand{\homs}[2]{\mbox{\ensuremath{\mathsf{Hom}(#1 \to #2)}}}

\newcommand{\embs}[2]{\mbox{\ensuremath{\mathsf{Emb}(#1 \to #2)}}}

\newcommand{\subs}[2]{\mbox{\ensuremath{\mathsf{Sub}(#1 \to #2)}}}

\newcommand{\edgesubs}[2]{\mbox{\ensuremath{\mathsf{EdgeSub}(#1 \to #2)}}}
\newcommand{\auts}[1]{\ensuremath{\mathsf{Aut}(#1)}}
\newcommand{\cphoms}[2]{\ensuremath{\mathsf{cp}\text{-}\mathsf{Hom}}(#1 \to #2)}
\newcommand{\cfhoms}[2]{\ensuremath{\mathsf{cf}\text{-}\mathsf{Hom}}(#1 \to #2)}

\newcommand{\coledgesubs}[2]{\mbox{\ensuremath{\mathsf{ColEdgeSub}(#1 \to #2)}}}

\newcommand{\fptred}{\ensuremath{\leq^{\mathrm{fpt}}_{\mathrm{T}}}}

\newcommand{\colmatch}{\text{\sc{ColMatch}}}
\newcommand{\homsprob}{\text{\sc{Hom}}}
\newcommand{\cphomsprob}{\text{\sc{cp-Hom}}}

\newcommand{\edgesubsprob}{\text{\sc{EdgeSub}}}

\newcommand{\coledgesubsprob}{\text{\sc{ColEdgeSub}}}

\newcommand{\NP}{\ensuremath{\mathsf{NP}}}
\newcommand{\ccP}{\ensuremath{\mathsf{P}}}

\newcommand{\W}[1]{\ensuremath{\mathsf{W[#1]}}}
\newcommand{\pW}[1]{\ensuremath{\mathsf{Mod}_p\mathsf{W[#1]}}}
\newcommand{\threeW}[1]{\ensuremath{\mathsf{Mod}_3\mathsf{W[#1]}}}

\def\fracture#1#2{\ensuremath{#1\raisebox{.2ex}{\rotatebox[origin=c]{-15}{$\sharp$}}#2}}
\newcommand{\tu}{\ensuremath{\vartriangle}}
\newcommand{\td}{\ensuremath{\triangledown}}
\newcommand{\tl}{\ensuremath{\vartriangleleft}}
\newcommand{\tr}{\ensuremath{\vartriangleright}}
\def\basegraph#1{\ensuremath{\mathcal{H}(#1)}}

\newcommand{\redc}[2][red,fill=red]{\tikz[baseline=-0.5ex]\draw[#1,radius=#2] (0,0) circle ;}
\newcommand{\greenc}[2][black!30!green,fill=black!30!green]{\tikz[baseline=-0.5ex]\draw[#1,radius=#2] (0,0) circle ;}
\newcommand{\bluec}[2][blue,fill=blue]{\tikz[baseline=-0.5ex]\draw[#1,radius=#2] (0,0) circle ;}
\newcommand{\yellowc}[2][black!10!yellow,fill=black!10!yellow]{\tikz[baseline=-0.5ex]\draw[#1,radius=#2] (0,0) circle ;}
\newcommand{\blackc}[2][black,fill=black]{\tikz[baseline=-0.5ex]\draw[#1,radius=#2] (0,0) circle ;}
\newcommand{\brownc}[2][brown,fill=brown]{\tikz[baseline=-0.5ex]\draw[#1,radius=#2] (0,0) circle ;}

\title{Parameterized (Modular) Counting and Cayley Graph Expanders}

\author{Norbert Peyerimhoff}{Department of Mathematical Sciences, Durham University, United Kingdom}{norbert.peyerimhoff@durham.ac.uk}{}{}

\author{Marc Roth}{Merton College, University of Oxford, United Kingdom}{marc.roth@merton.ox.ac.uk}{https://orcid.org/0000-0003-3159-9418}{}

\author{Johannes Schmitt}{Mathematical Institute, University of Bonn, Germany}{schmitt@math.uni-bonn.de}{https://orcid.org/0000-0001-5774-3508}{}

\author{Jakob Stix}{Mathematical Institute, Goethe-Universität Frankfurt, Germany}{stix@math.uni-frankfurt.de}{}{}

\author{Alina Vdovina}{School of Mathematics and Statistics, Newcastle University, United Kingdom}{Alina.Vdovina@newcastle.ac.uk}{}{}

\authorrunning{N. Peyerimhoff, M. Roth, J. Schmitt, J. Stix and A. Vdovina}

\Copyright{Norbert Peyerimhoff, Marc Roth, Johannes Schmitt, Jakob Stix and Alina Vdovina}


\nolinenumbers


\acknowledgements{The second author is grateful to Holger Dell for fruitful discussions on early drafts of parts of this work.
The third author was supported by the SNF Early Postdoc.Mobility grant 184245 and thanks
the Max Planck Institute for Mathematics in Bonn for its hospitality. }

\begin{document}
\maketitle

\begin{abstract}
    We study the problem $\#\edgesubsprob(\Phi)$ of counting $k$-edge subgraphs satisfying a given graph property $\Phi$ in a large host graph $G$. Building upon the breakthrough result of Curticapean, Dell and Marx (STOC 17), we express the number of such subgraphs as a finite linear combination of graph homomorphism counts and derive the complexity of computing this number by studying its coefficients. 
    
    Our approach relies on novel constructions of low-degree Cayley graph expanders of $p$-groups, which might be of independent interest. The properties of those expanders allow us to analyse the coefficients in the aforementioned linear combinations over the field $\mathbb{F}_p$ which gives us significantly more control over the cancellation behaviour of the coefficients. Our main result is an exhaustive and fine-grained complexity classification of $\#\edgesubsprob(\Phi)$ for minor-closed properties $\Phi$, closing the missing gap in previous work by Roth, Schmitt and Wellnitz (ICALP 21). 
    
    Additionally, we observe that our methods also apply to modular counting. Among others, we obtain novel intractability results for the problems of counting $k$-forests and matroid bases modulo a prime $p$. Furthermore, from an algorithmic point of view, we construct algorithms for the problems of counting $k$-paths and $k$-cycles modulo $2$ that outperform the best known algorithms for their non-modular counterparts.
    
    In the course of our investigations we also provide an exhaustive parameterized complexity classification for the problem of counting graph homomorphisms modulo a prime $p$.
\end{abstract}

\section{Introduction}
In this work we study the problem of counting small patterns in large host graphs. With applications in a diverse set of disciplines such as constraint satisfaction problems~\cite{DalmauJ04,BulatovZ20}, database theory~\cite{DurandM15,ChenM16} and network science~\cite{Miloetal02,Nogaetat08,Schilleretal15}, it is unsurprising that this problem has received significant attention from the viewpoint of parameterized and fine-grained complexity theory in recent years~\cite{ArvindR02,FlumG04,McCartin06,Curticapean13,CurticapeanM14,CurticapeanDM17,BrandDH18,RothS18,RothSW20,LokshtanovSZ21,RothSW20unpub}.

We continue this line of work and study the problem of counting $k$-edge subgraphs that satisfy a graph property $\Phi$: For any fixed $\Phi$, the problem $\#\edgesubsprob(\Phi)$ asks, on input a graph $G$ and a positive integer $k$, to compute the number of (not necessarily induced) subgraphs with $k$ edges in $G$ that satisfy $\Phi$. In particular, we focus on instances in which $k$ is significantly smaller than $G$. Formally, we choose $k$ to be the \emph{parameter} of the problem and ask for which $\Phi$ there is a function $f$ such that $\#\edgesubsprob(\Phi)$ can be solved in time $f(k)\cdot |V(G)|^{O(1)}$; in this case we call the problem \emph{fixed-parameter tractable} with respect to the parameter $k$. 

If $\#\edgesubsprob(\Phi)$ is not fixed-parameter tractable, it is desirable to improve the exponent of $|V(G)|$ in the running time as far as possible. For example, the best known algorithm for counting $k$-edge subgraphs~\cite{CurticapeanDM17} can be used to solve $\#\edgesubsprob(\Phi)$ in time $f(k)\cdot |V(G)|^{0.174 k +o(k)}$~\cite{RothSW20unpub}. Additionally, it was shown in recent work that $\#\edgesubsprob(\Phi)$ is fixed-parameter tractable whenever $\Phi$ has \emph{bounded matching number}, that is, whenever there is a constant upper bound on the size of the largest matching of any graph satisfying $\Phi$~\cite{RothSW20unpub}. If, for each $k$, the property $\Phi$ is true for only one graph on $k$ edges, then the previous fixed-parameter tractability result is best possible: In this case, $\#\edgesubsprob(\Phi)$ becomes an instance of the counting version of the parameterized subgraph isomorphism problem which has been fully classified by Curticapean and Marx~\cite{CurticapeanM14}. 

However, for arbitrary $\Phi$, much less is known about the complexity of $\#\edgesubsprob(\Phi)$. In~\cite{RothSW20unpub}, two of the authors, together with Wellnitz, presented first results for more general properties such as connectivity, Eulerianity and, in particular, an \emph{almost} exhaustive classification for minor-closed properties $\Phi$, leaving (partially) open the case of forbidden minors of degree at most $2$. In this work, we close this gap and provide a full dichotomy result:
\begin{restatable}{theorem}{minclosehard}\label{thm:minclose-hard}
Let $\Phi$ be a minor-closed graph property. If $\Phi$ is trivially true or of bounded matching number, then $\#\edgesubsprob(\Phi)$ is fixed-parameter tractable. 
Otherwise, $\#\edgesubsprob(\Phi)$ is $\#\W{1}$-hard and, assuming the Exponential Time Hypothesis, it cannot be solved in time \[f(k)\cdot |G|^{o(k/\log k)}\] 
for any function $f$.
\end{restatable}
Here, $\#\W{1}$ is the parameterized counting analogue of $\NP$; a formal definition is provided in Section~\ref{sec:prelims}. Particular cases for which we obtain novel intractability results are given by the following (minor-closed) properties; the formal intractability results are stated and proved in Section~\ref{sec:lowerbounds} as Corollaries~\ref{cor:forest_exact_and_3}, \ref{cor:treedepth}, and~\ref{cor:CDVinv}.
\begin{itemize}
    \item $\Phi(H)=1$ if $H$ is a forest.
    \item $\Phi(H)=1$ if $H$ is a linear forest.
    \item $\Phi(H)=1$ if the tree-depth of $H$ is bounded by a constant.
    \item $\Phi(H)=1$ if the Colin de Verdi\`ere Invariant of $H$ is bounded by a constant.
\end{itemize}
Additionally, we investigate the property of being bipartite. For this case, we present not only a novel fine-grained lower bound, but also a $\#\W{1}$-hardness result, which was not known before.
\begin{restatable}{theorem}{bipartitehard}\label{thm:bipartite-hard}
Let $\Phi$ be the property of being bipartite.
Then $\#\edgesubsprob(\Phi)$ is $\#\W{1}$-hard and, assuming the Exponential Time Hypothesis, it cannot be solved in time \[f(k)\cdot |G|^{o(k/\log k)}\] 
for any function $f$.
\end{restatable}

Our hardness results crucially rely on a novel construction of families of low-degree Cayley graph expanders of $p$-groups, which might be of independent interest. We will present the new Cayley graph expanders in the following theorem; their construction, as well as their role in the hardness proofs for $\#\edgesubsprob(\Phi)$ will be elaborated on in Section~\ref{sec:intro_techniques}.
\begin{restatable}{theorem}{main_expanders}\label{thm:main_expanders_intro}
Let $p \ge 3$ be a prime number, and $d \geq 2$ be an integer. We assume that $d \ge (p+3)/2$ if $p \ge 7$. 

Then there is an explicit construction of a sequence of finite $p$-groups $\Gamma_i$ of orders that tend to infinity, with symmetric generating sets $S_i$ of cardinality $2d$ such that the Cayley graphs $\mathcal{C}(\Gamma_i, S_i)$ form a family of expanders (of fixed valency $2d$ on a set of vertices of $p$-power orders and with vertex transitive automorphism groups).
\end{restatable}

Our methods do not only apply to exact counting, but also to modular pattern counting problems: Here the goal is to compute the number of occurrences of the pattern \emph{modulo a fixed prime} $p$. In classical complexity theory, the study of modular counting problems has a rich history, such as the algorithm for computing the permanent modulo $2^\ell$~\cite{Valiant79}, the so-called accidental algorithms~\cite{Valiant08}, Toda's Theorem~\cite{Toda91}, classifications for modular \#CSPs and Holants~\cite{GuoHLX11,GuoLV13} and the line of research on the modular homomorphism counting problem~\cite{FabenJ15,GobelGR14,GobelGR16,GobelLS18,KazeminiaB19,FockeGRZ21,GobelLCF21}, only to name a few. 

While results are scarcer, the \emph{parameterized} complexity of modular (pattern) counting problems has also been studied in recent years~\cite{BjorklundDH15,DorflerRSW19,CurticapeanDH21}, and we contribute to this line of research as follows: First, we provide a novel intractability result for modular counting of forests and matroid bases. We write $\#_p\textsc{Forests}$ for the problem of, given a graph $G$ and a positive integer~$k$, computing the number of forests with $k$ edges in $G$, modulo $p$. Similarly, we write $\#_p\textsc{Bases}$ for the problem of, given a linear matroid $M$ of rank $k$ in matrix representation, computing the number of bases of $M$, modulo $p$; the parameter of both problems is given by $k$.

\begin{restatable}{theorem}{forestmatroids}\label{thm:forest_matroids}
For each prime $p\geq 3$, the problems $\#_p\textsc{Forests}$ and $\#_p\textsc{Bases}$ are $\mathsf{Mod}_p\W{1}$-hard and, assuming the randomised Exponential Time Hypothesis, cannot be solved in time $f(k)\cdot |G|^{o(k/\log k)}$ (resp.\ $f(k)\cdot |M|^{o(k/\log k)}$), for any function $f$.
\end{restatable}
Here, $\mathsf{Mod}_p\W{1}$ is the parameterized modular counting version of $\NP$. Roughly speaking, a problem is $\mathsf{Mod}_p\W{1}$ if it is at least as hard as counting $k$-cliques modulo $p$; we give a formal definition in Section~\ref{sec:prelims}.
Additionally, we provide an algorithmic result for counting $k$-paths and $k$-cycles modulo $2$:

\begin{restatable}{theorem}{maincyclesmod}\label{thm:main_cycles_mod_2}
The problems of counting $k$-paths and $k$-cycles in a graph $G$ modulo $2$ can be solved in time $k^{O(k)}\cdot |V(G)|^{k/6 + O(1)}$.
\end{restatable}

We emphasize that the algorithm in the previous theorem is faster than the best known algorithms for (non-modular) counting of $k$-cycles/$k$-paths, which run in time $k^{O(k)}\cdot |V(G)|^{13k/75 + o(k)}$~\cite{CurticapeanDM17}. Furthermore, it follows from a result by Curticapean, Dell and Husfeldt~\cite{CurticapeanDH21} that counting $k$-paths modulo $2$ is $\mathsf{Mod}_2\W{1}$-hard, implying that we cannot hope for an algorithm for $k$-paths running in time $f(k)\cdot |V(G)|^{O(1)}$. \pagebreak

\noindent Finally, we study the parameterized complexity of counting homomorphisms modulo~$p$. In the classical setting, the related problem of modular counting homomorphisms with \emph{right-hand side} restrictions received much attention: For any fixed graph $H$, the problem $\#_p\textsc{HomsTo}(H)$\footnote{For $p=2$, the problem is usually denoted by $\oplus\textsc{HomsTo}(H)$.} asks, on input a graph $G$, to compute the number of homomorphisms from~$G$ to~$H$ modulo $p$. Despite significant effort~\cite{FabenJ15,GobelGR14,GobelGR16,GobelLS18,KazeminiaB19,FockeGRZ21,GobelLCF21}, the problem has not been fully classified for each graph $H$.

In this work, we consider the related \emph{left-hand side} version of the problem. Adapting the definitions of Grohe, Dalmau and Jonsson~\cite{Grohe07,DalmauJ04} for detecting and exact counting of homomorphisms, we define a problem $\#_p\homsprob(\mathcal{H})$ for each class of graphs $\mathcal{H}$ and for each prime $p$: This problem expects as input a graph
$H\in \mathcal{H}$ and an arbitrary graph $G$, and the goal is to compute the number of homomorphisms from $H$ to $G$, modulo $p$. The problem is parameterized by the size of $H$, that is, we assume $H$ to be significantly smaller than $G$.

It is known that the decision version $\homsprob(\mathcal{H})$ is fixed-parameter tractable (even polynomial-time solvable) if the treewidth of the cores of $\mathcal{H}$ is bounded by a constant, and $\W{1}$-hard otherwise~\cite{Grohe07}. Similarly, the (exact) counting version $\#\homsprob(\mathcal{H})$ is known to be fixed-parameter tractable (even polynomial-time solvable) if the treewidth of the graphs of $\mathcal{H}$ is bounded by a constant, and $\#\W{1}$-hard otherwise~\cite{DalmauJ04}.

In case of counting modulo $p$, we establish an exhaustive classification for $\#_p\homsprob(\mathcal{H})$ along what we call the $p$\emph{-reduced quotients} of the graphs in $\mathcal{H}$. Let $H$ be a graph and let $\alpha$ be an automorphism of $H$ of order $p$. Then we define the quotient graph $H/\alpha$ to have a vertex for each orbit of the action of $\alpha$ on $V(H)$, and two vertices corresponding to orbits $\mathcal{O}_1$ and $\mathcal{O}_2$ are made adjacent if and only if there are vertices $v_1\in \mathcal{O}_1$ and $v_2 \in \mathcal{O}_2$ such that $\{v_1,v_2\}\in E(H)$. This induces a finite (possibly trivial) sequence $H=H_1,\dots,H_\ell$ where for $i=1, \ldots, \ell-1$ we set $H_{i+1}=H_i/\alpha_i$ for some automorphism $\alpha_i$ of order $p$ of $H_i$ and where the last graph $H_\ell$ does not have an automorphism of order $p$. Then $H^\ast_p:=H_\ell$ is called the $p$-reduced quotient of $H$.\footnote{We remark that $H^\ast_p$ is related to the notion of involution-free reductions used in the analysis of the right-hand side version of the problem~\cite{FabenJ15,GobelLS18}. However, the difference is that the $p$-reduced quotient identifies non-fixed points of an order-$p$ automorphism by including a vertex for each orbit, while the involution-free reduction just deletes all non-fixed points.}
We will see that $H^\ast_p$ is well-defined by proving that each of the aforementioned sequences yields the same graph, up to isomorphism. 

Let us now state our classification for $\#_p\homsprob(\mathcal{H})$. In what follows, given a class $\mathcal{H}$, we write $\mathcal{H}^\ast_p$ for the $p$-reduced quotients without self-loops of graphs in $\mathcal{H}$. We first present the algorithmic part:

\begin{restatable}{theorem}{mainhomsalgo}\label{thm:main_homs_algo}
Let $p\geq 2$ be a prime and let $\mathcal{H}$ be a class of graphs.
The problem $\#_p\homsprob(\mathcal{H})$ can be solved in time
\[\exp(\mathsf{poly}(|V(H)|)) \cdot  |V(G)|^{\mathsf{tw}(H^\ast_p)+O(1)}\,.\]
In particular, $\#_p\homsprob(\mathcal{H})$ is fixed-parameter tractable if the treewidth of $\mathcal{H}^\ast_p$ is bounded.
\end{restatable}
Here $\mathsf{tw}$ denotes treewidth.
\begin{remark}
Using the quasi-polynomial time algorithm for $\textsc{GI}$ due to Babai~\cite{Babai16}, we will also show how the algorithm in the previous theorem can be improved to run in quasi-polynomial time. Additionally, proving that the construction of the $p$-reduced quotient is at least as hard as the graph automorphism problem, we observe that a polynomial-time algorithm is unlikely, unless the construction of the $p$-reduced quotients can be avoided.
\end{remark} \pagebreak

\noindent For the intractability part of our classification, we show that unbounded treewidth of $\mathcal{H}^\ast_p$ yields hardness:

\begin{restatable}{theorem}{mainhomshard}\label{thm:main_homs_hard}
Let $p\geq 2$ be a prime and let $\mathcal{H}$ be a computable class of graphs. If the treewidth of $\mathcal{H}^\ast_p$ is unbounded, then $\#_p\homsprob(\mathcal{H})$ is $\pW{1}$-hard and, assuming the randomised Exponential Time Hypothesis, cannot be solved in time
\[f(|H|) \cdot |G|^{o(\mathsf{tw}(H^\ast_p)/\log \mathsf{tw}(H^\ast_p))} \]
for any function $f$.
\end{restatable}

\paragraph*{Can You Beat Treewidth?}
We conclude the presentation of our results by commenting on the factor of $1/(\log \dots)$ in the exponents of all of our fine-grained lower bounds. This factor is related to the conjecture of whether it is possible to ``beat treewidth''~\cite{Marx10}. In particular, we point out that the factor can be dropped in \emph{all} of our lower bounds if this conjecture, formally stated as Conjecture~1.3 in~\cite{Marx13}, is true.

\subsection*{Technical Overview}\label{sec:intro_techniques}
Our central approach follows the so-called Complexity Monotonicity framework due to Curticapean, Dell and Marx~\cite{CurticapeanDM17}. We express the counting problems considered in this work as formal linear combination of homomorphism counts, which allows us to derive the complexity of the problem at hand by analyzing the coefficients. 

More precisely, let us fix a graph property $\Phi$ and a positive integer $k$. Given a graph $G$, we furthermore write $\#\edgesubs{\Phi,k}{G}$ for the number of $k$-edge subgraphs of $G$ that satisfy $\Phi$. It was shown in~\cite{RothSW20unpub} that there exists a function of finite support $a_{\Phi,k}$ from graphs to rationals such that for every graph $G$ we have
\begin{equation}\label{eq:intro_hom_lincomb}
    \#\edgesubs{\Phi,k}{G} = \sum_{H} a_{\Phi,k}(H) \cdot \#\homs{H}{G}\,,
\end{equation}
where $\#\homs{H}{G}$ is the number of graph homomorphisms from $H$ to $G$. Curticapean, Dell and Marx~\cite{CurticapeanDM17} have shown that computing a linear combination as in~\eqref{eq:intro_hom_lincomb} is \emph{precisely as hard as} computing its hardest term. Fortunately, the complexity of counting and detecting homomorphisms from $H$ to $G$ is thoroughly classified~\cite{DalmauJ04,Marx10}: Roughly speaking, the higher the treewidth of $H$, the harder it is to compute the number of homomorphisms from~$H$ to~$G$. Therefore, proving hardness of computing $\#\edgesubs{\Phi,k}{G}$ reduces to the purely combinatorial problem of determining which of the coefficients $a_{\Phi,k}(H)$ in~\eqref{eq:intro_hom_lincomb} for high-treewidth graphs $H$ are non-zero.

Unfortunately, it has turned out that the coefficients of such linear combinations for related pattern counting problems are often determined by (or even equal to) a variety of algebraic and topological invariants, whose analysis is known to be a difficult problem in its own right. For example, in case of the vertex-induced subgraph counting problem, the coefficient of the clique is the reduced Euler characteristic of a simplicial graph complex~\cite{RothS18}, the coefficient of the biclique is the so-called alternating enumerator~\cite{DorflerRSW19}, and, more generally, the coefficients of dense graphs are related to the $h$-and $f$-vectors associated with the property of the patterns that are to be counted~\cite{RothSW20}. In all of the previous works mentioned here, the complexity analysis of the respective pattern counting problems therefore amounted to understanding the cancellation behaviour of those invariants. To do so, the papers used tools from combinatorial commutative algebra and, to some extent, topological fixed-point theorems. \pagebreak

\noindent In case of $\#\edgesubsprob(\Phi)$, two of the authors, together with Wellnitz, observed that the coefficients of high-treewidth low-degree vertex-transitive graphs can be analysed much easier than generic graphs of high treewidth such as the clique or the biclique~\cite{RothSW20unpub}. First, it was shown that the coefficient of a graph $H$ with $k$ edges in~\eqref{eq:intro_hom_lincomb} is equal to the \emph{indicator} of $\Phi$ and $H$, defined as follows:\footnote{To be precise, the identity in~\eqref{eq:indicator_intro} was obtained for a coloured version of $\#\edgesubsprob(\Phi)$. However, we will mostly rely on this result in a blackbox manner; all details of the coloured version necessary for the treatment in this paper will be carefully introduced when needed.}
\begin{equation}\label{eq:indicator_intro}
    a(\Phi,H) := \sum_{\sigma \in \mathcal{L}(\Phi,H)} \prod_{v\in V(H)} (-1)^{|\sigma_v|-1} (|\sigma_v|-1)!\,.
\end{equation}
Here, $\mathcal{L}(\Phi,H)$ is the set of \emph{fractures} $\sigma$ of $H$ such that the associated \emph{fractured graph} $\fracture{H}{\sigma}$ satisfies $\Phi$. Here, a fracture of a graph $H$ is a tuple $\rho = (\rho_v)_{v\in V(H)}$, where~$\rho_v$ is a partition of the set of edges $E_H(v)$ of $H$ incident to $v$. 
Given a fracture $\rho$ of $H$, the {fractured graph} $\fracture{H}{\rho}$ is obtained from $H$ be splitting each vertex $v\in V(H)$ according to $\rho_v$; an illustration is provided in Figure~\ref{fig:simpler_fractures_intro}. 

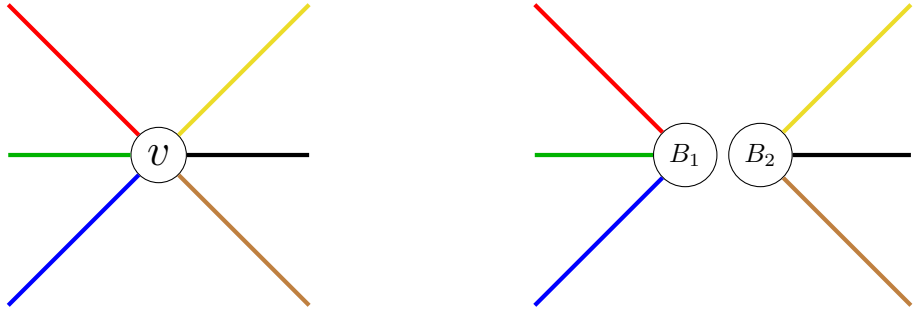
\begin{figure}[t!]
    \centering
    \begin{tikzpicture}[scale=1]
    \node[circle,draw] (m) at (0,0) {\LARGE $v$};
    \draw[ultra thick,red] (m) -- (-2,2);
    \draw[ultra thick,black!30!green] (m) -- (-2,0);
    \draw[ultra thick,blue] (m) -- (-2,-2);
    
    \draw[ultra thick,black!10!yellow] (m) -- (2,2);
    \draw[ultra thick] (m) -- (2,0);
    \draw[ultra thick,brown] (m) -- (2,-2);
    \begin{scope}[shift={(7,0)}]
    \node[circle,draw] (b1) at (0,0) {$B_1$};
    \draw[ultra thick,red] (b1) -- (-2,2);
    \draw[ultra thick,black!30!green] (b1) -- (-2,0);
    \draw[ultra thick,blue] (b1) -- (-2,-2);
    \begin{scope}[shift={(1,0)}]
    \node[circle,draw] (b2) at (0,0) {$B_2$};
    \draw[ultra thick,black!10!yellow] (b2) -- (2,2);
    \draw[ultra thick] (b2) -- (2,0);
    \draw[ultra thick,brown] (b2) -- (2,-2);
    \end{scope}
    \end{scope}
    \end{tikzpicture}
    \caption{\label{fig:simpler_fractures_intro} Illustration of the construction of a fractured graph. The left picture shows a vertex $v$ of a graph~$H$ with incident edges $E_H(v)=\{  \redc{2pt},\greenc{2pt},\bluec{2pt},\yellowc{2pt},\blackc{2pt},\brownc{2pt}\}$. The right picture shows the splitting of $v$ in the construction of the fractured graph~$\fracture{H}{\rho}$ for a fracture $\rho$ satisfying that the partition $\rho_v$ contains two blocks $B_1 =\{ \redc{2pt},\greenc{2pt},\bluec{2pt}\}$, and $B_2=\{\yellowc{2pt},\blackc{2pt},\brownc{2pt}\}$.}
\end{figure}

As a consequence, the $\#\W{1}$-hardness results of Theorems~\ref{thm:minclose-hard} and~\ref{thm:bipartite-hard} can be obtained if we find a family of graphs $H$ of unbounded treewidth, such that $a(\Phi,H)\neq 0$ for infinitely many graphs $H$ in this family. The almost tight conditional lower bound under the Exponential Time Hypothesis will, additionally, require sparsity of the graphs. In combination with the main observation in~\cite{RothSW20unpub}, stating that the indicator $a(\Phi,H)$ can be analysed much easier for vertex-transitive graphs, we propose that regular Cayley graph expanders are the right choice for the family of graphs to be considered. Indeed, those graphs are sparse, have high treewidth and are always vertex transitive. A particular family of Cayley graph expanders was already used in~\cite{RothSW20unpub}, but it turned out to be impossible to prove Theorems~\ref{thm:minclose-hard} and~\ref{thm:bipartite-hard} relying only on this family of Cayley graph expanders.

In this work, we therefore present novel constructions of families of low-degree Cayley graph expanders. Those will not only allow us to prove most of our main theorems by analysing their indicators, but might be of independent interest. For the sake of presentation, we decided to encapsulate the treatment of our constructions in separate sections, both in the extended abstract and the main part of the paper. We hope that this makes the paper accessible both for readers primarily interested in the novel construction of Cayley graph expanders, as well as for readers  mainly interested in the analysis of the pattern counting problems. In particular, this last group may safely skip the next subsection and rely only on Theorem~\ref{thm:main_expanders_intro}.

\paragraph*{Construction of Low-Degree Cayley Graph Expanders}
We prove \cref{thm:main_expanders_intro} via an explicit  construction of the groups $\Gamma_i$ and the symmetric generating sets $S_i$ in \cref{Sect:CayleyExpanders} motivated by number theoretic objects. 

Let us fix a prime $p \geq 3$. The starting point is an explicit arithmetic lattice (a discrete subgroup) in a group of generalized quaternions over a function field in characteristic $p$. The quaternion algebra is at the heart of the mathematical properties of extracting the finite $p$-groups and the expansion property of the resulting Cayley graphs, but it is not crucial for understanding the construction. Concretely, for any choice of elements $\alpha \neq \beta \in \mathbb{Z}/(p-1) \mathbb{Z}$ we construct an infinite group $\Gamma_{p;\alpha,\beta}$ defined in terms of $2(p+1)$ generators $a_k, b_j$ (where the indices $k, j$ run through sets $K,J \subseteq \mathbb{Z}/(p^2-1) \mathbb{Z}$ defined depending on $\alpha,\beta$) and relations of length $4$. The set of relations is described by explicit algebraic equations in the field $\mathbb{F}_{p^2}$. In \cite{RSV} these groups were realized by mapping the generators $a_k$, $b_j$ to explicit generalized quaternions, leading ultimately to an explicit injective group homomorphism 
\begin{equation} \label{eqn:PsiIntro}
    \Psi \colon \Gamma_{p;\alpha,\beta}\to \GL_3(\mathbb{F}_p[[t]])\,.
\end{equation}
In other words, every element of $\Gamma_{p;\alpha,\beta}$ is sent to an invertible $3 \times 3$-matrix whose entries are power series in some formal variable $t$, whose coefficients live in the finite field $\mathbb{F}_p$ with $p$ elements. This is made explicit for $p=3$ in \cref{sec:3_cayley_expanders}, but could also be made explicit for any $p \geq 5$. Since the applications do not depend on concrete matrices, we merely state its existence.

To construct the finite $p$-groups $\Gamma_i$, consider the group homomorphism
\[
\pi_i : \GL_3(\mathbb{F}_p[[t]]) \to  \GL_3(\mathbb{F}_p[t]/(t^{i+1}))
\]
taking a matrix with power series entries and truncating the power series after the term of order $t^i$. Then the group $\GL_3(\mathbb{F}_p[t]/(t^{i+1}))$ is finite, and we define $\Gamma_i$ to be the image of the group $\Gamma_{p;\alpha,\beta}$ under the composition $\pi_i \circ \Psi$. These groups $\Gamma_i$ are easily shown to be $p$-groups and they are what is called congruence quotients (by construction).
The generators $a_k, b_j$ from the construction of $\Gamma_{p;\alpha,\beta}$ map to symmetric generating sets $T_i$ of $\Gamma_i$, i.e., to the set of cosets $a_kN_i, b_jN_i$ when $\Gamma_i = \Gamma_{p;\alpha,\beta}/N_i$ is considered as a factor group.
Using results from \cite{RSV}, we know that the Cayley graphs $G_i = \mathcal{C}(\Gamma_i,T_i)$ associated to the congruence quotient groups $\Gamma_i$ with respect to the generating sets $T_i$ are expanders. 
This argument is worked out in \cite{RSV} by Rungtanapirom and two of the authors, and it is based on a similar approach in the classical papers by Lubotzky, Phillips and Sarnak  \cite{LPS} and by Morgenstern \cite{Morgenstern}.
We note here that the results of \cite{RSV} ultimately rely on deep number theoretic results, namely a translation of the spectrum of the adjacency operator into Satake parameters of an associated automorphic representation and most crucially on work of Drinfeld on the geometric Langlands programme for $\GL_2$.

At this point we have proven \cref{thm:main_expanders_intro}
for the particular valency $2d = 2(p+1)$.
In order to obtain the more general valencies stated in the theorem, we recall in \cref{sec:3_change_of_gens} that a uniformly controlled change of the generating sets $T_i$ of the groups $\Gamma_i$ (the generators must be mutually expressible in words of uniformly bounded length) preserves the expander property. This change of generating set is best performed by finding a smaller generating set for the underlying infinite group $\Gamma_{p;\alpha,\beta}$. This is done in \cref{prop:generateWithBandoneA} reducing to $d = (p+3)/2$ for all $p \geq 3$. The reduction is based on the explicit form of the relations and a combinatorial group theoretic result from \cite{SV} on the local permutation structure of the underlying geometric square complex. To improve even further for $p=3$ we consider in \cref{sec:3_cayley_expanders} a concrete presentation of $\Gamma_{3;0,1}$ which in \cref{thm:good_expanders} is shown to reduce to $2$ generators. For $p=5$, the explicit \cref{example:p5actionOnT6xT6generators3} achieves a reduction to $2$ generators for $\Gamma_{5;0,2}$. It follows again from the theory recalled in \cref{sec:3_change_of_gens} that by adding generators (as necessary) we obtain \cref{thm:main_expanders_intro} for all $d$'s in the range that the theorem promises.

\noindent While this is not needed for the purposes of our hardness results, all of the constructions above are explicit, certainly in the weak sense that for a fixed $p$, the sequence of graphs $G_i$ from \cref{thm:main_expanders_intro} is computable. We also would like to emphasize again, that the expanders constructed for the proof of \cref{thm:main_expanders_intro} consist of vertex transitive graphs, of prime power number of vertices, with a fairly low bound on the degree. All of this is made possible by working with very specific generalized quaternion groups in positive characteristic. 

\paragraph*{Analysis of the Indicators}

Having established the existence of the low-degree Cayley graph expanders, we turn back to the analysis of the indicator 
\begin{equation} \label{eqn:indicatorreminder}
    a(\Phi,H)=\sum_{\sigma \in \mathcal{L}(\Phi,H)} \prod_{v\in V(H)} (-1)^{|\sigma_v|-1} (|\sigma_v|-1)!\,.
\end{equation}
Recall that we claimed the analysis of $a(\Phi,H)$ to be easier for vertex-transitive graphs. Let us now elaborate on this claim. First of all, we restate the formal definition of Cayley graphs for readers who skipped the explicit construction of our expanders: the {\itshape Cayley graph}  of a group $\Gamma$ together with a symmetric generating set\footnote{This means a subset $S \subseteq \Gamma$ of the group that generates this group and satisfies $S^{-1} = S$.} $S \subseteq \Gamma$ is the graph $G = \mathcal{C}(\Gamma,S)$ with vertex set $V(G) = \Gamma$ and edge set 
\[
E(G) = \{(x,xs) \in V(G) \times V(G) ; x \in \Gamma, \ s \in S\}.
\]
Since $S$ is symmetric, with any edge $(x,xs)$ the Cayley graph also contains the edge with opposite orientation $(xs,x) = (xs,(xs)s^{-1})$. Hence we consider Cayley graphs as the underlying unoriented graph. 

Given a Cayley graph $G$ as above, the group $\Gamma$ acts on the graph by letting $g \in \Gamma$ send the vertex $v \in V(G) = \Gamma$ to $gv$. 
 This action extends to the set of fractures $\mathcal{L}(\Phi,H)$ and since the terms $\prod_{v\in V(H)} (-1)^{|\sigma_v|-1} (|\sigma_v|-1)!$ in the formula \eqref{eqn:indicatorreminder} are shown to be invariant under this action, the group $\Gamma$ naturally permutes these summands.
Since our Cayley graph expanders $G_i$ arise from $p$-groups $\Gamma_i$, it follows that when evaluating the indicator $a(\Phi,G_i)$ \emph{modulo $p$},
only those contributions from fractures fixed under $\Gamma_i$ survive.
Now recall that~$\sigma_v$ is a partition of the edges incident to~$v$. The fixed-point fractures~$\sigma$ will satisfy that all~$\sigma_v$ are equal if we identify the edges incident to $v$ with the elements of the generating set. Since, for fixed $p$, our Cayley graph expanders have constant degree, we can thus prove the indicator to be non-zero modulo $p$ by considering just a constant number of fractures. This approach was first used in~\cite{RothSW20unpub}, and we will show that it becomes significantly more powerful if applied to our novel Cayley graph expanders.

Now let us illustrate and sketch this approach for the property $\Phi$ of being bipartite, that is, for proving Theorem~\ref{thm:bipartite-hard}. 
By Theorem~\ref{thm:main_expanders_intro}, there is a family $\mathcal{G}$ of $5$-group Cayley graph expanders of degree $6$. For graphs $G\in \mathcal{G}$, we can show that the indicator $a(\Phi,G)$ does not vanish, given that $\Phi$ is the property of being bipartite. Theorem~\ref{thm:bipartite-hard} will then follow by the argument outlined above; the detailed and formal proof is presented in Section~\ref{sec:bipartite_hardness}.

In the first step, given a graph $G=\mathcal{C}(\Gamma_i,S_i) \in \mathcal{G}$ we need to establish which fixed-point fractures $\sigma$ are contained in $\mathcal{L}(\Phi,G)$, that is, for which $\sigma$ the fractured graph $\fracture{G}{\sigma}$ is bipartite. Since fixed-point fractures $\sigma=(\sigma_v)_{v\in V(G)}$ of $G$ satisfy that all $\sigma_v$ correspond to one particular partition of $S_i$, we will ease notation and identify $\sigma$ with this partition. Using that $S_i$ is a symmetric set of generators of cardinality $6$, that is, $S_i=\{g_1,g_2,g_3,g_1^{-1},g_2^{-1},g_3^{-1}\}$, we define a graph $\basegraph{\sigma}$ as follows: 

\noindent It has a vertex $w^B$ for each block $B$ of $\sigma$, and its set of (multi)edges is given by
\begin{equation} 
E(\basegraph{\sigma}) = \left\{ \{w^{B}, w^{B'} \} : \text{ one multiedge for each } g \in \{g_1,g_2,g_3\} \text{ s.t. } g \in B, g^{-1} \in B' \right\}\,.
\end{equation}
Note that we see $\basegraph{\sigma}$ as a graph with possible loops and possible multiedges. In particular, the graph $\basegraph{\sigma}$ has precisely $3$ edges.

The important property of $\basegraph{\sigma}$, which we will prove in Section~\ref{sec:bipartite_hardness}, is that $\basegraph{\sigma}$ is bipartite if and only if $\fracture{G}{\sigma}$ is bipartite. Consequently, for $\Phi$ being the property of being bipartite, the indicator $a(\Phi,G)$ is given by the following drastically simplified\footnote{A priori, the formula \eqref{eqn:indicatorreminder} leads to the version of formula \eqref{eqn:indicatorsimplifiedintro} with all summands taken to the power $|V(G)|$. However, since the number of vertices is a power of $p=5$, by Fermat's little theorem it can be omitted.} expression, if considered modulo $5$:
\begin{equation} \label{eqn:indicatorsimplifiedintro}
    a(\Phi,G)\equiv \sum_{\sigma: \basegraph{\sigma}\text{ is bipartite}} (-1)^{|\sigma|-1}\cdot (|\sigma|-1)! \mod 5\,,
\end{equation}
where the sum is over partitions $\sigma$ of $S_i$. We provide the evaluation of the above expression in Table~\ref{tab:bipartbasegraph_intro} and observe that the result is $-16 \neq 0 \mod 5$. Since this argument applies to all members of the family of $5$-group Cayley graph expanders, we conclude that the indicator is non-zero infinitely often, which ultimately proves Theorem~\ref{thm:bipartite-hard}.

\begin{table}[htb]
    \centering
    \begin{tabular}{cccc}
        $\basegraph{\sigma}$ & No. of $\sigma$ & $|\sigma|$ &  $(-1)^{|\sigma|-1}\cdot (|\sigma|-1)!$ \\ \hline 
         \begin{tikzpicture}[main node/.style={circle,draw,font=\Large,scale=0.5}]
         \draw[white] (-0.2,0.2)  rectangle (1.2, 1);
         \node[main node] (A) at (0,0) {};
         \node[main node] (B) at (1,0) {};
         \node[main node] (C) at (0,0.3) {};
         \node[main node] (D) at (1,0.3) {};
         \node[main node] (E) at (0,0.6) {};
         \node[main node] (F) at (1,0.6) {};         
         \draw (A) -- (B);
         \draw (C) -- (D);
         \draw (E) -- (F);
         \end{tikzpicture} & 1 & 6 & $-120\cdot 1$\\
         \begin{tikzpicture}[main node/.style={circle,draw,font=\Large,scale=0.5}]
         \draw[white] (-0.2,0.2)  rectangle (1.2, 1);
         \node[main node] (A) at (0,0) {};
         \node[main node] (B) at (1,0) {};
         \node[main node] (C) at (0,0.3) {};
         \node[main node] (D) at (1,0.3) {};
         \node[main node] (E) at (2,0) {};
         \draw (A) -- (B) -- (E);
         \draw (C) -- (D);
         \end{tikzpicture} & 12 & 5 & $24\cdot 12$\\
         \begin{tikzpicture}[main node/.style={circle,draw,font=\Large,scale=0.5}]
         \draw[white] (-0.2,0.2)  rectangle (1.2, 1);
         \node[main node] (A) at (0,0) {};
         \node[main node] (B) at (1,0) {};
         \node[main node] (C) at (3,0) {};
         \node[main node] (E) at (2,0) {};
         \draw (A) -- (B) -- (E) -- (C);
         \end{tikzpicture} & 24 & 4 & $- 6\cdot 24$\\    
         \begin{tikzpicture}[main node/.style={circle,draw,font=\Large,scale=0.5}]
         \draw[white] (-0.2,0.2)  rectangle (1.2, 1);
         \node[main node] (A) at (0,0) {};
         \node[main node] (B) at (1,0) {};
         \node[main node] (C) at (1,0.3) {};
         \node[main node] (D) at (1,-0.3) {};
         \draw (A) -- (B);
         \draw (A) -- (C);
         \draw (A) -- (D);
         \end{tikzpicture} & 8 & 4 & $-6\cdot 8$\\    
         \begin{tikzpicture}[main node/.style={circle,draw,font=\Large,scale=0.5}]
         \draw[white] (-0.2,0.2)  rectangle (1.2, 1);
         \node[main node] (A) at (0,0) {};
         \node[main node] (B) at (1,0) {};
         \node[main node] (C) at (2,0) {};
         \node[main node] (D) at (3,0) {};
         \draw (A) to[bend left] (B);
         \draw (A) to[bend right] (B);
         \draw (C) -- (D);
         \end{tikzpicture} & 6 & 4 & $-6\cdot 6$\\     
         \begin{tikzpicture}[main node/.style={circle,draw,font=\Large,scale=0.5}]
         \draw[white] (-0.2,0.2)  rectangle (1.2, 1);
         \node[main node] (A) at (0,0) {};
         \node[main node] (B) at (1,0) {};
         \node[main node] (C) at (2,0) {};
         \draw (A) to[bend left] (B);
         \draw (A) to[bend right] (B);
         \draw (B) -- (C);
         \end{tikzpicture} & 24 & 3 & $2\cdot 24$\\       
         \begin{tikzpicture}[main node/.style={circle,draw,font=\Large,scale=0.5}]
         \draw[white] (-0.5,0.2)  rectangle (1.2, 1);
         \node[main node] (A) at (0,0) {};
         \node[main node] (B) at (1,0) {};
         \draw (A) to[bend left] (B);
         \draw (A) to (B);
         \draw (A) to[bend right] (B);
         \end{tikzpicture} & 4 & 2 & $-1\cdot 4$\\\hline
         Total contribution &  &  & $-16 \equiv 4 \mod 5$\\
         ~
    \end{tabular}
    \caption{List of bipartite graphs $\basegraph{\sigma}$ for $3$ generators; here we give the isomorphism class of $\basegraph{\sigma}$, the number of partitions $\sigma$ with the corresponding isomorphism class, the number of blocks of sigma and the total contribution to $a(\Phi,G) \mod 5$ for the property $\Phi$ of being bipartite.}
    \label{tab:bipartbasegraph_intro}
\end{table}

The proof of Theorem~\ref{thm:minclose-hard}, presented in Section~\ref{sec:minorclosed}, will follow a comparable technique, but it will require multiple, more involved cases.

\paragraph*{Extension to Modular Counting}
Our understanding of the cancellation behaviour of the indicators $a(\Phi,H)$ modulo $p$ does not only allow us to analyse the complexity of exact counting of $k$-edge subgraphs satisfying~$\Phi$, but also extends to counting $k$-edge subgraphs modulo $p$. 

To this end, we provide the necessary set-up for parameterized modular counting. In particular, the basis for our intractability results on the modular counting versions is given by the \emph{decision} problem of detecting so-called colour-prescribed homomorphisms, which is known to be hard for graphs of high treewidth due to Marx~\cite{Marx10}. Using a version of the Schwartz-Zippel-Lemma due to Williams et al.\ \cite{WilliamsWWY15}, we are able to reduce an instance $I$ of this problem to an instance $I'$, such that, with high probability, $I'$ has precisely one solution if $I$ has at least one solution, and $I'$ has no solutions if $I$ has no solutions. In a second step, the obtained instance $I'$ can then easily be reduced to the modular counting version for each prime $p$. Since the first step of this reduction is randomised, we need to assume the randomised Exponential Time Hypothesis for our fine-grained lower bounds.

Afterwards, we prove a variant of the Complexity Monotonicity principle for modular counting in the case of colour-prescribed homomorphisms. As a consequence, our hardness results for modular subgraph counting problems, including Theorem~\ref{thm:forest_matroids}, can be proved following the same strategy as outlined in the previous subsection.

Moreover, instead of only presenting intractability results, we investigate whether the expression as a linear combination of homomorphism counts can also be used to achieve improved algorithms for modular subgraph counting problems. And indeed, considering the linear combination~\eqref{eq:intro_hom_lincomb} modulo $2$ for the property $\Phi$ of being a path or a cycle, allows us to prove that each graph $H$ with degree at least $5$ vanishes in the linear combination, that is, $a_{\Phi,k}(H)=0$. More precisely, we will prove this for a version of the problem in which two vertices of the $k$-paths or the $k$-cycles are already fixed. This must be done to avoid automorphisms of even order, which turns out to be necessary for~\eqref{eq:intro_hom_lincomb} to be well-defined modulo $2$, since some of the coefficients $a_{\Phi,k}(H)=0$ are of the form $\#\auts{H}^{-1}$. 

\noindent The algorithms for counting $k$-paths and $k$-cycles modulo $2$ turn then out to be very simple: Essentially, we will see that it suffices to guess the two fixed vertices, and thereafter the algorithm evaluates Equation~\eqref{eq:intro_hom_lincomb} modulo $2$, by computing each non-vanishing term using a standard treewidth-based dynamic programming algorithm for counting homomorphisms. Since each graph $H$ whose coefficient survives modulo $2$ has degree at most $4$, we can rely on known results on the treewidth of bounded degree graphs~\cite{FominGSS09}. Ultimately, this allows us to prove Theorem~\ref{thm:main_cycles_mod_2}.

Finally, our classification for counting homomorphisms modulo $p$ builds upon the well-established algorithms and reduction sequences used both in the classification for the decision problem~\cite{Grohe07}, as well as in the classification for the exact counting problem~\cite{DalmauJ04}. However, the difficulty in proving our classification for counting modulo $p$ is due to graphs $H$ which have high treewidth but admit automorphisms of order $p$. For those graphs, we can neither rely on an algorithm for exact counting, nor does the known hardness proof transfer.

We solve this problem by considering the $p$-reduced quotients. Let us denote the function that maps a graph $G$ to the number of homomorphisms from $H$ to $G$, modulo $p$, by $\#_p\homs{H}{\star}$. We show that for each graph $H$ we have
$\#_p\homs{H}{\star} = \#_p\homs{H^\ast_p}{\star}$.
As a consequence, it suffices to consider the $p$-reduced quotients for our classification. Since, by definition, those graphs to not admit an automorphism of order $p$, we are able to show that the known methods for proving classifications for homomorphism problems apply.

\noindent Let us conclude by pointing out that, while proving that the $p$-reduced quotient is uniquely defined up to isomorphism, we also establish a modular variant of \lovasz' criterion for graph isomorphism via homomorphism counts~(see Chapter~5 in~\cite{Lovasz12}):
\begin{restatable}{lemma}{homslovasz}\label{lem:homs_lovasz}
Let $H$ and $H'$ be graphs, neither of which has an automorphism of order $p$. Suppose that for all graphs $G$ we have that
\[\#_p\homs{H}{G} = \#_p\homs{H'}{G} \,.\]
Then $H$ and $H'$ are isomorphic.
\end{restatable}

\subsection*{Conclusion and Open Questions}
All of our hardness results for modular subgraph counting problems only apply to primes $p\geq 3$ and have, using the randomised Exponential Time Hypothesis as a slightly stronger assumption, the same complexity as their counterparts from exact counting. However, for $p=2$ the complexity landscape seems different: We obtained an improvement for counting $k$-cycles and $k$-paths modulo $2$. Moreover, there are known instances of the counting version of the parameterized subgraph isomorphism problem, such as counting $k$-matchings, where exact counting, as well as counting modulo $p$ for each prime $p\geq 3$ is fixed-parameter \textbf{in}tractable, while the computation becomes fixed-parameter tractable if done modulo $2$~\cite{Curticapean13,CurticapeanDH21}.

Since, additionally, many of our hardness proofs do not apply to the case of counting modulo $2$, we propose a thorough investigation of the complexity of the parameterized subgraph counting problem modulo $2$ as the next step in this line of research. As a starting point, we suggest the problem of counting bipartite $k$-edge subgraphs modulo $2$: While our proofs extend to counting such subgraphs modulo $p$ for some primes $p\geq 3$, a computer-aided search revealed that, for $p=2$, our approach \emph{cannot} work for any family of Cayley graph expanders of degree at most $12$; details are provided in Remark~\ref{rem:sage_results}. Indeed, we conjecture that our methods for proving intractability can be used to show that the problem is intractable for each prime $p>2$, but not for $p=2$, which leads to the question of whether this problem might be fixed-parameter tractable.

There are also interesting open questions concerning $p$-group Cayley graph expanders with low degree. To describe them, fix some prime $p$ and consider the set $D(p) \subseteq \mathbb{Z}_{\geq 0}$ of integers $d$ such that there exists a sequence of finite $p$-groups $\Gamma_i$ of orders that tend to infinity, with symmetric generating sets $S_i$ of cardinality $2d$ such that the Cayley graphs $\mathcal{C}(\Gamma_i, S_i)$ form a family of expanders. 
With any $d \in D(p)$ actually any $d' \ge d$ also lies in $D(p)$, because we can find a uniform bound on the length of a word in $S_i$ to produce a new additional generator for $\Gamma_i$, showing $d+1 \in D(p)$. So the ultimate question is the following:
%Then we can ask:
\begin{question}
What is the behaviour of the function $p \mapsto d(p) = \min D(p)$?
\end{question}
Since $2$-regular graphs are never expanders, we know that $d(p) \geq 2$ for all primes $p$. Moreover, combining the construction of \cite{PeyerimhoffV11} (for $p=2$) with \cref{thm:main_expanders_intro}, \cref{prop:generateWithBandoneA}, and some further examples that we computed, we obtain the following values and bounds for the function $d$:
\begin{center}
\begin{tabular}{c||c|c|c}
    p & $p \in \{2, 3, 5, 7, 11, 13\}$ & $17 \le p \le 83$ & $89 \le p$ \\ \hline
    d(p) & 2 & $d(p) \in \{2,3\}$ & $2 \leq d(p) \leq (p+3)/2$
\end{tabular}
\end{center}

Based on this experimental evidence we make the following conjecture:
\begin{conjecture}
For every $p \geq 3$ there is a group among the $\Gamma_{p,\alpha,\beta}$ that is $3$-generated. In particular, there are $p$-group Cayley graph expanders of fixed valency $2d$ for all $p \geq 3$ and all $d \geq 3$.
\end{conjecture}
If the conjecture is satisfied, the function $d$ above would be uniformly bounded from above by $3$.

\newpage

\tableofcontents

\newpage

\section{Preliminaries}\label{sec:prelims}
Given a function $f:X\times Y \rightarrow Z$ and an element $x\in X$, we write $f(x,\star): Y \rightarrow Z$ for the function $y\mapsto f(x,y)$. Furthermore, given a finite set $S$, we write $|S|$ and $\#S$ for the cardinality of $S$, and given a prime $p\geq 2$, we set $\#_pS:=\# S ~\mathsf{mod}~ p$. Similarly, for a function $f: X \rightarrow \mathbb{N}$, we write $\#_pf : X \rightarrow \mathbb{F}_p$ for the function that maps $x\in X$ to $f(x) ~\mathsf{mod}~ p$.

\subsection{Graphs and Homomorphisms}
We consider \emph{undirected graphs} $G = (V(G),E(G))$ with vertex set $V(G)$ and set of edges $E(G)$. Our graphs are without self-loops and multiple edges, unless stated otherwise. 
For a positive integer $k$, we write $C_k$, $K_k$ and $M_k$ for the $k$-cycle, the $k$-clique, and the $k$-matching, respectively.

A \emph{(graph) homomorphism} from $F$ to $G$ is an edge-preserving mapping $\varphi:V(F)\rightarrow V(G)$, and we write $\homs{F}{G}$ for the set of all homomorphisms from $F$ to $G$. A \emph{(graph) isomorphism} from $F$ to $G$ is a bijection $\pi: V(F)\rightarrow V(G)$ such that for all pairs of vertices $u,v\in V(F)$ we have $\{u,v\}\in E(F)$ if and only if $\{\pi(u),\pi(v)\}\in E(G)$. We say that $F$ and $G$ are \emph{isomorphic}, denoted by $F\cong G$, if an isomorphism from $F$ to $G$ exists. An isomorphism from $F$ to itself is called an \emph{automorphism}, and we write $\auts{F}$ for the set of all automorphisms of $F$.

Given a graph $G$ and a subset of edges $A\subseteq E(G)$, we write $G[A]$ for the graph obtained from $(V(G),A)$ by deleting all isolated vertices. The graph $G[A]$ is naturally a subgraph of the graph~$G$.

Given a graph $H$ and a partition $\rho$ of $V(H)$, the \emph{quotient graph} $H/\rho$ has as vertices the blocks of $\rho$, and two blocks $B$ and $\hat{B}$ are made adjacent if and only if there are vertices $u\in B$ and $v\in \hat{B}$ such that $\{u,v\}\in E(H)$. We emphasize that quotient graphs may contain self-loops. There is a natural surjective map of graphs $H \to H/\rho$ that determines $\rho$.

We say that a graph $F$ is a \emph{minor} of a graph $H$, if $F$ can be obtained from $H$ by a sequence of vertex and edge deletions, and edge contractions; self-loops and multiple edges are deleted after the latter.

\paragraph*{Coloured Graphs}
Let $H$ be a graph.
An $H$\emph{-coloured graph} is a pair of a graph $G$ and a homomorphism $c: G \to H$,
%$c$ from $G$ to $H$, 
called the $H$\emph{-colouring}. To avoid notational clutter, we might say that a graph $G$ is $H$-coloured if the $H$-colouring is implicit or clear from the context. 

A homomorphism $\varphi$ from a graph $H$ to an $H$-coloured graph $G$ with colouring $c$ is called \emph{colour-prescribed} if for each $v\in V(F)$ we have $c(\varphi(v))=v$. We write $\cphoms{H}{G}$ for the set of all colour-prescribed homomorphisms from $H$ to $G$.

\paragraph*{Expander Graphs}
Let $d>0$ be an integer and let $c>0$ be a rational. Let furthermore $\mathcal{G}$ be an infinite class of graphs $\{G_1,G_2,\dots\}$ with $|V(G_i)|=n_i$. The class $\mathcal{G}$ is called a \emph{family of} $(n_i,d,c)$\emph{-expanders} if for all $i\geq 0$, the graph $G_i$ is $d$-regular and furthermore satisfies
\[\forall X \subseteq V(G_i): |S(X)| \geq c\left(1-\frac{|X|}{|V(G_i)|} \right)|X| \,. \]
Here $S(X)$ denotes the set of all vertices in $V(G_i)\setminus X$ that are adjacent to a vertex in $X$.

In Section~\ref{sec:3_change_of_gens} we recall a spectral reformulation of the expander property in terms of the non-trivial eigenvalues of the graph Laplace operator on $\mathbb{R}$-valued functions on the set of vertices.

\subsection{Parameterized and Fine-grained Complexity Theory}
We will follow the notation of the textbook of Flum and Grohe~\cite{FlumG06}. A \emph{parameterized (counting) problem} is a pair of a (counting) problem $P$ and a polynomial-time computable\footnote{In some literature, the parameterization is not enforced to be polynomial-time; we refer the reader to the discussion in~\cite[Section~1.2]{FlumG06} on that matter and point out that all parameterizations in this work are polynomial-time computable.} \emph{parameterization} $\kappa: \{0,1\}^\ast \rightarrow \mathbb{N}$. 

An algorithm $\mathbb{A}$ is called a \emph{fixed-parameter tractable algorithm} (``\emph{fpt-algorithm}'') with respect to a parameterization $\kappa$ if there is a computable function $f$ such that, on input $x$, the running time of $\mathbb{A}$ is bounded by $f(\kappa(x))\cdot \mathsf{poly}(|x|)$. A parameterized (counting) problem with parameterization $\kappa$ is called \emph{fixed-parameter tractable} if it can be solved by an fpt-algorithm with respect to $\kappa$.

Given two parameterized (counting) problems $(P,\kappa)$ and $(\hat{P},\hat{\kappa})$, a \emph{parameterized Turing-reduction} from $(P,\kappa)$ to $(\hat{P},\hat{\kappa})$ is an $\mathbb{A}$ equipped with oracle access to $\hat{P}$ that satisfies the following constraints:
\begin{enumerate}
    \item $\mathbb{A}$ is an fpt-algorithm with respect to $\kappa$.
    \item There is a computable function $g$ such that, on input $x$, every oracle query $y$ satisfies $\hat{\kappa}(y)\leq g(\kappa(x))$.
\end{enumerate}
We write $(P,\kappa) \fptred (\hat{P},\hat{\kappa})$ if a parameterized Turing-reduction exists.

We write $\textsc{Clique}$ for the problem that is given as input a graph $G$ and a positive integer $k$ and is expected to decide whether $G$ contains a clique of size $k$; the parameterization is given by $k$, that is, $\kappa(G,k):=k$. Similarly, we write $\#\textsc{Clique}$ for the problem of counting cliques of size $k$. It is known that neither $\textsc{Clique}$ nor $\#\textsc{Clique}$ are fixed-parameter tractable (not even solvable in time $f(k)\cdot |V(G)|^{o(k)}$ for any function $f$) unless the Exponential Time Hypothesis fails~\cite{Chenetal05,Chenetal06}:

\begin{conjecture}[Exponential Time Hypothesis~\cite{ImpagliazzoP01}]
The \emph{Exponential Time Hypothesis (ETH)} asserts that $3$\textsc{-SAT} cannot be solved in time $\exp(o(n))$, where $n$ is the number of variables of the input formula.\lipicsEnd
\end{conjecture}

A parameterized (counting) problem is $\W{1}$-hard (resp.\ $\#\W{1}$-hard) if it can be reduced from $\textsc{Clique}$ (resp.\ $\#\textsc{Clique}$) via parameterized Turing-reductions. In particular, $(\#)\W{1}$-hard problems are not fixed-parameter tractable unless ETH fails.

The following four problems are the main objects of study in this work:
\begin{itemize}
    \item Let $\mathcal{H}$ be a class of graphs. The problem $\#\homsprob(\mathcal{H})$ asks, on input a graph $H\in \mathcal{H}$ and a graph $G$, to compute the number of homomorphisms from $H$ to $G$; the parameterization is given by $|H|$.
    \item The problem $\#\cphomsprob(\mathcal{H})$ is the colour-prescribed version of $\#\homsprob(\mathcal{H})$: On input a graph $H\in \mathcal{H}$ and an $H$-coloured graph $G$, the goal is to compute the number of colour-prescribed homomorphisms from $H$ to $G$.
    \item Let $\Phi$ be a graph property. The problem $\#\edgesubsprob(\Phi)$ asks, on input a graph $G$ and a positive integer $k$, to compute the number of $k$-edge subgraphs of $G$ that satisfy $\Phi$, that is, the number of edge subsets $A\in E(G)$ of size $k$ such that $\Phi(G[A])=1$. The parameter is $k$.
    \item The problem $\#\coledgesubsprob(\Phi)$ is the edge-colourful version of $\#\edgesubsprob(\Phi)$: Here, the edges of the graph $G$ are coloured with $k$ distinct colours and the goal is to count only the $k$-edge subsets satisfying $\Phi$ that contain each colour precisely once.
\end{itemize}
The corresponding decision versions are defined similarly and drop the ``$\#$'' in their notation.

\paragraph*{Parameterized Modular Counting}
For what follows, let $p\geq 2$ be a fixed prime. We denote the modular counting versions of the aforementioned problems by using the symbol $\#_p$; for example, $\#_p\textsc{Clique}$ denotes the problem of counting $k$-cliques modulo $p$.
The notion for intractability of parameterized counting modulo $p$ is given by $\pW{1}$\emph{-hardness}, which we define via parameterized Turing-reductions from $\#_p\textsc{Clique}$.\footnote{We note that in~\cite{BjorklundDH15,CurticapeanDH21}, hardness for $\pW{1}$ (also denoted by $\oplus\W{1}$ for $p=2$) is defined via parameterized parsimonious reductions. However, since some of our reductions rely on the Complexity Monotonicity principle and thus on solving systems of linear equations by posing multiple oracle queries~\cite{CurticapeanDM17,RothSW20unpub}, we are required to use parameterized Turing-reductions instead.}

For the purpose of this paper, all hardness results for modular counting problems will be obtained by reducing from the problem $\#_p\cphomsprob(\mathcal{H})$ of counting colour-prescribed homomorphisms from graphs in~$\mathcal{H}$, modulo $p$.

For the formal statement we rely on the randomised version of ETH (cf.\ \cite{Delletal14}):
\begin{conjecture}[Randomised Exponential Time Hypothesis]
The \emph{randomised Exponential Time Hypothesis (rETH)} asserts that $3$\textsc{-SAT} cannot be solved in time $\exp(o(n))$ by a randomised algorithm with error probability at most $1/3$, where $n$ is the number of variables of the input formula.\lipicsEnd
\end{conjecture}

Using a result of Marx~\cite{Marx10} on the (decision) problem $\cphomsprob(\mathcal{H})$, we can invoke a version of the Schwartz-Zippel-Lemma due to Williams et al.\ \cite{WilliamsWWY15} to obtain the following almost tight lower bound for $\#_p\cphomsprob(\mathcal{H})$. We point out that the construction is identical to the argument in~\cite{CurticapeanDH21} on counting vertex-colourful $k$-edge subgraphs, and we provide the proof only for the sake of completeness.  
\begin{lemma}\label{lem:hardness_basis_mod}
Let $\mathcal{H}$ be a recursively enumerable class of graphs. If the treewidth of $\mathcal{H}$ is unbounded then $\#_p\cphomsprob(\mathcal{H})$ is $\pW{1}$-hard and cannot be solved in time
\[f(|H|)\cdot |G|^{o(\mathsf{tw}(H)/\log \mathsf{tw}(H))}\,, \]
for any function $f$, unless rETH fails.
\end{lemma}
\begin{proof}
We begin with the conditional lower bound under rETH. Marx~\cite{Marx10} has shown that, unless ETH fails, the decision version $\cphomsprob(\mathcal{H})$\footnote{To be precise, Marx stated the result for the related problem $\textsc{PartitionedSub}(\mathcal{H})$, which is however, equivalent to $\cphomsprob(\mathcal{H})$ as shown e.g.\ in~\cite[Section 2]{RothSW20}.} cannot be solved in time $f(|H|)\cdot |G|^{o(\mathsf{tw}(H)/\log \mathsf{tw}(H))}$. We construct a (randomised) reduction to the modular counting version using a Schwartz-Zippel-Lemma due to Williams et al.\ \cite[Lemma~2.1]{WilliamsWWY15} which reads as follows: Let $m\geq 2$ be an integer and let $P(x_1,\dots,x_n)$ be a non-zero multilinear degree $d$ polynomial over the ring $\mathbb{Z}/m\mathbb{Z}$. Then
\begin{equation}\label{eq:sw_lemma}
    \Pr_{(a_1,\dots,a_n)\in \{0,1\}^n}[P(a_1,\dots,a_n)\neq 0] \geq 2^{-d}\,.
\end{equation}
Let $H$ and $G$ (with $H$-colouring $c$) be an instance of $\cphomsprob(\mathcal{H})$. Let furthermore $k=|V(H)|$ and assume that $V(G)=\{1,\dots,n\}$. Consider the following polynomial over $\mathbb{F}_p$:
\[P(x_1,\dots,x_n):= \sum_{\varphi \in \cphoms{H}{G}} ~\prod_{v\in V(H)} x_{\varphi(v)} \,.\]
Observe that $P$ is multilinear since $\varphi$ must be injective due to being colour-prescribed. Note further that, given $(a_1,\dots,a_n)\in \{0,1\}^n$, the evaluation $P(a_1,\dots,a_n)$ is equal to $\#\cphoms{H}{G'}$, where $G'$ is the graph obtained from $G$ by deleting vertex $i$ if and only if $a_i=0$, and which is coloured by the restriction $c'$ of $c$ to $V(G')$; it is clear that $c'$ is still an $H$-colouring since $G'$ is a subgraph of $G$.

Invoking (\ref{eq:sw_lemma}), we obtain the following reduction; let us assume that we are given oracle access to $\#_p\cphomsprob(\mathcal{H})$: Just delete every vertex of $G$ uniformly at random to obtain $G'$ and keep the restriction of the $H$-colouring to the remaining vertices. If no colour-prescribed homomorphism from $H$ to $G$ exists then there is none from $H$ to $G'$ as well. Otherwise, $P$ is not a zero-polynomial and thus, with probability at least $2^{-d}$, we have that $\#\cphoms{H}{G'}\neq 0 \mod p$. Since the degree $d$ is bounded by $k=|V(H)|$, we obtain, via standard probability amplification, a constant positive success probability by (independently) repeating the experiment $O(2^{|V(H)|})$ many times, and querying the oracle for each pair $H$ and $G'$. If at least one trial yields $\#\cphoms{H}{G'}\neq 0 \mod p$, we report that $\cphoms{H}{G}\neq \emptyset$. Otherwise, we report that $\cphoms{H}{G}= \emptyset$. Since no oracle query modifies $H$, and since $G'$ is always a subgraph of $G$, we conclude that any algorithm for $\#_p\cphomsprob(\mathcal{H})$ running in time $f(|H|)\cdot |G|^{o(\mathsf{tw}(H)/\log \mathsf{tw}(H))}$, yields a randomised algorithm with one-sided and constant error-probability\footnote{We can even obtain an exponentially small (in $|G|$) error probability by repeating for $O(2^{|V(H)|}) \cdot |G|^{O(1)}$ many trials.} for $\cphomsprob(\mathcal{H})$ running in time $O(2^{|V(H)|}) \cdot f(|H|)\cdot |G|^{o(\mathsf{tw}(H)/\log \mathsf{tw}(H))}$, contradicting Marx' lower bound as mentioned above.

Finally, $\pW{1}$-hardness of $\#_p\cphomsprob(\mathcal{H})$ follows from a known parsimonious reduction from counting $k$-cliques to $\#\cphomsprob(\mathcal{H})$~\cite[Section~1.2.2 and Chapter~5]{Curticapean15}.
\end{proof}

\subsection{Fractured Graphs}\label{sec:obstructions}
Following the terminology of \cite{RothSW20unpub}, we define a \emph{fracture} of a graph $H$ to be a tuple $\rho = (\rho_v)_{v\in V(H)}$, where~$\rho_v$ is a partition of the set of edges $E_H(v)$ of $H$ incident to $v$. 

Given a fracture $\rho$ of $H$, the \emph{fractured graph} $\fracture{H}{\rho}$ is obtained from $H$ be splitting each vertex $v\in V(H)$ according to $\rho_v$; an illustration is provided in Figure~\ref{fig:simpler_fractures}. Formally, for each $v\in V(H)$ and block $B$ of $\rho_v$, the graph $\fracture{H}{\rho}$ contains a vertex $v^B$, and for each edge $e=\{u,v\}\in E(H)$, we connect $u^B$ and $v^{B'}$ if (and only if) $e\in B\cap B'$. 

\begin{figure}[t!]
    \centering
    \begin{tikzpicture}[scale=1]
    \node[circle,draw] (m) at (0,0) {\LARGE $v$};
    \draw[ultra thick,red] (m) -- (-2,2);
    \draw[ultra thick,black!30!green] (m) -- (-2,0);
    \draw[ultra thick,blue] (m) -- (-2,-2);
    
    \draw[ultra thick,black!10!yellow] (m) -- (2,2);
    \draw[ultra thick] (m) -- (2,0);
    \draw[ultra thick,brown] (m) -- (2,-2);
    \begin{scope}[shift={(7,0)}]
    \node[circle,draw] (b1) at (0,0) {$B_1$};
    \draw[ultra thick,red] (b1) -- (-2,2);
    \draw[ultra thick,black!30!green] (b1) -- (-2,0);
    \draw[ultra thick,blue] (b1) -- (-2,-2);
    \begin{scope}[shift={(1,0)}]
    \node[circle,draw] (b2) at (0,0) {$B_2$};
    \draw[ultra thick,black!10!yellow] (b2) -- (2,2);
    \draw[ultra thick] (b2) -- (2,0);
    \draw[ultra thick,brown] (b2) -- (2,-2);
    \end{scope}
    \end{scope}
    \end{tikzpicture}
    \caption{\label{fig:simpler_fractures} Illustration of the construction of a fractured graph. The left picture shows a vertex $v$ of a graph~$H$ with incident edges $E_H(v)=\{  \redc{2pt},\greenc{2pt},\bluec{2pt},\yellowc{2pt},\blackc{2pt},\brownc{2pt}\}$. The right picture shows the splitting of $v$ in the construction of the fractured graph~$\fracture{H}{\rho}$ for a fracture $\rho$ satisfying that the partition $\rho_v$ contains two blocks $B_1 =\{ \redc{2pt},\greenc{2pt},\bluec{2pt}\}$, and $B_2=\{\yellowc{2pt},\blackc{2pt},\brownc{2pt}\}$.}
\end{figure}
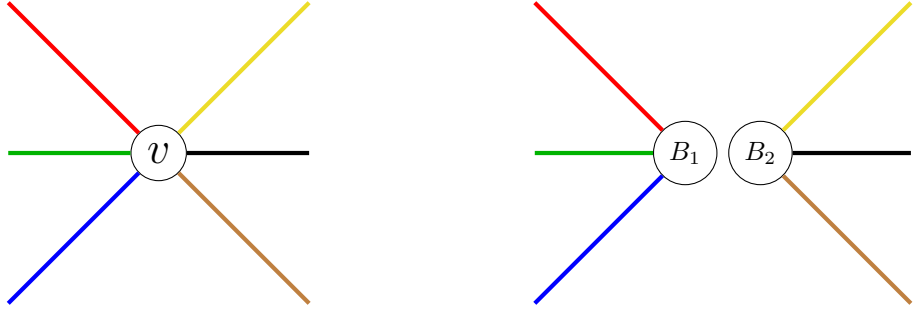

Given a graph property $\Phi$ and a graph $H$, we write $\mathcal{L}(\Phi,H)$ for the set of all fractures $\rho$ of $H$ such that $\fracture{H}{\rho}$ satisfies $\Phi$. Furthermore, the \emph{indicator} $a(\Phi,H)$ of $\Phi$ and $H$ is defined as follows:
\[ a(\Phi,H) := \sum_{\sigma \in \mathcal{L}(\Phi,H)} ~\prod_{v \in V(H)} (-1)^{|\sigma_v|-1}\cdot (|\sigma_v|-1)! \,.\]
Now let $\mathcal{G}=\{G_1,G_2,\dots\}$ be a family of $(n_i,d,c)$-expanders for some positive integers $d$ and~$c$. We call $\mathcal{G}$ an \emph{obstruction} for $\Phi$ if $a(\Phi,G_i)\neq 0$ for infinitely many $i$. Given a prime $p$, we call $\mathcal{G}$ a $p$\emph{-obstruction} for $\Phi$ if $a(\Phi,G_i)\neq 0 \mod p$ for infinitely many $i$.

The following result, which we prove to be an immediate consequence of~\cite{RothSW20unpub}, shows that counting small subgraphs satisfying a property $\Phi$ is intractable whenever $\Phi$ has an obstruction.
\begin{lemma}\label{lem:monotonicity}
Let $\Phi$ be a computable graph property. If $\Phi$ has an obstruction, then $\#\coledgesubsprob(\Phi)$ is $\#\W{1}$-hard and, assuming ETH, cannot be solved in time $f(k)\cdot |G|^{o(k/\log k)}$, for any function $f$.

If $\Phi$ has a $p$-obstruction, then $\#_p\coledgesubsprob(\Phi)$ is $\pW{1}$-hard and, assuming rETH, cannot be solved in time $f(k)\cdot |G|^{o(k/\log k)}$, for any function $f$.
\end{lemma}
\begin{proof}
Assume first that $\Phi$ has an obstruction $\mathcal{G}$. Let $\mathcal{H}[\Phi,\mathcal{G}]$ be the set of all $G_i\in \mathcal{G}$ for which $a(\Phi,G_i)\neq 0$. Since $\mathcal{G}$ is an obstruction, we have that $\mathcal{H}[\Phi,\mathcal{G}]$ is infinite. There is a (tight) reduction from $\#\cphomsprob(\mathcal{H}[\Phi,\mathcal{G}])$ to $\#\coledgesubsprob(\Phi)$ which yields the desired lower bounds~\cite[Lemma~3.8]{RothSW20unpub}.\footnote{While formally in \cite[Lemma 3.8]{RothSW20unpub} the claimed result is stated for the uncoloured version of the counting problem, the proof given in~\cite{RothSW20unpub} establishes the hardness result for the edge-colourful version formulated in \cref{lem:monotonicity}. Here another technical remark is that in \cite{RothSW20unpub} the problem $\#\coledgesubsprob(\Phi)$ was defined in a more narrow fashion than above: the colouring on the input graph $G$ needed to be induced from a graph homomorphism $G \to H$, see \cite[Section 3]{RothSW20unpub}. However, the hardness results for this restricted problem shown in~\cite{RothSW20unpub} then imply the hardness results for the more general problem we consider here. } 

Now if $\Phi$ has a $p$-obstruction $\mathcal{G}$, and if we let $\mathcal{H}[\Phi,\mathcal{G}]$ be the set of all $G_i\in \mathcal{G}$ for which $a(\Phi,G_i)\neq 0 \mod p$, then the same reduction applies to counting modulo $p$, that is, we obtain a (tight) reduction from $\#_p\cphomsprob(\mathcal{H}[\Phi,\mathcal{G}])$ to $\#_p\coledgesubsprob(\Phi)$. More precisely, the reduction only relies on the existence of a unique solution of a system of linear of equations given by the Complexity Monotonicity principle~\cite[Lemma~3.6]{RothSW20unpub}, the corresponding matrix of which is a triangular matrix with $1$'s on the diagonal~\cite[Lemma~3.5]{RothSW20unpub}. Consequently, we obtain a unique solution even if arithmetic is done modulo $p$. The lower bounds for $\#_p\coledgesubsprob(\Phi)$ thus hold by Lemma~\ref{lem:hardness_basis_mod}.
\end{proof}

As elaborated in~\cite{RothSW20unpub}, proving that a family of expanders is an obstruction for a property~$\Phi$ is the hardest step towards intractability of $\#\edgesubsprob(\Phi)$. The Cayley graph expanders constructed in this work will turn out to be obstructions for a wide range of natural properties~$\Phi$.

\section{Cayley Graph Expanders} \label{Sect:CayleyExpanders}

The {\itshape Cayley graph} of a group $\Gamma$ generated by a symmetric set $S \subseteq \Gamma$, i.e., such that $S^{-1} = S$, is the graph $G = \mathcal{C}(\Gamma,S)$ with vertex set $V(G) = \Gamma$ and edge set 
\[
E(G) = \{(x,xs) \in V(G) \times V(G) ; x \in \Gamma, \ s \in S\}.
\]
Since $S$ is symmetric, with any edge $(x,xs)$ the Cayley graph also contains the edge with opposite orientation $(xs,x) = (xs,(xs)s^{-1})$. Hence we consider Cayley graphs as the underlying unoriented graph. 

In this section we recall the construction of certain discrete groups from \cite{SV,RSV} that are lattices in generalized quaternion algebras. Their group theory is controlled by representations with values in power series rings with coefficients in $\mathbb{F}_p$. In particular, these representations lead to well chosen finite congruence quotients that are $p$-groups. With a natural set of generators as in \cite{SV,RSV} these lead to sequences of Cayley graph expanders of valency $2(p+1)$ and vertex set of size a power of $p$, here $p$ is odd. We explain the known fact that a change of generators for the lattice does not destroy the expander property. Moreover, we analyse the relations and find that we can reduce the number of generators to $(p+3)/2$. For $p \ge 5$ (with an extra argument for $p=5$) we deduce the existence of a series of Cayley graph expanders of $p$-power order and now valency $2(p-2)$ that is necessary for the applications in this paper. Theorem~\ref{thm:savegeneratorsforpquotients} gives a more precise statement for the range of possible valencies of our construction.

\subsection{Change of Generators in Cayley Graphs}\label{sec:3_change_of_gens}

This section is concerned with the expansion property of Cayley graphs under change of generators. In general, the expansion property is not preserved (see, e.g., \cite[Section 11.4]{HooryLinialWigderson2006} for a brief discussion). However, if the two families of generators remain bounded and can be mutually written in words of bounded lengths of each other, then the expansion property is preserved (see, e.g., \cite[Prop. 3.5.1]{Kowalski2018}).
The remainder of this section is devoted to an easy self-contained proof in the following special case: Let $\Gamma$ be a finitely generated infinite group generated by two finite symmetric sets of generators $S_1$ and $S_2$, that is, $S_1^{-1} = S_1$ and $S_2^{-1} = S_2$. Let $\Gamma_k$ be finite index subgroups of $\Gamma$ with $[ \Gamma : \Gamma_k ] \to \infty$ and $H_k = \Gamma / \Gamma_k$ be the corresponding finite quotients. The sets $S_i$ of generators of $\Gamma$ can also be viewed as sets of generators of the quotients $H_k$. Then the following holds:

\begin{proposition}[see, e.g., {\cite[Prop. 3.5.1]{Kowalski2018}}]
If the Cayley graphs ${\mathcal{C}}(H_k,S_1)$ represent a family of expander graphs then 
${\mathcal{C}}(H_k,S_2)$ is also a family of expander graphs. 
\end{proposition}

The proof is based on a well-known spectral description of expander graphs: Let $G = (V,E)$ be a connected graph with vertex set $V$ and edge set $E$. The degree of a vertex $x \in V$ is denoted by $d_x$. The Laplacian $L_G$
of a function $f: V \to \mathbb{R}$ is defined as follows
\[ L_G f(x) = \sum_{y \sim x} (f(x)-f(y)) = d_x f(x) - \sum_{y \sim x} f(y)\,.\]
Then $L_G$ is a symmetric operator with non-negative real eigenvalues. Let $\mu_1(G) > 0$ be the smallest non-zero eigenvalue of $L_G$. Since Cayley graphs are regular graphs, there is a close relationship between the eigenvalues of $L_G$ and the adjacency matrix $A_G$ (see, e.g., \cite[Lemma 4.7]{HooryLinialWigderson2006}), namely $\mu_1(G) = \lambda_1(A_G)-\lambda_2(A_G)$, where $\lambda_1(A_G),\lambda_2(A_G)$ are the two largest eigenvalues of $A_G$. The Cayley graphs ${\mathcal{C}}(H_k,S_1)$ from above are an expander family if and only if there exists a positive constant $C > 0$ such that, for all $k \in \mathbb{N}$ (see, e.g., \cite[Theorem 4.11]{HooryLinialWigderson2006}):
\[ \mu_1({\mathcal{C}}(H_k,S_1)) \ge C\,.\]
Via the Rayleigh quotient 
\[
\mathcal{R}(f)  :=
\frac{\sum_{x \in V} f(x) \cdot L_Gf(x)}{\sum_{x \in V} f(x)^2} =
\frac{\sum_{\{x,y\} \in E} (f(x)-f(y))^2}{\sum_{x \in V} f(x)^2} 
\]
the eigenvalue $\mu_1(G)$ has the following variational description (see, e.g., \cite[Prop. 1.82]{KrebsShaheen2011} and \cite[Prop. 3.4.3]{Kowalski2018}) 
\[
\mu_1(G)  = \inf \Big\{ \mathcal{R}(f) : \, \sum_{x \in V} f(x) = 0, f \neq 0 \Big\}.
\]

The aim is to show that there exists a positive constant $K > 0$ such that, for all $k \in \mathbb{N}$,
\[ \mu_1({\mathcal{C}}(H_k,S_2)) \ge K \mu_1({\mathcal{C}}(H_k,S_1))\,.\]
Then the expanding property of the family ${\mathcal{C}}(H_k,S_1)$ implies a similar expanding property of the family ${\mathcal{C}}(H_k,S_2)$, albeit with a different spectral expansion constant.

Before starting the proof, we first consider the description of the Rayleigh quotient $\mathcal{R}(f)$ in the case of a Cayley graph ${\mathcal{C}}(H,S)$: The vertex set of this graph is given by $V = H$ and the set of edges is given by $E = \{ \{x,xs\}: \, x \in V, s \in S \}$, where $E$ is understood as a multiset. Then the Rayleigh quotient for a function $f: H \to \mathbb{R}$ can be written as
\[ \mathcal{R}_S(f) = \frac{1}{2}\, \frac{\sum_{x \in H} \sum_{s \in S} (f(x)-f(xs))^2}{\sum_{x \in H} f(x)^2}. \]
The Rayleigh quotient is equipped with the index $S$, that is, the set of generators of the Cayley graph, since it depends on this choice of generators.
Note that ${\mathcal{C}}(H,S)$ is vertex transitive since $H$ acts on the vertex set $V = H$ by group left-multiplication. The proof is complete if there exists a constant $K > 0$ such that, for all $k \in \mathbb{N}$ and all $f: H_k \to \mathbb{R}$,
\begin{equation} \label{eq:RS2KRS1} 
\mathcal{R}_{S_2}(f) \ge K \mathcal{R}_{S_1}(f). 
\end{equation}
This is the case if there exists, for every generator
$s \in S_1$, a constant $K'(s) > 0$ such that
\begin{equation} \label{eq:toshow} 
\sum_{x \in H_k} (f(x)-f(xs))^2 \le K'(s) \sum_{x \in H_k} \sum_{t \in S_2} (f(x)-f(xt))^2, 
\end{equation}
and that his constant $K'(s)$ does not depend on $k \in \mathbb{N}$. Since $S_2$ is a set of generators of $\Gamma$, any $s \in S_1 \subset \Gamma$ can be written in the form 
\[ s = t_{1} t_{2} \cdots t_{n} \]
with $t_{1},\dots,t_{n} \in S_2$. The same relation between the generators holds in each of the quotients $H_k$, $k \in \mathbb{N}$. We abbreviate $s_j = t_1 \cdots t_j$, so $s_0 = 1$ and $s_n = s$, and $s_{j} = s_{j-1} t_{j}$. Let us now show \eqref{eq:toshow}:
\begin{multline*}
\sum_{x \in H_k} (f(x)-f(xs))^2 = \sum_{x \in H_k} \Big( \sum_{j=1}^{n} (f(xs_{j-1}) - f(xs_{j})) \Big)^2 \\ \le n \sum_{x \in H_k} \sum_{j=1}^{n} \big(f(xs_{j-1}) - f(xs_{j-1}t_{j})\big)^2
 = n \sum_{y \in H_k} \sum_{j=1}^{n} (f(y) - f(yt_{j}))^2 \\
\le n^2 \sum_{y \in H_k} \sum_{t \in S_2} (f(y)-f(yt))^2.
\end{multline*}
This shows \eqref{eq:toshow} with $K'(s) = n^2$, where $n$ is the length of a word in $S_2$ expressing $s$.

Let us finally show \eqref{eq:RS2KRS1}:
\begin{multline*}
\mathcal{R}_{S_1}(f) = %\frac{1}{2}\, 
\frac{\sum_{s \in S_1} \sum_{x \in H_k} (f(x)-f(xs))^2}{2 \sum_{x \in H_k} f(x)^2} \le %\frac{1}{2}\, 
\frac{\sum_{s \in S_1} K'(s) \sum_{x \in H_k} \sum_{t \in S_2} (f(x)-f(xt))^2}{2 \sum_{x \in H_k} f(x)^2} \\
\le \Big( \sum_{s \in S_1} K'(s) \Big)\, %\frac{1}{2}
\frac{\sum_{x \in H_k} \sum_{t \in S_2}(f(x)-f(xt))^2}{2 \sum_{x \in H_k} f(x)^2}
= \Big( \sum_{s \in S_1} K'(s) \Big)\, \mathcal{R}_{S_2}(f).
\end{multline*}
This implies that \eqref{eq:RS2KRS1} holds with $K = \left( \sum_{s \in S_1} K'(s) \right)^{-1}$.

\subsection{Unbounded Torsion in Congruence Quotients}

Let $\Gamma$ be a group and let $R$ be a ring. Let $\GL_n(R)$ denote the \emph{general linear group} of $n \times n$ matrices with entries in $R$ and whose determinant is a unit in $R$. A \emph{representation} of $\Gamma$ with coefficients in $R$ of dimension $n$ is a group homomorphism
\[
\rho \colon \Gamma \to \GL_n(R),
\]
i.e., $\rho(g_1 \cdot g_2) = \rho(g_1) \cdot \rho(g_2)$ for all $g_1,g_2 \in \Gamma$. The representation $\rho$ is said to be \emph{faithful} if $\rho$ is injective.

The ring $\mathbb{F}_p[[t]]$ is the ring of formal power series in the variable $t$ and coefficients in the finite field $\mathbb{F}_p$ of $p$ elements. Considering a formal power series $f(t) = \sum_{k \geq 0} a_k t^k$ only up to terms of order $t^{i+1}$ or higher (fixed precision) as 
\[
f(t) = \sum_{k = 0}^i a_k t^k + O(t^{i+1})
\]
is described by a ring homomorphism $\pi_i \colon \mathbb{F}_p[[t]] \to \mathbb{F}_p[[t]] / (t^{i+1}) = \mathbb{F}_p[t] / (t^{i+1})$.

A group is said to be \emph{torsion-free} if all non-trivial elements have infinite order. 

\begin{lemma}
\label{lem:growingorderintrucatedquotients}
Let $\Gamma$ be an infinite, torsion-free group and let 
\[
\rho \colon \Gamma \to \GL_n(\mathbb{F}_p[[t]])
\]
be a faithful representation. For any $i\geq 0$ let $\Gamma_i$ be the finite group which is the image of the composition
\[
\Gamma \xrightarrow{\rho} \GL_n(\mathbb{F}_p[[t]]) \xrightarrow{\pi_i} \GL_n(\mathbb{F}_p[t] / (t^{i+1}))\,.
\]
and we denote by $\psi_i = \pi_i \circ \varphi: \Gamma \to \Gamma_i$ the corresponding surjective map from $\Gamma$ to $\Gamma_i$. 

Then for $g \in \Gamma$ a nontrivial element, the order of $\psi_i(g)$ in $\Gamma_i$ tends to infinity as $i$ increases.
\end{lemma}
\begin{proof}
We argue by contradiction. If $N$ is an upper bound for the orders of all $\psi_i(g)$, then with $\ell = N!$ we have $\psi_i(g^\ell) = \psi_i(g)^\ell = 1$ for all $i$. Since $\rho(g^\ell)$ is described formally to arbitrary precision by $\psi_i(g^\ell) = 1$, we conclude that $\rho(g^\ell) = 1$. Because $\rho$ is faithful, it follows that $g^\ell = 1$. Moreover, as $\Gamma$ is torsion-free, we conclude $g=1$, a contradiction.
\end{proof}

\subsection{Quartic Cayley Graph Expanders of 3-Groups}\label{sec:3_cayley_expanders}
We start with the least complex example underlying our construction of Cayley graph expanders, namely a group $\Gamma$ generated by $a_1,a_2,a_3,a_4$ subject
to relations 
\[
a_1a_3a_1a_4, \quad 
a_1a_3^{-1}a_2a_3^{-1}, \quad
a_1a_4^{-1}a_2^{-1}a_4^{-1}, \quad
a_2a_3a_2a_4^{-1}\,.
\]
The group $\Gamma$ agrees with the group $\Gamma_2$ on page 457 of \cite{SV}, where the correspondence is $a_1= g_0$, $a_2=g_2$, $a_3=g_3$, $a_4= g_1$. 
The group $\Gamma$ can be obtained as the group $\Gamma_{3;0,1}$ described in Section~\ref{sec:morelattices} (a particular case of the group $\Gamma_{M,\delta}$ of \cite[Section 2.8]{RSV} for $q=3$) by means of the identification $\Gamma \cong \Gamma_{3;0,1}$ via $a_1 \mapsto a_0$, $a_2 \mapsto a_2$, $a_3 \mapsto b_3^{-1}$ and $a_4 \mapsto b_1^{-1}$.
It was shown by Stix and Vdovina in \cite{SV}, that this group is a quaternionic lattice.
The group $\Gamma$ is torsion-free by \cite[Theorem 30]{SV}.

Consider the function field $\mathbb{F}_3(x)$ over the finite field $\mathbb{F}_3$ in one variable $x$. 
Using computer algebra, it is easy to verify that the following assignments give a well-defined representation $\Psi: \Gamma \to \mathrm{GL}_3(\mathbb{F}_3(x))$: 

\begin{align}
\Psi(a_1) & = \left(\begin{array}{rrr}
\frac{1}{x} & 1 + \frac{1}{x} & -x +1 - \frac{1}{x} \\
- \frac{1}{x} & - \frac{1}{x} & -x + \frac{1}{x} \\
- \frac{1}{x} & 1 - \frac{1}{x} & x + 1 + \frac{1}{x}
\end{array}\right) \label{align:matricesPsi}
\\[3ex]
\Psi(a_2) & =  \left(\begin{array}{rrr}
\frac{1}{x} & - x + 1 & - x^{3} - x^{2} -  x \\
0 & 1 & - x^{2} + x \\
0 & 0 & x
\end{array}\right) \notag
\\[3ex]
\Psi(a_3) & =  \left(\begin{array}{rrr}
0 & 0 & x^{2} + 1 \\
0 & -1 & - x - 1 \\
\frac{1}{x^{2} + 1} & \frac{x + 1}{x^{2} + 1} &
-1 + \frac{x}{x^{2} + 1}
\end{array}\right) \notag
\\[3ex]
\Psi(a_4) & =  \left(\begin{array}{rrr}
\frac{2 x}{x^{2} + 1} - 1 & -x - \frac{1}{x^{2} +1} & 
x^2 - x + 1 - \frac{x}{x^{2} + 1} \\
\frac{2x + 1}{x^{2} + 1} & \frac{2 x - 1}{x^{2} + 1} &
\frac{2x + 1}{x^{2} + 1} - x  - 1\\
\frac{1}{x^{2} + 1} & \frac{2 x}{x^{2} + 1} &
\frac{1}{x^{2} + 1} - 1
\end{array}\right) \notag
\end{align}
In fact, all matrices $\Psi(a_i)$, for $i=1,\ldots, 4$ 
 have determinant $1$.  
The denominators of the matrix entries are nonzero at $x = 1$. Therefore we can substitute $x = 1 + t$ and expand the matrix entries, which are now rational functions in $t$, as formal power series in $\mathbb{F}_3[[t]]$.
The determinant is still $1$, so that the resulting matrices are invertible as matrices with values in $\mathbb{F}_3[[t]]$. We thus obtain a representation 
\[
\widetilde \Psi: \Gamma \to \mathrm{GL}_3(\mathbb{F}_3[[t]]).
\]
\begin{lemma} \label{lem:faithfulreptoSL3}
The group homomorphism $\widetilde \Psi: \Gamma \to \mathrm{GL}_3(\mathbb{F}_3[[t]])$ given above is injective and defines a faithful representation of $\Gamma$.
\end{lemma}
\begin{proof}
Because $\widetilde{\Psi}(a_2)$ agrees with $\Psi(a_2)$ under the necessary identifications, it has the same order, which is visibly infinite (due to the upper triangular shape and the entry $x$ on the diagonal). Therefore $\widetilde{\Psi}$ has infinite image. 

The group $\Gamma$ is an arithmetic lattice in a group of rank $2$ by construction in \cite{SV}. By \cite{Margulis}, therefore  all homomorphic images of $\Gamma$ are finite, or the kernel of the homomorphism is finite. As $\widetilde{\Psi}(\Gamma)$ is infinite, it follows that the kernel of $\widetilde{\Psi}$ is finite. But 
$\Gamma$ is torsion-free, hence all finite subgroups are trivial, and trivial kernel means that $\widetilde{\Psi}$ is faithful.
\end{proof}

As in Lemma~\ref{lem:growingorderintrucatedquotients} we consider the truncated representations up to orders $t^{i+1}$ and higher
\[
\widetilde{\Psi}_i : \Gamma \xrightarrow{\widetilde{\Psi}} 
\mathrm{GL}_3(\mathbb{F}_3[[t]]) \xrightarrow{\pi_i} \mathrm{GL}_3(\mathbb{F}_3[t]/(t^{i+1}))
\]
and denote the image by 
$\Gamma_i := \widetilde{\Psi}_i(\Gamma)$.

\begin{remark}
As an alternative to the proof of Lemma~\ref{lem:faithfulreptoSL3} one may observe that $\Gamma$ is a lattice in a quaternion algebra $D$ over a field $K$ (or rather an arithmetic lattice in $D^\times/K^\times$) that is split by $\mathbb{F}_3(x)$ (as an extension of $K$). The homomorphism $\Psi$ is nothing but the one induced from $D$ acting on the $3$-dimensional purely imaginary quaternions in $D$ by conjugation. It follows that $\Gamma$ acts faithfully since $D^\times/K^\times$ acts faithfully. 

More importantly, it follows from this construction that the truncated homomorphisms $\Gamma \to \Gamma_i$ are finite congruence quotients of $\Gamma$.
\end{remark}

\begin{lemma} \label{lem:residuallypro-3}
The finite groups $\Gamma_i$ are $3$-groups, i.e, the order of $\Gamma_i$ is a power of $3$.
\end{lemma}
\begin{proof}
Forgetting the term of order $t^i$ yields a group homomorphism $\Gamma_i \to \Gamma_{i-1}$. Its kernel is naturally a subgroup of the additive group $\mathrm{M}_3(\mathbb{F}_3)$ of $3 \times 3$ matrices with coefficients in $\mathbb{F}_3$, hence a $3$-group. Indeed, the elements in $\ker(\Gamma_i \to \Gamma_{i-1})$ have the form $\mathbb{1}_3 + At^i$ for an $A \in \mathrm{M}_3(\mathbb{F}_3)$ and the identity matrix $\mathbb{1}_3$. The assignment $\mathbb{1}_3 + At^i \mapsto A$ is an injective group homomorphism $\ker(\Gamma_i \to \Gamma_{i-1}) \to \mathrm{M}_3(\mathbb{F}_3)$. 

It remains to prove the claim for $i=0$, the leading constant term, and to conclude by induction thanks to Lagrange's theorem on group orders.

For $i=0$, we simply plug in $t=0$, or what is the same $x=1$, into \eqref{align:matricesPsi} to get formulas for $\widetilde{\Psi}_0: \Gamma \to \Gamma_0 \subseteq \GL_3(\mathbb{F}_3)$ as follows:
\[
\widetilde{\Psi}_0(a_1) = \left(\begin{array}{rrr}
1 & - 1 & -1 \\
- 1 & - 1 & 0 \\
- 1 & 0 & 0
\end{array}\right),
\quad
\widetilde{\Psi}_0(a_2)  =  \left(\begin{array}{rrr}
1 & 0 & 0 \\
0 & 1 & 0\\
0 & 0 & 1
\end{array}\right) ,
\]
\[
%\quad
\widetilde{\Psi}_0(a_3)  =  \widetilde{\Psi}_0(a_4)  = \left(\begin{array}{rrr}
0 & 0 & - 1 \\
0 & -1 & 1 \\
- 1 &  1 & 1
\end{array}\right) ,
\]
which generate a cyclic group of order $3$, namely $\Gamma_0$.
\end{proof}

In fact, the kernel $N_i = \ker(\widetilde{\Psi}_i \colon \Gamma \to \Gamma_i)$ is by construction described by congruences modulo $t^{i+1}$, and --- unraveling the definition of \cite{SV,RSV} --- we see that $N_i$ is a congruence subgroup of the lattice $\Gamma$. There is a corresponding infinite series of square complexes $P_i$, the quotient of the product of trees $T_4 \times T_4$ by $N_i$, with the number of vertices being a power of 3. Indeed, the $3$-group $\Gamma_i = \Gamma/N_i$ acts simply transitively on the vertices of $P_i$. The $1$-skeleton of $P_i$ is a the Cayley graph $G_i = \mathcal{C}(\Gamma_i,S_i)$, for $S_i$ the image in $\Gamma_i$ of  the set of generators $\{a_1^{\pm 1}, a_2^{\pm 1}, a_3^{\pm 1}, a_4^{\pm 1}\}$ of $\Gamma$ as considered in \cite{SV,RSV}. It follows from \cite{RSV} that these $G_i$ form an infinite series of Cayley graphs expanders of valency $8$ for $3$-groups.

\begin{theorem}\label{thm:good_expanders}
The quaternionic lattice $\Gamma_2$ of page 457 of \cite{SV}, introduced above as $\Gamma$, is an infinite group with two generators $x_0,x_1$, and an infinite sequence $\{N_i\}_{i\in \mathbb{N}}$ of normal subgroups such that the following are true:
\begin{enumerate}
    \item 
    The indices $[\Gamma:N_i]=n_i$ are powers of $3$.
    \item 
    Let $v_0=x_0N_i$ (resp.\ $v_1=x_1N_i$) be the image of the generator $x_0$ (resp.\ $x_1$) under the quotient map $\Gamma \to \Gamma_i := \Gamma/N_i$.
    The orders of %each of 
    the four subgroups $H^1_i$, $H^2_i$, $H^3_i$, $H^4_i$ of $\Gamma_i$, generated by $v_0$, $v_1$, $v_1^{-1}v_0$, $v_1v_0$, respectively, converge to infinity as $i$ increases. %Here $v_0=x_0N_i$ and $v_1=x_1N_i$.
    \item 
    There exists a positive constant $c>0$ such that the set $\mathcal{G}=\{G_1,G_2,\dots\}$ of Cayley graphs 
    \[
    G_i=\mathcal{C}(\Gamma_i,\{v^{\pm 1}_0,v^{\pm 1}_1\})
    \] 
    is a family of $(n_i,4,c)$-expanders.
\end{enumerate}
\end{theorem}
\begin{proof}
The generators are the elements $x_0 = a_2$ and $x_1 = a_3$ of above. These two elements generate $\Gamma$ because the other defining generators can be written (using the defining relations) as
\[
a_1 = a_3 a_2^{-1} a_3, \quad a_4 = a_2 a_3 a_2.
\]
The normal subgroups are the groups 
$N_i= \ker(\widetilde{\Psi}_i)$ of above. The indices $[\Gamma:N_i] = \#\Gamma_i$ are powers of $3$ due to Lemma~\ref{lem:residuallypro-3}. Lemma~\ref{lem:growingorderintrucatedquotients} shows the assertion on the asymptotic of the orders of the images of specific elements. In order to apply this to $v_0$, $v_1$, $v_1^{-1}v_0$ and $v_1v_0$ we must shoe that $x_0$, $x_1$, $x_1^{-1}x_0$ and $x_1x_0$ are all non-trivial. Consider the homomorphism $\Gamma \to \mathbb{Z}/2\mathbb{Z} \times \mathbb{Z}/2\mathbb{Z}$ given by counting modulo $2$ the number of even and odd indexed $a_i$ occurring in a word representing a group element of $\Gamma$. Then $x_0$, $x_1$, $x_1^{-1}x_0$ and $x_1x_0$ all have nontrivial image: $(1,0)$, $(0,1)$, $(1,1)$ and $(1,1)$ respectively, hence these elements are non-trivial.

The expander property for the Cayley graphs $\mathcal{C}(\Gamma_i,S_i)$, for the larger set of generators $S_i$ that is the image in $\Gamma_i$ of  the set of generators $\{a_1^{\pm 1}, a_2^{\pm 1}, a_3^{\pm 1}, a_4^{\pm 1}\}$ of $\Gamma$, was recalled before the theorem. The result of Section~\ref{sec:3_change_of_gens} shows that the Cayley graphs 
$G_i=\mathcal{C}(\Gamma_i,\{v^{\pm 1}_0,v^{\pm 1}_1\})$
of $\Gamma_i$ with respect to the generating set given by the images of $\{a_2^{\pm}, a_3^{\pm}\}$ in $\Gamma_i$ is still a sequence of expander graphs (although with a different expansion constant). The valency is now $4$ and the vertex set has cardinality $n_i$, a power of $3$. 
\end{proof}

\subsection{\texorpdfstring{Explicit Cayley Graphs $p$-Expanders with $p-2$ Generators}{Explicit Cayley Graphs p-Expanders with p-2 Generators}}
\label{sec:morelattices}
We recall the explicit description of the lattices from \cite[Section 2.8]{RSV}. For details we refer to loc. cit.

Let $p \ge 2$ be a prime number and let $\mathbb{F}_{p^2}$ be the field with $p^2$ elements. Its multiplicative group $\mathbb{F}_{p^2}^\times$ is cyclic, and we fix a generator $\delta \in \mathbb{F}_{p^2}^\times$. We define for $k,j \in \mathbb{Z}/(p^2-1)\mathbb{Z}$ such that $k \not\equiv j \pmod {p-1}$ the elements $x_{k,j}, y_{k,j} \in \mathbb{Z}/(p^2-1)\mathbb{Z}$ uniquely by 
\[
\delta^{x_{k,j}} = 1 + \delta^{j-k}, \quad \delta^{y_{k,j}} = 1 + \delta^{k-j}.
\]
(This is possible since $\delta^{j-k} \not= -1$.) 
We set further in $\mathbb{Z}/(p^2-1)\mathbb{Z}$:
\[
i(k,j) = j - y_{k,j}(p-1), \quad \ell(k,j) = k - x_{k,j}(p-1).
\]
We now fix two elements $\alpha \not= \beta \in \mathbb{Z}/(p-1)\mathbb{Z}$, consider the reduction modulo $p-1$ given by 
\[
\pr: \mathbb{Z}/(p^2-1)\mathbb{Z} \to \mathbb{Z}/(p-1)\mathbb{Z},
\]
and define $q+1$-element sets $K = \pr^{-1}(\alpha)$ and $J = \pr^{-1}(\beta)$. Since $p-1$ divides $\mu = (p^2-1)/2$, the sets $K$ and $J$ are preserved under translation by $\mu$.

The group $\Gamma_{p;\alpha, \beta}$ (the dependence on $\delta$ is implicit) is defined by generators
$a_k$ for $k \in K$ and $b_j$ for $j \in J$ 
subject to the relations: for all $k \in K$ and $j \in J$ we have
$a_{k} a_{k + \mu} = 1$, and $b_j b_{j+\mu} = 1$, and 
\begin{equation} \label{eq:squarerelation}
 a_k b_j  a_{\ell(k,j)}^{-1} b_{i(k,j)}^{-1} = 1.
\end{equation}

It was proven in \cite{RSV} that $\Gamma_{p;\alpha,\beta}$ is a quaternionic arithmetic lattice of rank $2$, and residually pro-$p$ by congruence $p$-group quotients. Moreover, congruence quotients yield Cayley graphs with respect to the given generators $A = \{a_k; k \in K\}$ together with $B = \{b_j; j \in J\}$ that form a sequence of expanders (as the $1$-skeleton of $2$-dimensional Ramanujan expander complexes). The valency of these graphs is $\#(K \cup J) = 2(p+1)$.

For the applications in this note we must reduce the number of essential generators (not counting inverses).

\begin{proposition}
\label{prop:generateWithBandoneA}
The group $\Gamma_{p;\alpha,\beta}$ can be generated by $(p+3)/2$ elements. More precisely, half of the elements of $B$ (omitting inverses) together with one $a \in A$ generate $\Gamma_{p;\alpha,\beta}$.
\end{proposition}
\begin{proof}
The relations \eqref{eq:squarerelation} comes from the squares in the one vertex square complex with complete bipartite link whose fundamental group $\Gamma_{p;\alpha,\beta}$ is. It follows that for all $j \in J$ the map
\[
\sigma_j \colon K \to K, \quad \sigma_j(k) =  \ell(k,j)
\]
is bijective. The argument of \cite[Proposition 35]{SV} works also for the groups $\Gamma_{p;\alpha,\beta}$ and shows that the group $P_B$ generated by all $\sigma_j$, for $j \in J$ in the symmetric group on the set $K$ acts transitive on $K$. Since the relation \eqref{eq:squarerelation} shows that 
\[
a_{\sigma_j(k)} = b_{i(k,j)}^{-1} a_k b_j,
\]
a subgroup containing all $b \in B$ will contain with any $a_k$ also the $a_{k'}$ for $k' \in K$ in the $P_B$-orbit of $k$. By transitivity of $P_B$ on $K$ this automatically involves all of $A$. This proves the theorem.
\end{proof}

\begin{example} \label{example:p5actionOnT6xT6generators3}
The following is an explicit example for the above construction for $p=5$ and $\alpha = 0, \beta = 2$ in $\mathbb{Z}/4\mathbb{Z}$. The group $\Gamma_{5; 0,2}$ is generated by elements (indices are to be considered modulo $24$)
\[
a_0,a_4,a_8, %a_{12},a_{16},a_{20},
b_2,b_6,b_{10}, %,b_{14},b_{18},b_{22}
\]
%with $a_i^{-1} = a_{i+12}$, and $b_j^{-1} = b_{j+12}$, 
and $9$ relations of length $4$
\[
a_{0}b_{2}a_{0}b_{10}, \ 
a_{0}b_{6}a_{0}^{-1}b_{6}^{-1}, \ 
a_{0}b_{2}^{-1}a_{4}^{-1}b_{2}^{-1}, \ 
a_{0}b_{10}^{-1}a_{8}b_{10}^{-1}, \ 
a_{4}b_{6}a_{4}b_{2}^{-1},  
\]
\[
a_{4}b_{10}a_{4}^{-1}b_{10}^{-1}, \ 
a_{4}b_{6}^{-1}a_{8}^{-1}b_{6}^{-1}, \ 
a_{8}b_{2}a_{8}^{-1}b_{2}^{-1}, \ 
a_{8}b_{10}a_{8}b_{6}^{-1}.  
\]
A direct inspection shows that $\Gamma_{5;0,2}$ can be generated by $2$ elements, for example $a_0, b_2$.
\end{example}

\begin{theorem}
\label{thm:savegeneratorsforpquotients}
Let $p \ge 3$ be a prime number, and let $d \geq 2$ be an integer. We assume that $d \ge (p+3)/2$ if $p \ge 7$.

There are infinitely many finite $p$-groups $\Gamma_i$ with order tending to infinity and  generating sets $T_i$ of cardinality $d$ such that the Cayley graphs $\mathcal{C}(\Gamma_i, T_i)$ form a family of expanders with number of vertices a power of $p$ and valency $2d$.

In particular, for $p \ge 5$ we may take $d = p-2$. 
\end{theorem}
\begin{proof}
We consider the groups constructed in \cite[Section 2.8]{RSV} as recalled in Section~\ref{sec:morelattices} with notation $\Gamma_{p;\alpha,\beta}$. By Proposition~\ref{prop:generateWithBandoneA}, these groups can be generated by $(p+3)/2$ elements. Moreover, for $p=5$ we take the group $\Gamma_{5;0,2}$ considered in 
Example~\ref{example:p5actionOnT6xT6generators3}, which can be generated by $2$ elements. And for $p=2$ we take the $2$-generated group $\Gamma_{3,0,1}$ of Theorem~\ref{thm:good_expanders}. For any $d$ as in the statement of the theorem, we may therefore choose a set $T$ of $d$ generators of the respective infinity group $\Gamma_{p;\alpha,\beta}$ (adding arbitrary elements if necessary for values of $d$ larger than the given minimal values).

If $p \ge 7$, then $p-2 \ge (p+3)/2$ and $d=p-2$ is a possible choice. For $p=5$, we have $p-2 \ge 2$, so $d=p-2$ is a valid choice for all $p \ge 5$.

It follows from \cite[Proposition 2.22]{RSV} that these groups $\Gamma_{p;\alpha,\beta}$ are residually pro-$p$ with respect to a suitable infinite sequence of congruence subgroup quotients of $p$-power order. 

The standard generating sets $A \cup B$ of $\Gamma_{p;\alpha, \beta}$ yield for the sequence of congruence quotients of $\Gamma_{p;\alpha, \beta}$ that the corresponding Cayley graphs form a series of expanders (as the $1$-skeleton of $2$-dimensional Ramanujan expander complexes), see \cite[Section 6]{RSV}. 

Indeed, by the results of Section~\ref{sec:3_change_of_gens}, all these Cayley graphs of the finite congruence $p$-group quotients of $\Gamma_{p;\alpha,\beta}$ with respect to the alternative set $T$ of $d$ generators are still expanders  (but not necessarily Ramanujan). This proves the theorem.
\end{proof}

Let us conclude this section by emphasizing that, in combination, Theorems~\ref{thm:good_expanders} and~\ref{thm:savegeneratorsforpquotients} yield Theorem~\ref{thm:main_expanders_intro}.

\section{Lower Bounds for (modular) Subgraph Counting}\label{sec:lowerbounds}

\subsection{Counting Subgraphs with Minor-Closed Properties}\label{sec:minorclosed}
In this section, we will prove the full classification for counting small subgraphs with minor-closed properties, which we restate for convenience:

\minclosehard*

We will begin with an outline of the proof.
Let $\Phi$ be a minor-closed graph property of unbounded matching number that is not trivially true. We will establish hardness of the colourful version $\#\coledgesubsprob(\Phi)$; the reduction to the uncoloured version $\#\edgesubsprob(\Phi)$ will be an easy consequence of the inclusion-exclusion principle.

In proving hardness of $\#\coledgesubsprob(\Phi)$, we follow the approach of~\cite{RothSW20unpub}: We cast the problem as computing finite linear combinations of (coloured) homomorphism counts and prove that there exists a family of regular expanders that do not vanish in those linear combinations, which is known to be sufficient for hardness as expanders have high treewidth. 
It was shown in~\cite{RothSW20unpub} that the expanders do not vanish in the aforementioned linear combinations whenever their indicators, as defined in Section~\ref{sec:obstructions}, are non-zero, in which case we called the family of expanders an obstruction for $\Phi$. Consequently, the proof of Theorem~\ref{thm:minclose-hard} requires us to show that every non-trivial minor-closed graph property of unbounded matching number has an obstruction; note that this proof strategy is fully encapsulated by Lemma~\ref{lem:monotonicity}.

In~\cite{RothSW20unpub} it was shown that $\Phi$ always has an obstruction if each forbidden minor has a vertex of degree at least $3$. More precisely, the obstruction was given by a family of Cayley graph expanders of $2$-groups. However, this family of expanders is not an obstruction if $\Phi$ has a forbidden minor of degree at most $2$.

In what follows, we will therefore show that our novel Cayley graph expanders of $3$-groups constructed in Section~\ref{sec:3_cayley_expanders} are obstructions for the remaining cases.

For technical reasons, we have to consider two special cases which are dealt with separately in the following subsection: Properties that are true on all cycles, which we will call ``unsuitable'', and properties of so-called bounded wedge-number.

\subsubsection{Unsuitable Properties and Wedge-Numbers}
Let us write $H'\prec H$ if $H'$ is a minor of $H$. The degree of a graph is the maximum degree of its vertices. Graph union is denoted by $+$. Furthermore, given a graph $G$ with $k$ (distinct) edge colours and a graph property $\Phi$, it will be convenient to write $\coledgesubs{\Phi,k}{G}$ for the set of all $k$-edge subsets $A\subseteq E(G)$ such that $A$ contains each colour (precisely) once and $\Phi(G[A])=1$; observe that $\#\coledgesubsprob(\Phi)$ then just rewrites to the problem of computing the cardinality $\#\coledgesubs{\Phi,k}{G}$.

Given a minor-closed graph property $\Phi$, write $\mathcal{F}(\Phi)$ for the set of minimal forbidden minors of $\Phi$. By the Robertson-Seymour-Theorem~\cite{RobertsonS04}, we have that $\mathcal{F}(\Phi)$ is finite.
We furthermore write $\mathcal{F}_2(\Phi)$ for the subset of $\mathcal{F}(\Phi)$ containing only graphs of degree at most $2$. We say that $\Phi$ is \emph{suitable} if $\Phi(C_k)=0$ for some positive integer $k\geq 3$; otherwise $\Phi$ is called \emph{unsuitable}.

Observe that every unsuitable minor-closed graph property $\Phi$ has unbounded matching number: Assuming otherwise, we obtain that $\Phi(M_d)=0$ for some positive integer $d$. But $M_d \prec C_{2d}$, thus $\Phi(C_{2d})=0$ since $\Phi$ is minor-closed.

The following lemma allows us to show that every unsuitable property can be reduced from a suitable one; we also include a modular version of the reduction which will be required later.

\begin{lemma}\label{lem:minor_trick}
	Let $\Phi$ be a minor-closed graph property. If $\mathcal{F}_2(\Phi)\neq \emptyset$, but $\Phi$ is unsuitable, then there exists a graph property $\Psi$ such that
	\begin{enumerate}
		\item $\Psi$ is minor-closed and of unbounded matching number,
		\item $\Psi$ is suitable, and
		\item $\#\coledgesubsprob(\Psi) \fptred \#\coledgesubsprob(\Phi)$ and, on input $G$ and $k$, every oracle query $(G',k')$ of the reduction satisfies $|G'|\in O(|G|)$ and $k' \in O(k)$.
		\item For each prime $p$ we have $\#_p\coledgesubsprob(\Psi) \fptred \#_p\coledgesubsprob(\Phi)$ and, on input $G$ and $k$, every oracle query $(G',k')$ of the reduction satisfies $|G'|\in O(|G|)$ and $k' \in O(k)$.
	\end{enumerate}
\end{lemma}
\begin{proof}
	Since $\Phi$ is not suitable, it is true on arbitrarily large cycles $C_k$. Since any finite union of paths is a minor of a sufficiently large cycle, any such union also satisfies $\Phi$. Now any graph $F \in \mathcal{F}_2$ is a union of circles and paths\footnote{We see an isolated vertex as a path of length $0$.} and by the argument above, it must contain at least one circle. Using again that $\Phi$ is not suitable, we infer that each graph in $\mathcal{F}_2(\Phi)$ is the union of a cycle and a non-empty path-cycle-packing. Since $\mathcal{F}_2(\Phi)\neq \emptyset$, we can choose $F\in\mathcal{F}_2(\Phi)$ such that the number of cycles in $F$ is minimal among all graphs in $\mathcal{F}_2(\Phi)$. Thus $F=C + R$ for a cycle $C$ and a (non-empty) path-cycle-packing~$R$. We define
	\[ \Psi(H)=1 :\Leftrightarrow \Phi(H+R)=1 \,.\]
	Let us now prove 1.\ - 4.:
	\begin{enumerate}
		\item First assume that $\Psi(H)=1$ and $H'\prec H$; thus $H'+R \prec H+R$. Then \[\Psi(H)=1 \Rightarrow \Phi(H+R)=1 \Rightarrow \Phi(H'+R)=1 \Rightarrow \Psi(H')=1\,.\] 
		The second implication holds as $\Phi$ is minor-closed. Consequently, $\Psi$ is minor-closed as well.
		
		Next assume for contradiction that $\Psi$ has bounded matching number. Then there exists a positive integer $d$ such that $\Psi(M_d)=0$ which implies that $\Phi(M_d + R)=0$. Since $M_d +R$ is of degree $2$, and every minor of a graph of degree $2$ has degree at most $2$, one of the graphs in $\mathcal{F}_2(\Phi)$ must be a minor of $M_d + R$. Observe that every minor of $M_d +R$ contains strictly fewer cycles than $F(=C+R)$. However, we choose $F$ in such a way that the number of its cycles is minimal among all graphs in $\mathcal{F}_2(\Phi)$, which yields the desired contradiction.
		
		\item Since $0=\Phi(F)=\Phi(C+R)$ we immediately obtain that $\Psi(C)=0$. Hence $\Psi$ is suitable.
		\item Let $r= \#E(R)$. The reduction is straightforward: Given $G$ with $k$ edge-colours for which we wish to compute $\#\coledgesubs{\Psi,k}{G}$, we construct $G'$ as follows: We add a disjoint copy of $R$ to $G$ and colour the edges of $R$ arbitrarily with $r$ fresh colours, yielding a $k+r$-edge-coloured graph $G'$ of size $|R|+|G|\in O(|G|)$ --- recall that $|R|$ is a constant. Then every $k+r$-edge-colourful subset of edges in $G'$ consists precisely of all edges of $R$ and a $k$-edge-colourful subset of edges in $G$. By definition of $\Psi$, we immediately obtain that
		\[\#\coledgesubs{\Psi,k}{G} = \#\coledgesubs{\Phi,k+r}{G'} \,,\]
		which completes the reduction.
		\item The construction used in the previous case applies to counting modulo $p$ as well.
	\end{enumerate}
With all cases verified, the proof is concluded.
\end{proof}

In the next part, we will consider properties of bounded \emph{wedge-number}: A \emph{wedge}, denoted by $P_2$, is a path with two edges, and a $k$\emph{-wedge-packing}, denoted by $kP_2$, is the disjoint union of $k$ wedges. We say that a graph property $\Phi$ has \emph{bounded wedge-number} if there exists a constant $d$ such that $\Phi(kP_2)=0$ for all $k\geq d$.

In what follows, we write $\#\colmatch$ for the problem of counting edge-colourful $k$-matchings, that is, on input a graph $G$ with $k$ edge-colours, the goal is to compute the number of $k$-matchings in $G$ that contain each colour precisely once. Similarly, given a prime~$p$, we write $\#_p\colmatch$ for the problem of counting edge-colourful $k$-matchings modulo~$p$. The problem $\#\colmatch$ is known to be hard~\cite[Theorem III.1]{CurticapeanM14}.\footnote{To be precise, \cite[Theorem III.1]{CurticapeanM14} establishes hardness if $G$ has \emph{at least} $k$ edge-colours. The case of $G$ having \emph{precisely} $k$ edge-colours is shown to be hard in~\cite[Section 5.2]{Curticapean15}.} Furthermore, $\#_p\colmatch$ is hard for each prime $p\geq 3$~\cite{CurticapeanDH21}. Observe that $\#\colmatch$ is equivalent to the problem $\#\coledgesubsprob(\Psi)$ for the property $\Psi$ of excluding $P_2$ as a minor; the same holds if counting is done modulo $p$. 

\begin{lemma}\label{lem:minor_trick_wedges}
Let $\Phi$ be a minor-closed graph property. If $\Phi$ has unbounded matching number, but bounded wedge-number, then
\[\#\colmatch \fptred \#\coledgesubsprob(\Phi) \text{ and } \#_p\colmatch \fptred \#_p\coledgesubsprob(\Phi)\,,\]
for each prime $p$.
For both reductions, on input $G$ and $k$, every oracle query $(G',k')$ satisfies $|G'|\in O(|G|)$ and $k' \in O(k)$.
\end{lemma}
\begin{proof}
    We only prove $\#\colmatch \fptred \#\coledgesubsprob(\Phi)$; the same construction applies to modular counting.
    Since $\Phi$ has bounded wedge-number, we have $\Phi(dP_2)=0$ for some non-negative integer $d$.
    Let $s$ be the minimum non-negative integer such that
    \[\exists b\geq 0: \Phi(sP_2 +M_b) =0 \,.\]
    We have that $s\leq d$ (for $b=0$), and observe further that $s>0$, as $\Phi$ has unbounded matching number and is minor-closed.
    Let furthermore $b(s)$ be the minimum over all $b$ such that $\Phi(sP_2 +M_b) =0$. This enables us to construct the reduction as follows:
    
    Let $G$ be the input of $\#\colmatch$, that is, $G$ has $k$ edge-colours and the goal is to compute the number of edge-colourful $k$-matchings in $G$. 
    
    We set $G' := G+ (s-1)P_2 + M_{b(s)}$ and colour $(s-1)P_2 + M_{b(s)}$ with $2(s-1) + b(s)= \#E((s-1)P_2 + M_{b(s)})$ fresh colours. Setting $k' := k+2(s-1)+b(s)$ we observe that every $k'$-edge-colourful set $A'$ of edges of $G'$ decomposes into \[A' = A ~\dot\cup~ E((s-1)P_2 + M_{b(s)})\,,\] where $A$ is a $k$-edge-colourful set of edges in $G$.
    We claim that the number of edge-colourful $k$-matchings in $G$ is equal to $\#\coledgesubs{\Phi,k'}{G'}$.
    
    To verify the latter claim, we prove that $\Phi(G'[A'])=1$ if and only if~$A$ is a matching: If~$A$ is a matching, then $\Phi(G'[A'])=\Phi((s-1)P_2 + M_{k+b(s)})=1$, by our choice of $s$. 
    If $A$ is not a matching, then $G[A]$ contains a single wedge $P_2$ as a subgraph. Consequently $G'[A']$ contains $sP_2+ M_{b(s)}$ as a minor. Then $\Phi(G[A'])=0$ by our choice of $b(s)$.
    
    Our reduction thus computes $G'$ and returns $\#\coledgesubs{\Phi,k'}{G'}$ by querying the oracle. Since $s$, and thus $b(s)$ are fixed and independent of the input, we additionally obtain the desired conditions on the size of $G'$ and $k'$, which concludes the proof.
\end{proof}

\subsubsection{The Full Classification}

We start by establishing hardness for minor-closed properties that are both, suitable and of unbounded wedge-number.
\begin{lemma}\label{lem:obstruction}
Let $\Phi$ be a minor-closed graph property. If $\Phi$ is suitable and of unbounded wedge-number, then $\Phi$ has an obstruction and a $3$-obstruction.
\end{lemma}
\begin{proof}
We show that the family $\mathcal{G}=\{G_1,G_2,\dots\}$ of Theorem~\ref{thm:good_expanders} is a $3$-obstruction for $\Phi$ (since a $3$-obstruction is, by definition, also an obstruction, this proves the lemma). 

Since $\Phi$ is suitable, there exists $t$ such that $\Phi(C_t)=0$. Now choose any $i$ such that the orders of the groups $H^1_i$, $H^2_i$, $H^3_i$, and $H^4_i$ (from Theorem~\ref{thm:good_expanders}) are at least $t$. 
We will show that for each such choice of~$i$, the indicator $a(\Phi,G_i)$ is non-zero. Since, by Theorem~\ref{thm:good_expanders}, the orders of the four groups are unbounded, the indicator will then be non-zero infinitely often, proving that $\mathcal{G}$ is an obstruction.

Recall that
\[ a(\Phi,G_i) := \sum_{\sigma \in \mathcal{L}(\Phi,G_i)} ~\prod_{v \in V(G_i)} (-1)^{|\sigma_v|-1}\cdot (|\sigma_v|-1)! \,.\]
Recall further that $\mathcal{L}(\Phi,G_i)$ is the set of fractures $\rho$ of $G_i$ such that $\fracture{G_i}{\rho}$ satisfies $\Phi$. 

The graph $G_i$ is $4$-regular, and every vertex $v\in V(G_i)$ corresponds to a coset of the quotient group~$\Gamma/N_i$. Moreover, every $v$ is adjacent to $vv_0$, $vv_0^{-1}$, $vv_1$ and $vv_1^{-1}$, where $v_0$ and $v_1$ are the generators of $\Gamma/N_i$. It will be convenient to label the edges incident to $v$ by
\[ \tr = \{v,v v_0\}, \tl =\{v, vv_0^{-1}\}, \tu = \{v,v v_1\}, \td =\{v, vv_1^{-1}\} \]
Thus a fracture of $G_i$ is a tuple $\rho=(\rho_v)_{v\in V(G_i)}$ of partitions of the set $\{\tr,\tl,\tu, \td\}$. 

Similarly as in~\cite[Section 4]{RothSW20unpub}, we observe that the quotient group $\Gamma/N_i$ acts on the graph $G_i$ by setting $g\vdash v:= gv$ for each $g\in \Gamma/N_i$ and $v\in V(G_i) =\Gamma/N_i$. Moreover, this action is transitive and for each $g\in \Gamma/N_i$, the function $g\vdash \star$ is an automorphism of $G_i$.

The action $\vdash$ extends to an action $\Vdash$ of $\Gamma/N_i$ on the set $\mathcal{L}(\Phi,H)$: Given $g\in \Gamma/N_i$ and $\rho\in \mathcal{L}(\Phi,H)$, the fracture $g \Vdash \rho$ is obtained from $\rho$ by permuting its entries according to the automorphism $g\vdash \star$.
This action is well-defined since $\fracture{H}{\rho}$ is isomorphic to $\fracture{H}{(g \Vdash \rho)}$, see Section 4.2 of~\cite{RothSW20unpub}. Thus $\fracture{H}{(g \Vdash \rho)}$ satisfies $\Phi$ if and only if $\fracture{H}{\rho}$ does. In particular, setting
\[ f(\sigma):= \prod_{v \in V(G_i)} (-1)^{|\sigma_v|-1}\cdot (|\sigma_v|-1)! \,,\]
we have that $f(\sigma)=f(\rho)$ whenever $\sigma$ and $\rho$ are in the same orbit. 

This enables us to rewrite
\[ a(\Phi,G_i) := \sum_{[\rho]} \#[\rho]\cdot f(\rho) \,,\]
where the sum is over all orbits $[\rho]$ of the action. We will proceed by considering $a(\Phi,G_i)\mod 3$. Recall that the size of any orbit must divide the order of the group. Since $\Gamma/N_i$ is a $3$-group, only orbits of size~$1$, i.e., fixed-points, survive modulo $3$. Due to transitivity of the group action on the vertices, the only fixed-points are fractures~$\rho$ for which all~$\rho_v$ are equal. In what follows, we will thus abuse notation and write e.g.\ $\rho=\{ \{\tr,\tl\}, \{\tu\},\{\td\}\}$ for the fixed-point $\rho$ in which $\rho_v = \{ \{\tr,\tl\}, \{\tu\},\{\td\}\}$ for all $v\in V(G_i)$.

We will analyse the contribution to $a(\Phi,G_i)$ (modulo $3$) of any possible fixed-points in the subsequent series of claims.

\begin{claim}
If $\tl$ and $\tr$ are in the same block of $\rho$, then $\rho\notin \mathcal{L}(\Phi,G_i)$.
\end{claim}
\begin{claimproof}
If $\tl$ and $\tr$ are in the same block, then $\fracture{G_i}{\rho}$ contains the (simple) cycle $C_b$ given by
\[ e \rightarrow v_0 \rightarrow v_0^2 \rightarrow \dots \rightarrow  v_0^{b-1} \rightarrow e\,,\]
where $b$ is the order of $v_0$ (and thus equal to the order of $H^1_i$). By our choice of $i$, we have $b\geq t$. Consequently $C_t \prec \fracture{G_i}{\rho}$. Since $\Phi(C_t)=0$ and $\Phi$ is minor-closed, we conclude that $\Phi(\fracture{G_i}{\rho})=0$ and thus $\rho\notin \mathcal{L}(\Phi,G_i)$.
\end{claimproof}
\begin{claim}
If $\tu$ and $\td$ are in the same block of $\rho$, then $\rho\notin \mathcal{L}(\Phi,G_i)$.
\end{claim}
\begin{claimproof}
Analogously to the previous claim; substitute $v_0$ by $v_1$, and $H^1_i$ by $H^2_i$.
\end{claimproof}

\begin{claim}
$\{ \{\tu,\tr\},\{\td,\tl\} \}\notin \mathcal{L}(\Phi,G_i)$.
\end{claim}
\begin{claimproof}
The fractured graph $\fracture{G_i}{\{\{\tu,\tr\},\{\td,\tl\}\}}$ contains the (simple) cycle $C_{2b}$ given by
\[ e \rightarrow v_1^{-1} \rightarrow v_1^{-1}v_0 \rightarrow v_1^{-1}v_0v_1^{-1} \rightarrow (v_1^{-1}v_0)^2 \rightarrow \dots \rightarrow  (v_1^{-1}v_0)^{b-1} \rightarrow e\,,\]
where $b$ is the order of $v_1^{-1}v_0$ (and thus equal to the order of $H^3_i$). By our choice of $i$, we have $2b\geq t$. Consequently $C_t \prec \fracture{G_i}{\rho}$. Since $\Phi(C_t)=0$ and $\Phi$ is minor-closed, we conclude that $\Phi(\fracture{G_i}{\rho})=0$ and thus $\rho\notin \mathcal{L}(\Phi,G_i)$.
\end{claimproof}
\begin{claim}
$\{ \{\tu,\tl\},\{\td,\tr\} \}\notin \mathcal{L}(\Phi,G_i)$.
\end{claim}
\begin{claimproof}
Analogously to the previous claim; substitute $v^{-1}_1$ by $v_1$, and $H^3_i$ by $H^4_i$.
\end{claimproof}

The only partitions of $\{\tu,\td,\tl,\tr\}$ not covered by one of the previous four claims are the finest partition $\bot=\{\{\tu\}, \{\td\},\{\tl\},\{\tr\}\}$, as well as the following four:
\begin{itemize}
    \item $\rho_1 = \{\{\tu,\tl\},\{\td\},\{\tr\} \}$ 
    \item $\rho_2 = \{\{\tu,\tr\},\{\td\},\{\tl\} \}$ 
    \item $\rho_3 = \{\{\td,\tl\},\{\tu\},\{\tr\} \}$ 
    \item $\rho_4 = \{\{\td,\tr\},\{\tu\},\{\tl\} \}$
\end{itemize}
First note that $\bot\in \mathcal{L}(\Phi,G_i)$ since $\Phi$ has unbounded wedge-number, the fractured graph of $\bot$ is a matching and every matching is a minor of a sufficiently large wedge-packing. However, the contribution from $\bot$ still vanishes, since we have
\[f(\bot) = \prod_{v \in V(G_i)} (-1)^{4-1}\cdot (4-1)! = 0 \mod 3 \,.\]

Finally, we observe that $\fracture{G_i}{\rho_j}$ is a wedge-packing for $j = 1, \ldots, 4$. Since $\Phi$ is minor-closed and of unbounded wedge-number, we obtain that $\Phi(\fracture{G_i}{\rho_j})=1$ and thus $\fracture{G_i}{\rho_j}\in \mathcal{L}(\Phi,G_i)$ for $j = 1, \ldots, 4$. The latter implies
\[a(\Phi,G_i) = 4 \cdot \prod_{v \in V(G_i)} (-1)^{3-1}\cdot (3-1)! = 2^{|V(G_i)|} \mod 3\,,\]
concluding the proof, since $2^{|V(G_i)|}\neq 0 \mod 3$.
\end{proof}

We are now able to prove Theorem~\ref{thm:minclose-hard}, that is, we show that all $k$-edge subgraphs satisfying a minor-closed property $\Phi$ can be counted efficiently if and only if $\Phi$ is either trivial or of bounded matching number.

%\begin{theorem}\label{thm:main_minor_closed_exact}
%Let $\Phi$ be a minor-closed graph property. If $\Phi$ is trivially true or of bounded matching number, then $\#\edgesubsprob(\Phi)$ is fixed-parameter tractable. 
%Otherwise, $\#\edgesubsprob(\Phi)$ is $\#\W{1}$-hard and, assuming ETH, cannot be solved in time \[f(k)\cdot |G|^{o(k/\log k)}\] 
%for any function $f$.
%\end{theorem}
\begin{proof}[Proof of Theorem~\ref{thm:minclose-hard}]
Note that every minor-closed graph property is characterised by a finite set of forbidden minors by the Robertson-Seymour-Theorem~\cite{RobertsonS04}. In particular, minor-closed graph properties are always computable.
The (fixed-parameter) tractability part of the classification was shown in~\cite[Main Theorem 1]{RothSW20unpub}. 

Thus assume that $\Phi$ is not trivially true and of unbounded matching number. Then $\mathcal{F}(\Phi)$ is non-empty. If $\mathcal{F}_2(\Phi)$ is empty, then the result follows again from~\cite{RothSW20unpub}. Hence consider the case $\mathcal{F}_2(\Phi)\neq \emptyset$. 

We will show that the colourful version $\#\coledgesubsprob(\Phi)$ satisfies the desired lower bound; a (tight) reduction to the uncoloured version $\#\edgesubsprob(\Phi)$ is obtained by a simple application of the inclusion-exclusion principle; details are given in~\cite[Lemma 3.7]{RothSW20unpub}.

If $\Phi$ has bounded wedge-number, then we apply Lemma~\ref{lem:minor_trick_wedges} and obtain a (tight) reduction 
\[ \#\colmatch \fptred \#\coledgesubsprob(\Phi)  \,.\]
Since the desired lower bounds hold for $\#\colmatch$ (see~\cite[Section 5.2]{Curticapean15} and~\cite[Theorem III.1]{CurticapeanM14}), this case is concluded.

If $\Phi$ is of unbounded wedge-number and suitable, then $\Phi$ has an obstruction by Lemma~\ref{lem:obstruction}. Thus we obtain hardness by Lemma~\ref{lem:monotonicity}.

If $\Phi$ is unsuitable, then we apply Lemma~\ref{lem:minor_trick} which yields a (tight) reduction 
\[ \#\coledgesubsprob(\Psi) \fptred \#\coledgesubsprob(\Phi) \,,\]
for a minor-closed and suitable property $\Psi$ of unbounded matching number. Now depending on whether $\Psi$ has bounded wedge-number we, again, either obtain hardness by Lemma~\ref{lem:minor_trick_wedges} or by the combination of Lemma~\ref{lem:obstruction} and Lemma~\ref{lem:monotonicity}.
\end{proof}

We observe that our treatment of minor-closed properties $\Phi$ with $\mathcal{F}_2(\Phi)\neq \emptyset$ via the $3$-group Cayley graph expanders allows us to prove a classification for such properties even if counting is done modulo $3$:
\begin{theorem}\label{thm:main_minor_closed_mod3}
Let $\Phi$ be a minor-closed graph property with a forbidden minor of degree at most $2$. If~$\Phi$ has bounded matching number, then $\#_3\edgesubsprob(\Phi)$ is fixed-parameter tractable. 

Otherwise, $\#_3\edgesubsprob(\Phi)$ is $\threeW{1}$-hard and, assuming rETH, cannot be solved in time \[f(k)\cdot |G|^{o(k/\log k)}\] 
for any function $f$.
\end{theorem}
\begin{proof}
Let us first point out that $\Phi$ is not trivially true since the set of forbidden minors is non-empty. More precisely, we have $\mathcal{F}_2(\Phi) \neq \emptyset$.

Consequently, we can follow the same strategy as in the proof of Theorem~\ref{thm:minclose-hard} to prove that $\#_3\coledgesubsprob(\Phi)$ satisfies the desired lower bound: Lemma~\ref{lem:monotonicity}, \ref{lem:minor_trick}, \ref{lem:minor_trick_wedges}, and~\ref{lem:obstruction} all apply to counting modulo $3$, so the argument remains intact. The only exception is that we need to rely on~\cite{CurticapeanDH21} for the lower bounds on $\#_3\colmatch$.

Finally, we need to provide a reduction from $\#_3\coledgesubsprob(\Phi)$ to $\#_3\edgesubsprob(\Phi)$. However, the reduction for exact counting
based on inclusion-exclusion as given in~\cite[Lemma~3.7]{RothSW20unpub} applies if arithmetic is done modulo $3$ as well. This concludes the proof.
\end{proof}
Note that the previous result does not hold in case of counting modulo $2$: Let $\Phi$ be the (minor-closed) property defined by the exclusion of $P_2$, the path of length $2$, as a minor. Observe that $\Phi$ is non-trivial and of unbounded matching number. However, $\#_2\edgesubsprob(\Phi)$ is equivalent to the problem of counting $k$-matchings modulo $2$, which was recently shown to be fixed-parameter tractable~\cite{CurticapeanDH21}. 

Let us emphasize some concrete and novel conditional lower bounds implied by Theorem~\ref{thm:minclose-hard} and~\ref{thm:main_minor_closed_mod3}. 

\paragraph*{Forests and Linear Forests}
It was shown in~\cite{BrandR17} that counting $k$-forests in a graph $G$ is $\#\W{1}$-hard. Implicitly, the proof of the latter result also yields a conditional lower bound of $f(k)\cdot |G|^{o(k/\log k)}$ under the Exponential Time Hypothesis. Our classification subsumes this result and additionally applies for linear forests (disjoint unions of paths) and for counting modulo $3$:
\begin{corollary}\label{cor:forest_exact_and_3}
The problems of counting forests and linear forests with $k$ edges in a graph~$G$ are $\#\W{1}$-hard and, assuming ETH, cannot be solved in time $f(k)\cdot |G|^{o(k/\log k)}$ for any function $f$.

In case of counting modulo $3$, both problems are $\threeW{1}$-hard and admit the same conditional lower bound under rETH.
\end{corollary}
\begin{proof}
Both properties, being a forest and being a linear forest, are minor-closed, have the triangle $K_3$ as forbidden minor, and are of unbounded matching number. The corollary hence holds by Theorem~\ref{thm:minclose-hard} and~\ref{thm:main_minor_closed_mod3}.
\end{proof}

We remark that unfolding the proof to the case of counting (linear) forests modulo $3$, shows that the most useful property of the Cayley graph expanders we used was the fact that they have relatively low degree ($4$). In Section~\ref{sec:mod_count_main}, we will use our general construction of $p$-group Cayley graph expanders of degree $2(p-2)$, given by Theorem~\ref{thm:savegeneratorsforpquotients}, to establish hardness of counting $k$-forests modulo~$p$ for any prime $p>2$.

\paragraph*{Subgraphs of Bounded Tree-Depth}
A further application is given by the property of having bounded tree-depth, a parameter that measures how ``star-like'' a graph is and whose boundedness is a stronger restriction than boundedness of treewidth. In particular, a variety of parameterized problems that remain hard on bounded treewidth graphs become fixed-parameter tractable on bounded tree-depth graphs, such as $H\textsc{-Colouring Reachability}$~\cite{Wrochna18}, and $1\textsc{-Planar Drawing}$~\cite{BannisterCE18}. We will only rely on a forbidden minor characterisation of bounded treedepth and refer the reader to~\cite[Chapter 6]{NesetrilO12} for a detailed exposition.

\begin{corollary}\label{cor:treedepth}
Let $b \geq 2$ be a fixed integer. The problem of counting $k$-edge subgraphs of tree-depth at most $b$ in a graph $G$ is $\#\W{1}$-hard and, assuming ETH, cannot be solved in time $f(k)\cdot |G|^{o(k/\log k)}$ for any function $f$.

In case of counting modulo $3$, the problem is $\threeW{1}$-hard and admits the same conditional lower bound under rETH.
\end{corollary}
\begin{proof}
The property of having tree-depth at most $b$ is minor-closed and has a path as forbidden minor~\cite{GiannopoulouT09}. Additionally, since $b\geq 2$, the property has unbounded matching number. The corollary hence holds by Theorem~\ref{thm:minclose-hard} and~\ref{thm:main_minor_closed_mod3}.
\end{proof}

\paragraph*{Subgraphs with Small Colin de Verdi\`ere Invariant}
Around 1990, Yves Colin de Verdi\`ere introduced a spectral graph invariant which is the maximal multiplicity of the second smallest eigenvalue of a family $\mathcal{M}$ of symmetric Schr\"odinger operators  satisfying a certain transversality condition (see (M3) below). More precisely, for a given graph $G$ with $n$ vertices, a real valued $n \times n$ matrix $M$ is called a ``Schr\"odinger operator'' on $G$ if we have \[ M_{ij} \begin{cases} < 0, & \text{if $\{i,j\}$ is an edge of $G$;} \\ = 0, & \text{if $\{i,j\}$ is not an edge and $i\neq j$.} \end{cases} \]
The family $\mathcal{M}$ consists of all symmetric matrices $M$ with the following three properties:
\begin{description}
    \item[(M1)] $M$ is a symmetric Schr\"odinger operator; 
    \item[(M2)] $M$ has exactly one negative eigenvalue which is simple;
    \item[(M3)] if $X$ is a symmetric $n \times n$ matrix such that $MX = 0$, and $X_{ij}=0$ whenever $i=j$ or $\{i,j\}$ is an edge of $G$, then $X=0$. 
\end{description}
Colin de Verdi\`ere's invariant is then given as 
\[ \mu(G) = \max_{M \in \mathcal{M}} \dim {\rm ker}(M)\,. \]
This spectral invariant is minor-monotone. This means, in particular, that $\mu(G) \le n$ has a characterization via a finite set of forbidden minors, by the Robertson-Seymour-Theorem~\cite{RobertsonS04}.  
Moreover, for integers $n \le 4$, graphs satisfying $\mu(G) \le n$ have interesting combinatorial properties: 
\begin{itemize}
    \item[(a)] $\mu(G) \le 1$ holds if and only if $G$ is linear forest; 
    \item[(b)] $\mu(G) \le 2$ holds if and only if $G$ is outerplanar; 
    \item[(c)] $\mu(G) \le 3$ holds if and only if $G$ is planar (that is, $G$ does not have $K_{3,3}$ or $K_5$ as minors); 
    \item[(d)] $\mu(G) \le 4$ if and only if $G$ is linklessly embeddable (that is, $G$ does not have any of the seven graphs in the Peterson family as a minor).
\end{itemize}
For more information on this invariant see, e.g., \cite{vdHLS99,CdV98,CdV04}. 

Our results yield a complete picture for counting small subgraphs with bounded Colin de Verdi\`ere invariant. If the bound is $0$, then the problem is trivial. In all other cases the problem becomes intractable:

\begin{corollary}\label{cor:CDVinv}
Let $b\geq 1$ be a fixed positive integer and let $\Phi_b(H)=1$ if and only if $\mu(H)\leq b$. Then $\#\edgesubsprob(\Phi_b)$ is $\#\W{1}$-hard and, assuming ETH, cannot be solved in time
\[f(k)\cdot |G|^{o(k/\log k)}\] 
for any function $f$.
\end{corollary}
\begin{proof}
The corollary is an immediate consequence of Theorem~\ref{thm:minclose-hard} and the fact that $\Phi_b$ is of unbounded matching number for each $b\geq 1$, since a matching is a linear forest.
\end{proof}

\subsection{Modular Counting of Forests and Matroid Bases}\label{sec:mod_count_main}
In this section, we extend our hardness result for counting $k$-forests modulo $3$ to counting modulo $p$ for any prime $p\geq 3$. The crucial ingredients of our proof are the $p$-group Cayley graph expanders with $p-2$ generators constructed in Theorem~\ref{thm:savegeneratorsforpquotients}. In particular, the relatively low degree of $2(p-2)$ provides us with additional control over their coefficient in the colour-prescribed homomorphism basis and ultimately allows us to establish those Cayley graph expanders as $p$-obstructions for counting acyclic $k$-edge subgraphs. The results for modular counting of matroid bases will then follow by deterministic matroid truncation.

Let $\Gamma$ be a finite group, let $S_0 \subseteq \Gamma$ be a set of $m$ generators and $S = \{g^{\pm 1} : g \in S_0 \} \subseteq \Gamma$ be the associated symmetric set of $2m$ generators.\footnote{In particular, we require that for $g \in S_0$ we have $g^{-1} \notin S_0$.}  Let $G = \mathcal{C}(\Gamma, S)$ be the associated Cayley graph. Similarly as in the proof of Lemma~\ref{lem:obstruction}, we have an action of $\Gamma$ on $G$ which extends to the fractures of $G$. We write $\mathcal{L}(\Phi,G)^\Gamma$ for the set of fixed-points, i.e., fractures $\sigma$ which are invariant under this action. As seen before, we can interpret $\sigma$ as a partition of the set $S$ into blocks.

We begin this section with a general construction that will appear prominently later: Let $m \geq 1$ and let $S$ and $S_0$ be as above. Then for any set partition $\sigma$ of $S$, we define a graph $\basegraph{\sigma}$. It has a vertex $w^B$ for each block $B$ of $\sigma$ and its set of (multi)edges is given by
\begin{equation} \label{eqn:edgesbasegraph}
E(\basegraph{\sigma}) = \left\{ \{w^{B}, w^{B'} \} : \text{ one multiedge for each } g \in S_0 \text{ such that } g \in B, g^{-1} \in B' \right\}\,.
\end{equation}
Note that we see $\basegraph{\sigma}$ as a graph with possible loops and possible multiedges. In particular, the graph $\basegraph{\sigma}$ has precisely $m$ edges. An alternative construction of $\basegraph{\sigma}$ is by taking the matching $M_{m}$ on the vertex set $S$ defined by the involution $s \mapsto s^{-1}$ on $S$ and identify all vertices in the same block of the partition $\sigma$. 

The following two lemmas will be needed below:

\begin{lemma} \label{Lem:countsigmainisomclass}
Let $m \geq 1$ and let $S$ be a finite set with $2m$ elements. Given any (simple) graph $T$ with $m$ edges, the number of set partitions $\sigma$ of $S$ such that $\basegraph{\sigma}$ is isomorphic to $T$ is given by $2^m m! / |\auts{T}|$.
\end{lemma}
\begin{proof}
Given the data above, consider the two sets
\begin{align*}
M_0 &= \left\{\sigma : \sigma\text{ partition of }S \text{ such that } \basegraph{\sigma} \cong T \right \}\,,\\
M &= \left\{(\sigma, \varphi) : \sigma\text{ partition of }S,\text{ and } \varphi: \basegraph{\sigma} \xrightarrow{\sim} T \text{ isomorphism}\right\}\,.
\end{align*}
In the lemma we want to count the number of elements of $M_0$, but as an auxiliary set we use $M$, which explicitly records the data of the isomorphism $\varphi: \basegraph{\sigma} \xrightarrow{\sim} T $. First, we note that the automorphism group $\auts{T}$ acts on $M$, where $\eta \in \auts{T}$ sends the pair $(\sigma, \varphi)$ to $(\sigma, \eta \circ \varphi)$. We claim that the action is free: since $\varphi$ is an isomorphism, the equality $\eta \circ \varphi = \varphi$ implies that $\eta$ is the identity. Thus the orbits of the action all have cardinality $|\auts{T}|$. On the other hand, we observe that the map
\[M \to M_0, \quad (\sigma, \varphi) \mapsto \sigma\]
is surjective and the fibres of this map are precisely the orbits of the above action of $\auts{T}$. Indeed, the surjectivity is clear from the definition and given $(\sigma, \varphi)$ and $(\sigma, \varphi')$ in the fibre of $\sigma$, we have that the automorphism $\eta = \varphi' \circ \varphi^{-1}$ of $T$  satisfies $\eta \cdot (\sigma, \varphi) = (\sigma, \varphi')$.
Combining these two observations we see
\[
|M| = |M_0| \cdot |\auts{T}|.
\]
Thus to conclude we need to show that $|M| = 2^m m!$. To do this observe that we can identify the elements of $S$ with the vertices of the matching $M_m$ in such a way that there is an edge $\{g, g^{-1}\}$ for each $g \in S_0$. 
Given a partition $\sigma$ of $S$, note that we can see $\basegraph{\sigma}$ as the quotient of $M_m$ obtained by identifying the vertices belonging to the blocks of $\sigma$. In particular there is a canonical, well-defined quotient map $q_\sigma : M_m \to \basegraph{\sigma}$.
Let $\mathsf{Sur}(M_m, T)$ be the set of surjections from $M_m$ to the graph $T$ (where we mean graph homomorphisms that are surjective, hence bijective, on the set of edges). Then we have a map
\[
G : \mathsf{Sur}(M_m, T) \to M, \quad \psi \mapsto (\sigma = \{\psi^{-1}(w) : w \in V(T)\}, \overline{\psi} )\,,
\]
where $\overline{\psi} : \basegraph{\sigma} \to T$ is the unique map such that $\psi = \overline{\psi} \circ q_\sigma$. A short computation shows that $G$ is a bijection with inverse given by
\[
G^{-1} : M \to \mathsf{Sur}(M_m, T), \quad (\sigma, \varphi) \mapsto \varphi \circ q_\sigma\,.
\]
Thus the proof is finished once we show that $|\mathsf{Sur}(M_m, T)| = 2^m m!$. But this is easy to see: to specify a surjection from $M_m$ to $T$ we exactly have to give a bijection from the set of edges of $M_m$ to the $m$ edges of $T$ (for which we have $m!$ possibilities) and for each of these edges we have two choices of orientation in our map (giving the factor of $2^m$), because $T$ is a simple graph.
\end{proof}

\begin{lemma} \label{Lem:Cayleyformula}
 Given $n \geq 1$ we have
 \begin{equation}
     \sum_{\substack{T \text{ tree on}\\\text{$n$ vertices}}} \frac{1}{|\auts{T}|} = \frac{n^{n-2}}{n!}\,,
 \end{equation}
 where the sum goes over \emph{isomorphism classes} of trees $T$.
\end{lemma}
\begin{proof}
For the proof we use Cayley's formula: the number of labeled trees $\widehat{T}$ on $n$ vertices is given by $n^{n-2}$. The natural action of the symmetric group $S_n$ on the $n$ vertices induces an action on the set of labeled trees $\widehat T$ and the stabilizer of such a tree is equal to its automorphism group. Moreover, two labeled trees $\widehat{T}_1, \widehat{T}_2$ are in the same orbit if and only if their underlying unlabeled graphs are isomorphic. Thus 
\begin{align*}
 n^{n-2} =  \sum_{\substack{T \textup{ tree on}\\\textup{$n$ vertices}}} \frac{|S_n|}{|\auts{T}|} = n! \cdot \sum_{\substack{T \textup{ tree on}\\\textup{$n$ vertices}}} \frac{1}{|\auts{T}|}
\end{align*}
by the Orbit-Stabilizer-Theorem.
\end{proof}

Now recall that $G$ is the Cayley graph of $\Gamma$ and $S$. We show that the fractured graph $\fracture{G}{\sigma}$ is a forest if and only if $\basegraph{\sigma}$ is. 

\begin{lemma} \label{Lem:coverforest}
There exists a natural graph homomorphism 
\begin{equation} \label{eqn:fracgraphcovering}
    \Psi: \fracture{G}{\sigma} \to \basegraph{\sigma},\ v^B \mapsto w^B
\end{equation}
which is surjective and a local isomorphism (i.e. the edges incident to $v^B$ map bijectively to the edges at $w^B$). Moreover, the graph $\fracture{G}{\sigma}$ is a forest if and only if $\basegraph{\sigma}$ is.
\end{lemma}
\begin{proof}
First we show that $\Psi$ is a well-defined graph homomorphism. Given $v \in V(G) = \Gamma$, the edges of $G$ incident to $v$ are given by $\{v, vg\}$ for $g \in S$. Then, given a block $B$ of $\sigma$, the edges incident to $v^B$ inside $\fracture{G}{\sigma}$ are in bijection with $B$ and given by 
\[
\{v^B, (vg)^{B'}\} \text{ for }g \in B, g^{-1} \in B'\,.
\]
Comparing to the edges \eqref{eqn:edgesbasegraph} of $\basegraph{\sigma}$ we see that $\Psi$ is not only a well-defined graph homomorphism but in fact, as claimed above, a local isomorphism. Finally, the surjectivity (both on vertices and edges) is also clear.

To see the last claim, first assume that $\fracture{G}{\sigma}$ is \emph{not} a forest, and let $C$ be a  a circular walk without backtracking\footnote{A walk is a sequence of vertices $w_0, w_1, \ldots, w_\ell$ such that $w_{i-1}$ and $w_{i}$ are connected by an edge for $i=1, \ldots, \ell$. We say that the walk is without backtracking if there does not exist an $i$ such that $w_{i-1}=w_{i+1}$.} inside $\fracture{G}{\sigma}$. Under the homomorphism $\Psi$ it maps to a circular walk $\Psi(C)$ and, since $\Psi$ is a local isomorphism, there is again no backtracking in $\Psi(C)$. Thus, the graph $\basegraph{\sigma}$ is not a forest.

Conversely let $C'$ in $\basegraph{\sigma}$ be a circular walk without backtracking starting at some vertex $w^B$. Choose a vertex $v_0^B$ in the preimage of $w^B$ under $\Psi$ and let $C$ in $\fracture{G}{\sigma}$ be the unique lift of the walk $C'$. By this we mean that we start at $v_0^B$ and for the first edge taken by the path $C'$ from $w^B$, we take the unique edge incident to $v_0^B$ mapping to it. Iterating the process for the subsequent edges of $C'$ we obtain the walk $C$, which terminates at some vertex $v_1^B$. The whole process can now itself be iterated: we continue the walk $C$ by concatenating it with the unique lift of $C'$ starting this time at vertex $v_1^B$, terminating at $v_2^B$, and we can continue from there. In this way we can obtain an arbitrarily long walk in the graph $\fracture{G}{\sigma}$. But note that this walk involves no backtracking (since the original walk $C'$ in $\basegraph{\sigma}$ had no backtracking and $\Psi$ is a local isomorphism). Thus, since the graph $\fracture{G}{\sigma}$ is finite, the infinite walk must contain a circular sub-walk which, as seen before, involves no backtracking. Thus $\fracture{G}{\sigma}$ is not a forest.
\end{proof}

\begin{proposition} \label{Prop:forestpobstruction}
Let $p \geq 5$ be a prime, $\Gamma$ a finite $p$-group, $S_0 \subseteq \Gamma$ a set of $q=p-2$ generators such that $S = \{g^{\pm 1} : g \in S_0\}$ has $2q$ elements. Then for the property $\Phi$ of being a forest, we have $a(\Phi,G) \neq 0 \mod p$. 
\end{proposition}
\begin{proof}
Recall that the number $a(\Phi,G)$ is given by
\[ 
a(\Phi,G) = \sum_{\sigma \in \mathcal{L}(\Phi,G)} ~\prod_{v \in V(G)} (-1)^{|\sigma_v|-1}\cdot (|\sigma_v|-1)! \,.
\]
As before, when evaluating modulo $p$, we can reduce to the fractures $\sigma \in \mathcal{L}(\Phi,G)^\Gamma$ invariant under the action of $\Gamma$. Such $\sigma$ can be interpreted as partitions of the set $S$. Rewriting the above formula we have
\[
a(\Phi,G) \equiv \sum_{\sigma \in \mathcal{L}(\Phi,G)^\Gamma} ~\left( (-1)^{|\sigma|-1}\cdot (|\sigma|-1)! \right)^{|V(G)|} \mod p \,.
\]
Since $|V(G)| = |\Gamma|$ is a power of $p$ by the assumption that $\Gamma$ is a $p$-group, by Fermat's little theorem we have $u^{|V(G)|} \equiv u \mod p$ for all integers $u$ so that we can remove the exponent $|V(G)|$ in the formula above. Looking at the index set of the sum, we note that for $\sigma \in \mathcal{L}(\Phi,G)^\Gamma$ with $|\sigma|>p$ we have that $p$ divides the term $(|\sigma|-1)!$ inside the sum, so that the corresponding summands vanish modulo $p$. On the other hand, for the property $\Phi$ of being a forest, we have by \cref{Lem:coverforest} that a partition $\sigma$ of $S$ is contained in $\mathcal{L}(\Phi,G)^\Gamma$ if and only if $\basegraph{\sigma}$ is a forest. 

The graph $\basegraph{\sigma}$ has $|S_0|=p-2$ edges and $|\sigma|$ many vertices, where, as seen above, we can assume $|\sigma| \leq p$. If $\basegraph{\sigma}$ is a forest, then the number of trees it contains (the connected components) is its Euler characteristic $|V(\basegraph{\sigma})| - |E(\basegraph{\sigma})| = |\sigma| - (p-2)$. For $|\sigma|\leq p-2$, the graph $\basegraph{\sigma}$ has at least as many edges as it has vertices, and thus it can never be a forest. We are left with two cases:\footnote{As a reassurance to the reader: a priori, the graph $\basegraph{\sigma}$ was allowed to have loops or multi-edges. However, any graph $\basegraph{\sigma}$ having either of those is certainly not a forest (since it has a cycle of length $1$ or $2$ respectively) and thus from here on, only standard (simple) graphs appear.}
\begin{itemize}
    \item Case I : $|V(\basegraph{\sigma})| = |\sigma| = p-1$, which forces $\basegraph{\sigma} \cong T$ to be a tree $T$,
    \item Case II : $|V(\basegraph{\sigma})|= |\sigma| = p$, which forces $\basegraph{\sigma} \cong T_1 + T_2$ to be a union of two trees $T_1, T_2$. We remark for later that, since the total number $p$ of vertices is odd, the two trees $T_1, T_2$ cannot be isomorphic (since one has odd and one has even number of vertices). Therefore we have
    \begin{equation} \label{eqn:T1T2automs}
        |\auts{T_1 + T_2}| = |\auts{T_1}| \cdot |\auts{T_2}|\,.
    \end{equation}
\end{itemize}
Denote by $\mathcal{L}_{\mathrm I}, \mathcal{L}_{\mathrm{II}}$ the set of partitions of $S$ corresponding to the cases I, II above. Then the current status of the calculation is that
\begin{align}
a(\Phi,G) & \equiv \sum_{\sigma \in \mathcal{L}_{\mathrm I} \sqcup \mathcal{L}_{\mathrm{II}}} ~ (-1)^{|\sigma|-1}\cdot (|\sigma|-1)! \mod p \nonumber\\
& \equiv \sum_{\sigma \in \mathcal{L}_{\mathrm I}} ~ (-1)^{p-2}\cdot (p-2)!
+ \sum_{\sigma \in \mathcal{L}_{\mathrm{II}}} ~ (-1)^{p-1}\cdot (p-1)!
\mod p \nonumber\\
& \equiv - |\mathcal{L}_{\mathrm I}| - |\mathcal{L}_{\mathrm{II}}| \mod p
\,. \label{eqn:aPhiGformulaCaseIII}
%& \equiv - |\mathcal{L}_{\mathrm I}| \cdot (p-2)!+ |\mathcal{L}_{\mathrm{II}}| \cdot (p-1)!\equiv (p-1) \cdot |\mathcal{L}_{\mathrm I}| \cdot (p-2)!+ |\mathcal{L}_{\mathrm{II}}| \cdot (p-1)!\mod p \nonumber \\
%& \equiv (p-1)! \cdot (|\mathcal{L}_{\mathrm I}| + |\mathcal{L}_{\mathrm{II}}|) 
\end{align}
where we used that $p$ is odd and, due to Wilson's theorem, we have:
\[
- (p-2)! \equiv (p-1) \cdot (p-2)! \equiv (p-1)! \equiv -1 \mod p \,.
\]

To count the number of elements $\sigma \in \mathcal{L}_{\mathrm I}$, we can group them according to the isomorphism class $T$ of the tree $\basegraph{\sigma}$. Then, by \cref{Lem:countsigmainisomclass} we have
\begin{align*}
    |\mathcal{L}_{\mathrm I}| = 2^{p-2} (p-2)! \sum_{\substack{T \text{ tree on}\\\text{$p-1$ vertices}}} \frac{1}{|\auts{T}|} =  2^{p-2} (p-2)! \frac{(p-1)^{p-3}}{(p-1)!} = 2^{p-2} (p-1)^{p-4} \,.
\end{align*}
Here in the third equality we used \cref{Lem:Cayleyformula}. Plugging into the formula above we compute
\begin{align*}
- |\mathcal{L}_{\mathrm I}| & =  - 2^{p-2} (p-1)^{p-4} \equiv - 2^{p-2} (-1)^{p-4}  \equiv 2^{p-1} \cdot \frac{p+1}{2} \cdot (-1)^{p-3} \equiv \frac{p+1}{2} \mod p\,,
\end{align*}
where in the fourth congruence we used Fermat's little theorem.

We now turn to the sum in Case II. We are counting every forest twice by choosing a numbering $T_1$, $T_2$ for the two trees in the forest. It is important that never $T_1 \cong T_2$ as remarked above, so that for all forests we overcount with the factor $2$. Then we have
\begin{align}
    |\mathcal{L}_{\mathrm{II}}| & = \frac{1}{2} \cdot 2^{p-2} (p-2)! \sum_{j=1}^{p-1} \sum_{\substack{T_1 \text{ tree, } j \text{ vertices,}\\T_2 \text{ tree, } p-j \text{ vertices}}} \frac{1}{|\auts{T_1 + T_2}|}\nonumber \\
    & = 2^{p-3} (p-2)! \sum_{j=1}^{p-1} \left( \sum_{\substack{T_1 \text{ tree, } j \text{ vertices}}} \frac{1}{|\auts{T_1}|}\right) \cdot \left( \sum_{\substack{T_2 \text{ tree, } p-j \text{ vertices}}} \frac{1}{|\auts{T_2}|}\right) \nonumber \\
    & = 2^{p-3} (p-2)! \sum_{j=1}^{p-1} \frac{j^{j-2}}{j!} \cdot \frac{(p-j)^{p-j-2}}{(p-j)!} \nonumber \\
    & = 2^{p-3} \sum_{j=1}^{p-1} \big( j^{j-2} (p-j)^{p-j-3} 
    \cdot \prod_{k=2}^{j} \frac{p-k}{k} \big) \,, \label{eqn:Case2middlestep}
\end{align}
where the first equality uses \cref{Lem:countsigmainisomclass}, the second uses the observation \eqref{eqn:T1T2automs} and the third uses \cref{Lem:Cayleyformula}. In the fourth equation we rearranged the factors of the factorials.

To continue, we observe that the final formula \eqref{eqn:Case2middlestep} for $|\mathcal{L}_{\mathrm{II}}|$ in fact \emph{is well-defined modulo p}. That is, we never divide by a number divisible by~$p$. Thus we are allowed to evaluate and simplify this expression modulo $p$:
\begin{align*}
- |\mathcal{L}_{\mathrm{II}}| & = - 2^{p-3} \sum_{j=1}^{p-1} \big( j^{j-2} (p-j)^{p-j-3} \cdot \prod_{k=2}^{j} \frac{p-k}{k} \big) \mod p \\
& \equiv - \frac{1}{4} 2^{p-1} \sum_{j=1}^{p-1}  j^{j-2} (-j)^{p-j-3} \cdot (-1)^{j-1}  \mod p \\
& \equiv \frac{1}{4} \sum_{j=1}^{p-1}  j^{p-5} (-1)^{p-5} 
\equiv  \frac{1}{4} \sum_{j=1}^{p-1}  j^{p-5}  
%\equiv  \frac{1}{4} \sum_{j=1}^{p-1}  j^{-4}  
\mod p \,.
\end{align*}
To simplify the sum from $j=1$ to $p-1$, we observe that $j$ ranges over $\mathbb{F}_p^\times$. The map $\mathbb{F}_p^\times \to \mathbb{F}_p^\times, j \mapsto 2j$ is a bijection. Thus we have
\begin{equation} \label{eqn:doublingtrick}
16 \sum_{j=1}^{p-1} j^{p-5} \equiv 16 \sum_{j=1}^{p-1} (2j)^{p-5} \equiv 16 \cdot 2^{p-5} \sum_{j=1}^{p-1} j^{p-5} \equiv 2^{p-1} \sum_{j=1}^{p-1} j^{p-5} \equiv \sum_{j=1}^{p-1} j^{p-5} \mod p \,, 
\end{equation}
where the last equality is again Fermat's little theorem.
Subtracting the right hand side of \eqref{eqn:doublingtrick} from the left hand side, we see that that $p$ divides $15 \cdot \sum_{j=1}^{p-1} j^{p-5}$. 
If $p>5$, this implies that $p$ divides $\sum_{j=1}^{p-1} j^{p-5}$ and so by the previous computation we conclude $- |\mathcal{L}_{\mathrm{II}}| \equiv 0 \mod p$. 

We are now ready to plug in our computations in the formula \eqref{eqn:aPhiGformulaCaseIII} and to conclude the proof. In the case $p=5$ we obtain
\[
a(\Phi,G) \equiv  - |\mathcal{L}_{\mathrm I}| - |\mathcal{L}_{\mathrm{II}}| \equiv \frac{5+1}{2} + \frac{1}{4} \cdot 4 \equiv 4 \neq 0 \mod 5\,.
\]
%\[
%a(\Phi,G) = 2^{5-2} + 2^{5-3} (5-2)! (-1) (1^4 + 2^4 + 3^4 +4^4) = 2 \neq 0 \mod 5\,.
%\]
On the other hand, for $p>5$ we saw that the contribution from Case II vanishes and thus we have $a(\Phi,G) = (p+1)/2 \neq 0 \mod p$. In any case we can conclude that $a(\Phi,G)$ does not vanish modulo $p$, finishing the proof.
\end{proof}

We are now able to prove the establish hardness of the problems of modular counting of forests and matroid bases, which we restate for convenience. Recall that we write $\#_p\textsc{Forests}$ for the problem of, given a graph $G$ and a positive integer~$k$, computing the number of forests with $k$ edges in $G$, modulo $p$; the parameterization is given by~$k$. Similarly, we write $\#_p\textsc{Bases}$ for the problem of, given a linear matroid $M$ of rank $k$ in matrix representation, computing the number of bases of $M$, modulo $p$; the parameterization is given by the rank $k$ of $M$.

%\begin{theorem}
%For each prime $p\geq 3$, the problems $\#_p\textsc{Forests}$ and $\#_p\textsc{Bases}$ are $\mathsf{Mod}_p\W{1}$-hard and, assuming rETH, cannot be solved in time $f(k)\cdot |G|^{o(k/\log k)}$ (resp.\ $f(k)\cdot |M|^{o(k/\log k)}$), for any function $f$.
%\end{theorem}
\forestmatroids*
\begin{proof}
We prove the lower bound for $\#_p\textsc{Forests}$; a parsimonious reduction to $\#_p\textsc{Bases}$ follows easily by (deterministic) polynomial-time matroid truncation: Given a graph $G$, the number of $k$-forests in~$G$ is equal to the number of $k$-independent sets in the graphic matroid $M(G)$ (which is a linear matroid). A $k$-truncation of $M(G)$ is a matroid $M$ of rank $k$ whose bases are in one-to-one correspondence to the $k$-independent sets of $M(G)$, and thus, to the $k$-forests in $G$. We refer the reader to~\cite{LokshtanovMPS18} for a detailed exposition of $k$-truncations and, in particular, for a deterministic algorithm which, on input a linear matroid and a positive integer $k$, computes a (linear) $k$-truncation of the matroid in polynomial time.

Now observe that $\#_p\textsc{Forests}$ is equal to the problem $\#_p\edgesubsprob(\Phi)$ for $\Phi$ being the property of being a forest. For $p=3$, the lower bound is established by Corollary~\ref{cor:forest_exact_and_3}. For each prime $p\geq 5$ we have that the family of Cayley graph expanders constructed in Theorem~\ref{thm:savegeneratorsforpquotients} is a $p$-obstruction for $\Phi$ by Proposition~\ref{Prop:forestpobstruction}. Consequently, the lower bound follows by Lemma~\ref{lem:monotonicity}.
\end{proof}

\subsection{Exact Counting of Bipartite Subgraphs}\label{sec:bipartite_hardness}

In what follows, we prove our results on counting bipartite $k$-edge subgraphs, which we restate for convenience:

\bipartitehard*

\begin{table}[htb]
    \centering
    \begin{tabular}{cccc}
        $\basegraph{\sigma}$ & No. of $\sigma$ & $|\sigma|$ & Contribution   \\ \hline
         \begin{tikzpicture}[main node/.style={circle,draw,font=\Large,scale=0.5}]
         \draw[white] (-0.2,0.2)  rectangle (1.2, 1);
         \node[main node] (A) at (0,0) {};
         \node[main node] (B) at (1,0) {};
         \node[main node] (C) at (0,0.3) {};
         \node[main node] (D) at (1,0.3) {};
         \node[main node] (E) at (0,0.6) {};
         \node[main node] (F) at (1,0.6) {};         
         \draw (A) -- (B);
         \draw (C) -- (D);
         \draw (E) -- (F);
         \end{tikzpicture} & 1 & 6 & $-120\cdot 1$\\
         \begin{tikzpicture}[main node/.style={circle,draw,font=\Large,scale=0.5}]
         \draw[white] (-0.2,0.2)  rectangle (1.2, 1);
         \node[main node] (A) at (0,0) {};
         \node[main node] (B) at (1,0) {};
         \node[main node] (C) at (0,0.3) {};
         \node[main node] (D) at (1,0.3) {};
         \node[main node] (E) at (2,0) {};
         \draw (A) -- (B) -- (E);
         \draw (C) -- (D);
         \end{tikzpicture} & 12 & 5 & $24\cdot 12$\\
         \begin{tikzpicture}[main node/.style={circle,draw,font=\Large,scale=0.5}]
         \draw[white] (-0.2,0.2)  rectangle (1.2, 1);
         \node[main node] (A) at (0,0) {};
         \node[main node] (B) at (1,0) {};
         \node[main node] (C) at (3,0) {};
         \node[main node] (E) at (2,0) {};
         \draw (A) -- (B) -- (E) -- (C);
         \end{tikzpicture} & 24 & 4 & $- 6\cdot 24$\\    
         \begin{tikzpicture}[main node/.style={circle,draw,font=\Large,scale=0.5}]
         \draw[white] (-0.2,0.2)  rectangle (1.2, 1);
         \node[main node] (A) at (0,0) {};
         \node[main node] (B) at (1,0) {};
         \node[main node] (C) at (1,0.3) {};
         \node[main node] (D) at (1,-0.3) {};
         \draw (A) -- (B);
         \draw (A) -- (C);
         \draw (A) -- (D);
         \end{tikzpicture} & 8 & 4 & $-6\cdot 8$\\    
         \begin{tikzpicture}[main node/.style={circle,draw,font=\Large,scale=0.5}]
         \draw[white] (-0.2,0.2)  rectangle (1.2, 1);
         \node[main node] (A) at (0,0) {};
         \node[main node] (B) at (1,0) {};
         \node[main node] (C) at (2,0) {};
         \node[main node] (D) at (3,0) {};
         \draw (A) to[bend left] (B);
         \draw (A) to[bend right] (B);
         \draw (C) -- (D);
         \end{tikzpicture} & 6 & 4 & $-6\cdot 6$\\     
         \begin{tikzpicture}[main node/.style={circle,draw,font=\Large,scale=0.5}]
         \draw[white] (-0.2,0.2)  rectangle (1.2, 1);
         \node[main node] (A) at (0,0) {};
         \node[main node] (B) at (1,0) {};
         \node[main node] (C) at (2,0) {};
         \draw (A) to[bend left] (B);
         \draw (A) to[bend right] (B);
         \draw (B) -- (C);
         \end{tikzpicture} & 24 & 3 & $2\cdot 24$\\       
         \begin{tikzpicture}[main node/.style={circle,draw,font=\Large,scale=0.5}]
         \draw[white] (-0.5,0.2)  rectangle (1.2, 1);
         \node[main node] (A) at (0,0) {};
         \node[main node] (B) at (1,0) {};
         \draw (A) to[bend left] (B);
         \draw (A) to (B);
         \draw (A) to[bend right] (B);
         \end{tikzpicture} & 4 & 2 & $-1\cdot 4$\\\hline
         Total contribution &  &  & $-16$\\
         ~
    \end{tabular}
    \caption{List of bipartite graphs $\basegraph{\sigma}$ for $m=3$ generators; here we give the isomorphism class of $\basegraph{\sigma}$, the number of partitions $\sigma$ with the corresponding isomorphism class, the number of blocks of sigma and the total contribution to $a(\Phi,G_i) \mod p$. The number of possible $\sigma$ can be computed by enumeration or via a variant of \cref{Lem:countsigmainisomclass} where the graph is allowed to be a multigraph with loops (using the correct notion of the automorphism group of such a graph).}
    \label{tab:bipartbasegraph}
\end{table}

Similarly as in the previous section, we will use the graph $\basegraph{\sigma}$ to prove that the property of being bipartite has an obstruction. The central argument necessary for the latter is given by the following proposition.

\begin{proposition} \label{Pro:basegraphbipartite}
Let $\Gamma$ be a finite group of odd order, let $S_0 \subseteq \Gamma$ be a set of $m$ generators and $S = \{g^{\pm 1} : g \in S_0 \} \subseteq \Gamma$ be the associated symmetric set of $2m$ generators. Let $G = \mathcal{C}(\Gamma, S)$ be the associated Cayley graph and let $\sigma$ be a partition of $S$.  Then, the graph $\fracture{G}{\sigma}$ is bipartite if and only if $\basegraph{\sigma}$ is.\footnote{The property of being bipartite extends naturally to graphs with loops and multiedges: a graph $H$ is bipartite if there is a partition $V(H) = L \sqcup R$ of the vertices into two disjoint sets such that for an edge connecting vertices $v_1, v_2$, one of them is in $L$ and one of them is in $R$. This is easily seen to be equivalent to the property that the graph has no odd cycle.}
\end{proposition}
\begin{proof}
By \cref{Lem:coverforest} there is a well-defined graph homomorphism $\Psi: \fracture{G}{\sigma} \to \basegraph{\sigma}$. Thus if $\basegraph{\sigma}$ is bipartite with a partition $V(\basegraph{\sigma})=L \sqcup R$ of the vertices, then the partition $V(\fracture{G}{\sigma}) = \Psi^{-1}(L) \sqcup \Psi^{-1}(R)$ shows that $\fracture{G}{\sigma}$ is bipartite.

For the converse direction, assume that $\basegraph{\sigma}$ is not bipartite. This means that there is a cycle 
\[w^{B_0}, w^{B_1}, \ldots, w^{B_\ell}=w^{B_0} \in V(\basegraph{\sigma})\]
in $\basegraph{\sigma}$ of odd length $\ell$, which can be specified by a starting vertex $w^{B_0} \in V(\basegraph{\sigma})$ together with a choice of $\ell$ elements $g_1, \ldots, g_\ell \in S$ of the generating set (corresponding to the oriented edges that our cycle follows). Choose any $v \in V(\fracture{G}{\sigma})=\Gamma$ then we can lift the cycle above to a walk
\[
v^{B_0}, (v\cdot g_1)^{B_1}, \ldots, (v\cdot g_1 \cdots g_\ell)^{B_0} \in V(\fracture{G}{\sigma})
\]
in the graph $\fracture{G}{\sigma}$. Let $o \geq 1$ be the order of the element $g_1 \cdots g_\ell$ in the group $\Gamma$, which must be an odd number since it divides the order of the group which is assumed to be odd. Then we can repeat the lifting procedure above $o$ times (always starting the lift at the endpoint of the previous walk) to obtain a walk in $\fracture{G}{\sigma}$ of length $o \cdot \ell$ which is odd. The endpoint of this walk is the vertex
\[
(v\cdot (g_1 \cdots g_\ell)^o)^{B_0} = v^{B_0}
\]
that we started at, so indeed we found a walk of odd length (which must contain a cycle of odd length) and so the graph $\fracture{G}{\sigma}$ is not bipartite.
\end{proof}

\begin{proposition} \label{Prop:bipartiteobstruction}
Let $\Phi$ be the property of being bipartite, let $p \geq 3$ be a prime such that there exists a family 
$\mathcal{G}=\{G_1,G_2,\dots\}$ of $(n_i,6,c)$-expanders for some positive $c$ such that the $G_i$ are Cayley-graphs for some $p$-groups $\Gamma_i$. Then $\mathcal{G}$ is an obstruction for $\Phi$ modulo $p$.
\end{proposition}
\begin{proof}
By definition, it suffices to show that for all $i$ we have that the number 
\[ 
a(\Phi,G) = \sum_{\sigma \in \mathcal{L}(\Phi,G_i)} ~\prod_{v \in V(G)} (-1)^{|\sigma_v|-1}\cdot (|\sigma_v|-1)!
\]
is nonzero modulo $p$. Since the graph $G_i$ is a $6$-regular Cayley graphs for $\Gamma_i$, there is a set $S_0 \subseteq \Gamma_i$ of three generators such that $G_i = \mathcal{C}(\Gamma,S)$ for $S=\{g^{\pm 1}: g \in S_0\}$.

As before, when evaluating $a(\Phi,G)$ modulo $p$, we can reduce to the fractures $\sigma \in \mathcal{L}(\Phi,G_i)^\Gamma$ invariant under the action of $\Gamma$. Such a $\sigma$ can be interpreted as a partition of the set $S$. Rewriting the above formula we have
\[
a(\Phi,G_i) \equiv \sum_{\sigma \in \mathcal{L}(\Phi,G_i)^{\Gamma_i}} ~\left( (-1)^{|\sigma|-1}\cdot (|\sigma|-1)! \right)^{|V(G)|} \mod p \,.
\]
Since $|V(G)| = |\Gamma|$ is a power of $p$ by the assumption that $\Gamma$ is a $p$-group, by Fermat's little theorem we have $u^{|V(G)|} \equiv u \mod p$ for all integers $u$ so that we can remove the exponent $|V(G)|$ in the formula above:
\begin{equation} \label{eqn:aBipartitefinal}
a(\Phi,G_i) \equiv \sum_{\sigma \in \mathcal{L}(\Phi,G_i)^{\Gamma_i}} ~ (-1)^{|\sigma|-1}\cdot (|\sigma|-1)!  \mod p \,.
\end{equation}
By \cref{Pro:basegraphbipartite}, the graph $\fracture{G_i}{\sigma}$ is bipartite if and only if $\basegraph{\sigma}$ is. In \cref{tab:bipartbasegraph} we list all possible isomorphism classes of bipartite graphs $\basegraph{\sigma}$ as $\sigma$ varies through the partitions of the set $S$ of size $6$ -- here we use that our graphs are $6$-regular. Summing all the contributions to \eqref{eqn:aBipartitefinal} we see 
$
a(\Phi,G_i) = -16 \neq 0 \mod p\,.
$
\end{proof}
\begin{proof}[Proof of \cref{thm:bipartite-hard}]
By \cref{lem:monotonicity} it suffices to show that the property $\Phi$ of being bipartite has an obstruction. From \cref{Prop:bipartiteobstruction} it follows that such an obstruction exists modulo $p \geq 3$ if we can find a family 
$\mathcal{G}=\{G_1,G_2,\dots\}$ of $(n_i,6,c)$-expanders for some positive $c$ such that the $G_i$ are Cayley-graphs for some $p$-groups $\Gamma_i$. By \cref{thm:main_expanders_intro} such a family exists for $p=5$.
\end{proof}

\begin{remark}\label{rem:sage_results}
Instead of considering $6$-regular expanders as above, we could consider more generally $(2m)$-regular expanders for some $m \geq 2$ and use the same method as in the proof of \cref{Prop:bipartiteobstruction} to compute $a(\Phi,G_i)$ modulo $p$. The results for the first few $m$, computed using the software SageMath \cite{sagemath}, are as follows:
\begin{center}
\begin{tabular}{cccccc}
$m$ & 2 & 3 & 4 &5 & 6  \\
$a(\Phi,G_i) \mod p$ & 0 & -16 & 192 & -16576 & 1109760 
\end{tabular}
\end{center}
From this we see two things: firstly, $4$-regular expanders (such as used in \cite{RothSW20unpub}) cannot be used to show hardness since for $m=2$ the value of $a(\Phi,G_i)$ vanishes modulo $p$ for all $p$. Secondly, for $p=2$ the number  $a(\Phi,G_i)$ vanishes for all $m$ that we checked. If it vanishes for all $m$, then the question arises whether the problem of counting bipartite $k$-edge subgraphs modulo $2$ might actually be fixed-parameter tractable, or at least allow for a significant improvement over the brute-force algorithm. \lipicsEnd
\end{remark}

\section{Counting Paths and Cycles modulo 2}
In this part of the paper, our goal is to construct faster algorithms for instances of the modular subgraph counting problem which rely on the algorithmic part of the Complexity Monotonicity, that is, we will count subgraphs modulo $2$ via counting homomorphisms modulo $2$.
More precisely, we find new algorithms for counting $k$-paths and $k$-cycles modulo $2$ as presented by Theorem~\ref{thm:main_cycles_mod_2}, which we restate for convenience; in what follows, a $k$-path (respectively a $k$-cycle) is a path (respectively a cycle) with $k$ edges.
%\begin{theorem}\label{thm:cycle_count}
%There are (deterministic) algorithms for counting $k$-paths and $k$-cycles in a graph $G$ modulo $2$ in time $k^{O(k)}\cdot |V(G)|^{k/6 + O(1)}$. \lipicsEnd
%\end{theorem}
\maincyclesmod*
We emphasize that the algorithm in the previous theorem is faster than the best known algorithms for (non-modular) counting of $k$-cycles/$k$-paths, which run in time $k^{O(k)}\cdot |V(G)|^{13k/75 + o(k)}$~\cite{CurticapeanDM17}. Furthermore, it was recently shown by Curticapean, Dell and Husfeldt~\cite{CurticapeanDH21} that counting $k$-paths modulo $2$ is $\mathsf{Mod}_2\W{1}$-hard, implying that we cannot hope for an algorithm for $k$-paths running in time $f(k)\cdot |V(G)|^{O(1)}$; we will see later (Lemma~\ref{lem:path_cycle_red}) that counting $k$-paths modulo $2$ tightly reduces to counting $k$-cycles modulo $2$, ruling out an algorithm running in time $f(k)\cdot |V(G)|^{O(1)}$ for $k$-cycles as well. 

We will start with the case of $k$-cycles.
The idea of the proof relies on the algorithmic part of the Complexity Monotonicity principle due to Curticapean, Dell and Marx~\cite{CurticapeanDM17}: We aim to express the number of $k$-cycles (mod $2$) as a finite linear combination of homomorphism counts (mod $2$), which allows us to reduce the problem to modular counting of homomorphisms. However, it turns out that this is not possible, which is ultimately due to the fact that a cycle has an even number of automorphisms: Using a transformation due to \lovasz~\cite[Equation~5.18]{Lovasz12}, Curticapean, Dell and Marx~\cite{CurticapeanDM17} proved that
\begin{equation}
    \#\subs{C_k}{\star} = \#\auts{C_k}^{-1}\sum_{\rho\geq \bot} \mu(\bot,\rho) \cdot \#\homs{C_k/\rho}{\star}\,,
\end{equation}
where the sum is over all elements $\rho$ of the partition lattice of $V(C_k)$, $\bot$ is the smallest partition (consisting only of singleton sets), and $\mu$ is the M\"obius function over the partition lattice.

However, we cannot use this transformation in case we wish to count modulo $2$ since $\#\auts{C_k}$ is even. In fact, we can show that it is impossible to express the function $\#_2\subs{C_k}{\star}$ as a finite linear combination of homomorphism counts over $\mathbb{F}_2$:

Indeed, we always have $\#_2\subs{C_k}{C_k}=1$, but for $k=2k'$ even, we know that $\#_2\homs{H}{C_k}=0$ for all graphs $H$, since $C_k$ has an automorphism of even order acting freely on the vertices of $C_k$ (and thus freely on the set of homomorphisms from $H$ to $C_k$). 

Instead, we solve this issue by considering the following intermediate problem: given a graph $G$, two nodes $s,t\in V(G)$, and a positive integer $k$, the goal is to compute the parity of the number of $s$-$t$-paths in $G$ with $k$ edges.

In fact, instead of graphs $G$ with two marked vertices $s,t$ we need to more generally consider graphs equipped with two unary relations $S$ and $T$, i.e. subsets $S, T$ of the vertices of the graph $G$. Those will eventually allow us to make the automorphism group trivial and hence to avoid the aforementioned problem. We write $J=(H,S,T)$, where $H$ is a graph and $S,T\subseteq V(H)$, and we call $J$ a $2$-\emph{labelled graph}. A homomorphism from a $2$-labelled graph $J=(H,S,T)$ to a $2$-labelled graph $J'=(H',S',T')$ is a homomorphism $\varphi$ from $H$ to $H'$ which additionally satisfies that $\varphi(v)\in S'$ if $v\in S$ and $\varphi(v)\in T'$ if $v\in T$. Readers familiar with relational structures might recognise $2$-labelled graphs as finite structures over the signature $\langle E^2,S^1,T^1\rangle$. In particular, homomorphisms between $2$-labelled graphs are precisely the (relational) homomorphisms between those structures. We write $\homs{J}{J'}$ for the set of all homomorphisms from $J$ to $J'$, and we write $\embs{J}{J'}$ for the subset of $\homs{J}{J'}$ containing only injective homomorphisms.

Given a $2$-labelled graph $J=(H,S,T)$ and a partition $\rho$ of $V(H)$, we define a $2$-labelled quotient $J/\rho=(H/\rho,S/\rho, T/\rho)$, where $H/\rho$ is the quotient graph (w.r.t.\ $\rho$) of $H$, and a vertex of $H$ is contained in $S/\rho$ (respectively $T/\rho$) if the corresponding block of $\rho$ contains a vertex in $S$ (respectively $T$).

Similarly as in case of (unlabelled) graphs, we can express $\#\embs{J}{\star}$ as a finite linear combination of homomorphism counts; we refer the interested reader to the argument in Chapter~5.2.3 in~\cite{Lovasz12} and point out that the same argument applies to $2$-labelled graphs: For each pair of $2$-labelled graphs $J$ and $J'$, we have
\begin{equation}\label{eq:mobius_2_labelled}
    \#\embs{J}{J'} = \sum_{\rho \geq \bot} \mu(\bot,\rho) \cdot \#\homs{J/\rho}{J'}\,.
\end{equation}
Now let $G$ be a graph, let $s,t\in V(G)$, and let $P_k$ be the graph with $k$ edges (assume that $V(P_k)=\{0,\dots,k\}$). Clearly, the number of $s$-$t$-paths with $k$ edges in $G$ is equal to \[\#\embs{(P_k,\{0\},\{k\})}{(G,\{s\},\{t\})}\,.\]

We are now able to present our main technical insight of the current section: using~\eqref{eq:mobius_2_labelled}, we can obtain a faster algorithm for counting $s$-$t$-paths of length $k$ modulo $2$, by understanding the parity of the M\"obius function of the partition lattice.
\begin{lemma}\label{lem:parity_paths_algo}
There is a (deterministic) algorithm that, given $G$, $s$, $t$, and $k$, computes \[\#_2\embs{(P_k,\{0\},\{k\})}{(G,\{s\},\{t\})}\] in time $k^{O(k)}\cdot |V(G)|^{k/6 + O(1)}$.
\end{lemma}
\begin{proof}
By~\eqref{eq:mobius_2_labelled}, we have
\[ \#\embs{(P_k,\{0\},\{k\})}{(G,\{s\},\{t\})} = \sum_{\rho \geq \bot} \mu(\bot,\rho) \cdot \#\homs{(P_k,\{0\},\{k\})/\rho}{(G,\{s\},\{t\})}\,.\]
The explicit formula for the M\"obius function over the partition lattice (over a set with $k+1=|V(P_k)|$ elements) reads as follows~\cite{Stanley11} (see also~\cite{CurticapeanDM17}):
\[ \mu(\bot,\rho) = (-1)^{k+1-|\rho|}\cdot \prod_{B\in \rho}(|B|-1)!\,,\]
which is even if and only if there is a block $B\in \rho$ of size at least $3$. Let us thus write $\mathcal{P}^2$ for the set of all partitions $\rho$ of $V(P_k)$ that only contain blocks of size at most $2$. We have
\begin{equation}\label{eq:mobius_hack}
  \#\embs{(P_k,\{0\},\{k\})}{(G,\{s\},\{t\})} = \sum_{\rho \in \mathcal{P}^2} \#\homs{(P_k,\{0\},\{k\})/\rho}{(G,\{s\},\{t\})}\mod 2\,.
\end{equation}
Now observe that for every $\rho\in \mathcal{P}^2$, the graph $P_k/\rho$ has degree at most $4$ since each vertex of $P_k$ has degree at most $2$ and thus a vertex of $P_k/\rho$ corresponding to a block of $\rho$ of size at most $2$ can have degree at most $4$. Note further, that $P_k/\rho$ has at most $k$ edges, since the construction of a quotient graph can never add new edges. Consequently, using well-known bounds on the treewidth of $k$-edge graphs with bounded degree~\cite[Lemma 1]{FominGSS09},\footnote{Lemma 1 in \cite{FominGSS09} provides a bound for pathwidth, but treewidth is bounded (from above) by pathwidth.} we obtain that $\mathsf{tw}(P_k/\rho)\leq k/6 +O(1)$.

Our final algorithm thus computes $\#_2\embs{(P_k,\{0\},\{k\})}{(G,\{s\},\{t\})}$ by evaluating each term in~\eqref{eq:mobius_hack} and returning the sum (modulo $2$). Since the size of $\mathcal{P}^2$ is bounded by $k^{O(k)}$, it only remains to show how the numbers
\[\#_2\homs{(P_k,\{0\},\{k\})/\rho}{(G,\{s\},\{t\})}\] for $\rho \in \mathcal{P}^2$ can be computed in time $k^{O(k)}\cdot |V(G)|^{k/6 + O(1)}$. To this end, we recall that computing $\#_2\homs{(P_k,\{0\},\{k\})/\rho}{(G,\{s\},\{t\})}$ can be cast as counting homomorphisms between two relational structures, the left one of which has treewidth at most $k/6 +O(1)$. It is well-known that the latter can be done in time $k^{O(k)}\cdot |V(G)|^{k/6 + O(1)}$: First, we can compute an optimal tree decomposition\footnote{A tree decomposition of a structure is a tree decomposition of its underlying Gaifman graph. In particular, an optimal tree decomposition of $(P_k,\{0\},\{k\})/\rho$ is an optimal tree decomposition of $P_k/\rho$. We refer the reader to~\cite[Chapter~11]{FlumG06} for a detailed exposition.} of $(P_k,\{0\},\{k\})/\rho$ in time generously bounded by $k^{O(k)}$ --- see, for instance, \cite{FominTV15} for an algorithm running in time $\exp(O(k))$. Afterwards, we use the standard dynamic programming algorithm along the tree decomposition for counting homomorphisms in time $|V(G)|^{\mathsf{tw}(P_k/\rho)} \leq |V(G)|^{k/6 + O(1)}$ (see, for instance, the algorithm in~\cite[Theorem~14.7]{FlumG06}).
\end{proof}

The algorithm for counting $k$-cycles modulo $2$ is now an easy consequence.
\begin{lemma}\label{lem:cycle_count}
There exists a (deterministic) algorithm for counting $k$-cycles in a graph $G$ modulo $2$ in time $k^{O(k)}\cdot |V(G)|^{k/6 + O(1)}$.
\end{lemma}
\begin{proof}
Given $G$ and $k$, our goal is to count the number of $k$-cycles in $G$, modulo $2$.

Let $e_1,\dots, e_m$ be any ordering of the edges of $G$, and for $i=1, \ldots, m$ let $G_i$ be the graph with edges $e_1,\dots,e_{i}$ deleted. For each $k$-cycle $C$ in $G$ let $i=i(C)$ be the minimal index of an edge $e_i$ contained in $C$. Then the subgraph $C \setminus e_i$ is contained in $G_i$ and is a $u_i-v_i$-path of length $k-1$ for $\{u_i,v_i\}=e_{i}$. The map
\[
\{k\text{-cycles $C$ in $G$}\} \to \bigcup_{i=1}^m \left\{\text{$u_i-v_i$-paths with $k-1$ edges in $G_i$}\right\}, C \mapsto (i=i(C), C \setminus e_i)
\]
is easily seen to be a bijection with inverse $(i,P) \mapsto P \cup e_i$. Thus the sum of parities of the number of $u_i$-$v_i$-paths with $k-1$ edges in $G_i$ equals, modulo $2$, the parity of the number of $k$-cycles in~$G$.

We have seen that the number of $u_i$-$v_i$-paths with $k-1$ edges in $G_i$ is equal to
\[\#\embs{(P_{k-1},\{0\},\{k-1\})}{(G_i,\{u_i\},\{v_i\})} \,,\]
the parity of which can be computed in time
\[ (k-1)^{O(k-1)}\cdot |V(G_i)|^{(k-1)/6 + O(1)} \leq k^{O(k)}\cdot |V(G)|^{k/6 + O(1)}\]
by Lemma~\ref{lem:parity_paths_algo}.
Since there are $O(|V(G)|^2)$ many edges in $G$, the total running time is bounded by
\[O(|V(G)|^2) \cdot k^{O(k)}\cdot |V(G)|^{k/6 + O(1)} = k^{O(k)}\cdot |V(G)|^{k/6 + O(1)}\,,\]
concluding the case of $k$-cycles.
\end{proof}

For the case of $k$-paths, we could reduce to the problem of counting $s$-$t$-paths of length $k$ modulo $2$ as well. However as discussed before, we will instead reduce to counting $k$-cycles modulo $2$, which also lifts the $\mathsf{Mod}_2\W{1}$-hardness from $k$-paths~\cite{CurticapeanDH21} to $k$-cycles.

\begin{lemma}\label{lem:path_cycle_red}
There exists a (deterministic) polynomial-time algorithm equipped with oracle access to the function
\[(\hat{G},\hat{k})\mapsto \#_2\subs{C_{\hat{k}}}{\hat{G}} \]
that, on input a graph $G$ and a positive integer $k$, computes the number of $k$-paths in $G$, modulo $2$. Each oracle query $(\hat{G},\hat{k})$ satisfies that $|V(\hat{G})|\in O(|V(G)|)$ and $\hat{k}\in O(k)$.
\end{lemma}
\begin{proof}
Given $G$ and $k$ as input, we construct a graph $G^-_{u,v}$ for every pair of distinct vertices $u,v\in V(G)$ as follows: We add fresh vertices $u'$ and $v'$ and edges $\{u',u\}$ and $\{v',v\}$. We also construct the graph $G^+_{u,v}$ which is obtained from $G^-_{u,v}$ by adding the edge $\{u',v'\}$. 

Now observe that $\#_2\subs{C_{k+3}}{G^+_{u,v}} - \#_2\subs{C_{k+3}}{G^-_{u,v}}$ is equal to the number of $k+3$ cycles in $G^+_{u,v}$ that contain the edge $\{u',v'\}$, modulo $2$. The latter, however, is equal to the number of $k$-paths from $u$ to $v$ in $G$, modulo $2$. 
Consequently, the parity of the total number of $k$-paths in $G$ is given by
\[ \sum_{\substack{u,v\in V(G)\\u\neq v}} \#_2\subs{C_{k+3}}{G^+_{u,v}} - \#_2\subs{C_{k+3}}{G^-_{u,v}} \mod 2\,,\]
where the sum is over all \emph{unordered} pairs of distinct vertices. The latter can easily be computed in polynomial time given access to the oracle $(\hat{G},\hat{k})\mapsto \#_2\subs{C_{\hat{k}}}{\hat{G}}$.
Additionally, each oracle query satisfies $|V(\hat{G})| = |V(G)|+2$ and $\hat{k}= k+3$, concluding the proof.
\end{proof}

\begin{proof}[Proof of Theorem~\ref{thm:main_cycles_mod_2}]
Lemma~\ref{lem:cycle_count} provides the algorithm for $k$-cycles. In combination with Lemma~\ref{lem:path_cycle_red}, we also obtain the desired algorithm for $k$-paths.
\end{proof}

\section{Modular Counting of Homomorphisms}
In the last part of the paper, we will provide an exhaustive complexity classification of parameterized modular counting of homomorphisms. 
While we used colour-prescribed homomorphisms in the previous parts of the paper for our hardness results, we will now study the uncoloured version of the problem. 

In what follows, let $p\geq 2$ be any fixed prime.
Our first goal is to understand under which conditions
\[\#_p\homs{H}{\star} = \#_p\homs{F}{\star}\,.\]

It turns out that, similar to the work of Faben and Jerrum~\cite{FabenJ15} on modular counting of homomorphisms with right-hand side restrictions, the automorphisms of $H$ of order $p$ can be used to successively ``reduce'' $H$ without changing the function $\#_p\homs{H}{\star}$. However, in contrast to~\cite{FabenJ15}, where non-fixed-points of an automorphism of order $p$ could just be deleted, we have to identify them in a quotient graph.

\subsection{Reduced Quotients}

\begin{definition}[$p$-reduced quotients]
Let $H$ be a graph and let $\alpha$ be an automorphism of $H$ of order $p$.
We define $H/\alpha$ to be the quotient graph of $H$ with respect to the partition induced by the orbits of $\alpha$, that is, each block is of the form \[B_v=\{v,\alpha(v),\cdots,\alpha^{p-1}(v)\}\,.\]

Now let $H=H_0,H_1,\dots,H_\ell$ be a sequence of graphs such that for all $i\in\{0,\dots,\ell-1\}$ we have $H_{i+1} = H_i/\alpha_i$ where $\alpha_i$ is an automorphism of $H_i$ of order $p$. If $H_\ell$ does not have an automorphism of order $p$, then it is called the $p$\emph{-reduced} quotient of $H$, denoted by $H^\ast_p$.
\end{definition}
We emphasize that the $p$-reduced quotient may have self-loops. However, we will see later that $p$-reduced quotients with self-loops can be ignored for our analysis.
Observe that we speak of \emph{the} $p$-reduced quotient in the previous definition. Indeed, we will show that it is unique, up to isomorphism.

The proof of Lemma~\ref{lem:homs_lovasz}, which is necessary for what follows and which we restate for convenience, is an easy adaption of \lovasz' well-known result on graph isomorphism via homomorphism vectors~\cite[Chapt.\ 5.4]{Lovasz12}.

%\begin{lemma}
%Let $H$ and $H'$ be graphs, neither of which has an automorphism of order $p$. Suppose that for all graphs $G$ we have that
%\[\#_p\homs{H}{G} = \#_p\homs{H'}{G} \,.\]
%Then $H$ and $H'$ are isomorphic.
%\end{lemma}
\homslovasz*
\begin{proof}
Given two graphs $F$ and $G$, let us define
\[\mathsf{Sur}(F\rightarrow G):= \{ \varphi \in \homs{F}{G}~|~\varphi \text{ is vertex-surjective} \} \,.\]
By the principle of inclusion and exclusion, as well as by the assumption of the lemma, we have that for all $G$
\begin{align*}
\#_p\mathsf{Sur}(H\rightarrow G)  &=  \sum_{S\subseteq V(G)} (-1)^{\#V(G)-\#S} \cdot \#_p \homs{H}{G[S]} \mod p\\
~&=  \sum_{S\subseteq V(G)} (-1)^{\#V(G)-\#S} \cdot \#_p \homs{H'}{G[S]} \mod p\\
~&=\#_p\mathsf{Sur}(H'\rightarrow G) \mod p\,,
\end{align*}
where $G[S]$ is the subgraph of $G$ induced by $S$. Consequently, choosing $H$ and $H'$ for $G$, we obtain
\begin{align*}
\#_p\mathsf{Sur}(H'\rightarrow H)  &= \#_p\mathsf{Sur}(H\rightarrow H) = \#_p\auts{H} \text{ , and}\\
\#_p\mathsf{Sur}(H\rightarrow H')  &= \#_p\mathsf{Sur}(H'\rightarrow H') = \#_p\auts{H'}\,.
\end{align*}
Since neither of $H$ and $H'$ have an automorphism of order $p$, we have that, by Cauchy's Theorem, 
\[\#_p\auts{H} \neq 0 \text{ and } \#_p\auts{H'} \neq 0 \,.\]
Consequently, the sets $\mathsf{Sur}(H'\rightarrow H)$ and $\mathsf{Sur}(H\rightarrow H')$ are non-emtpy, implying the existence of the desired isomorphism.
\end{proof}

Observe that the condition on $H$ and $H'$ not having automorphisms of order $p$ is necessary: Suppose that $H$ is a $k$-matching and that $H'$ is a $k'$-matching for $k\neq k'$. Then, for all graphs $G$, we have \[\#_2\homs{H}{G}=\#_2\homs{H'}{G} = 0\,,\]
even though $H$ and $H'$ are not isomorphic.

The next lemma legitimises us to restrict on $p$-reduced quotients when considering the complexity of counting homomorphisms modulo $p$.
\begin{lemma}\label{lem:quotient_reduction}
Let $H$ be a graph. We have
\[\#_p\homs{H}{\star} = \#_p\homs{H^\ast_p}{\star} \,.\]
\end{lemma}
\begin{proof}
Let $\alpha$ be an automorphism of order $p$ of $H$ and let $G$ be a graph. We show that $\#_p\homs{H}{G} = \#_p\homs{H/\alpha}{G}$, which implies the lemma. 
Observe first that $\langle\alpha\rangle$ acts on the set $\homs{H}{G}$ via $(\alpha^i \circ \varphi)(v) := \varphi(\alpha^i(v))$ for $i\in \{0,\dots,p-1\}$.

Consequently, $\homs{H}{G}$ can be partitioned into the orbits of this action. Since the size of each orbit must divide the size of the group $\langle\alpha\rangle$, which is the prime $p$, we obtain that only the fixed-points survive modulo $p$. More precisely, we have
\begin{align*}
    ~&\#_p\homs{H}{G} \\
    =& \#_p\{\varphi \in \homs{H}{G}~|~\forall v \in V(H): \varphi(v) = \varphi(\alpha(v))= \dots = \varphi(\alpha^{p-1}(v)) \}\\
    =& \#_p\homs{H/\alpha}{G} \,,
\end{align*}
concluding the proof.
\end{proof}

As the next step, we study the complexity of computing $H^\ast_p$. Here, we rely on the work of Arvind, Beigel and Lozano \cite{ABL2000} on the complexity of modular counting of graph automorphisms.

\begin{lemma}\label{lem:compute_quotient}
Let $p$ be a fixed prime.
The $p$-reduced quotient of a graph $H$ with $k$ vertices can be computed in time $\exp(\mathsf{poly}(k))$.
\end{lemma}
\begin{proof}
Observe that, given an automorphism $\alpha$ of order $p$ of $H$, we can easily construct $H/\alpha$ in polynomial time. Now consider a sequence $H=H_0,\dots,H_\ell = H^\ast_p$ such that for all $i\in\{0,\ell-1\}$ $H_{i+1} = H_i/\alpha_i$ for an automorphism $\alpha_i$ of order $p$ of $H_i$. In each step, the number of vertices decreases by at least one, hence $\ell\in O(k)$.

Therefore, it remains to show that, given a graph with at most $k$ vertices, we can find an automorphism of order $p$ in time $\exp(\mathsf{poly}(k))$, or correctly decide that none exists; in the latter case, we have found the $p$-reduced quotient.

To this end, we rely on Lemma~13 of \cite{ABL2000} which provides a polynomial-time algorithm\footnote{To be precise, Lemma~13 of \cite{ABL2000} states that a subgroup of order $p$ is returned. However, its proof reveals that an automorphism of order $p$ is returned.} equipped with oracle access to the problem of determining whether the size of the automorphism group of a graph is divisible by $p$. Moreover, the latter problem is polynomial-time many-one reducible to $\textsc{GI}$ \cite[Theorem 5]{ABL2000}. Finally, $\textsc{GI}$ can be solved in time $\exp(O(\sqrt{k})\cdot \mathsf{poly}(\log k))$~\cite{BabaiL83}, concluding the proof.  
\end{proof}
\begin{remark}\label{rem:quasipoly}
Using the quasipolynomial time algorithm for $\textsc{GI}$ due to Babai~\cite{Babai16}, we can strengthen the previous result and obtain a quasipolynomial time ($\exp(\mathsf{poly}(\log k))$) algorithm as well. However, since the revised version of the quasipolynomial time algorithm for $\textsc{GI}$ has not been fully peer-reviewed yet, we decided to state the weaker result, which suffices for all purposes of this paper.
\lipicsEnd
\end{remark}

Finally, we note that computing the $p$-reduced quotient in polynomial time is at least as hard as the graph automorphism problem $\textsc{GA}$, which asks to decide whether a given graph has a non-trivial automorphism.
\begin{lemma}\label{lem:hard_quotient_GA}
Let $p$ be a fixed prime. If the $p$-reduced quotient of a graph can be computed in polynomial time, then $\textsc{GA} \in \ccP$.
\end{lemma}
\begin{proof}
It is shown in \cite[Theorem 4]{ABL2000}  that $\textsc{GA}$ is polynomial-time many-one reducible to the problem of deciding whether the size of the automorphism group of a graph is divisible by (a fixed) $p$. Recall that the size of a (finite) group is divisible by $p$ if and only if it contains an element of order $p$. Consequently, the automorphism group of a graph $H$ is divisible by $p$ if and only if the $p$-reduced quotient of $H$ is not equal to $H$, the latter of which is equivalent to the $p$-reduced quotient having fewer vertices than $H$.
\end{proof}

\subsection{Classification for Counting Homomorphisms modulo \texorpdfstring{$p$}{p}}
We are now able to prove an exhaustive complexity classification for parameterized modular counting of homomorphisms. Recall that $\#_p\homsprob(\mathcal{H})$ is the problem of, given a graph $H\in \mathcal{H}$ and a graph $G$, to compute the number of homomorphisms from $H$ to $G$ modulo $p$. For what follows, given a class of graphs $\mathcal{H}$, we write $\mathcal{H}^\ast_p$ for the class of all $p$-reduced quotient of graphs in $\mathcal{H}$ that do not have self-loops. We begin with the algorithmic part of our classification, which we restate for convenience:

%\begin{theorem}[Algorithm]
%The problem $\#_p\homsprob(\mathcal{H})$ can be solved in time
%\[\exp(\mathsf{poly}(|V(H)|)) \cdot  |V(G)|^{\mathsf{tw}(H^\ast_p)+O(1)}\,.\]
%In particular, $\#_p\homsprob(\mathcal{H})$ is fixed-parameter tractable if the treewidth of $\mathcal{H}^\ast_p$ is bounded.
%\end{theorem}
\mainhomsalgo*
\begin{proof}
On input $H$ and $G$, our algorithm first computes $H^\ast_p$ in time $\exp(\mathsf{poly}(|V(H)|))$ by Lemma~\ref{lem:compute_quotient}. If $H^\ast_p$ has a self-loop then we output $0$ since the input graph $G$ is not allowed to have self-loops. Otherwise, we compute $\#\homs{H^\ast_p}{G}$ in time $\exp(O(|V(H)|))\cdot |V(G)|^{\mathsf{tw}(H^\ast_p)+O(1)}$ using the algorithm of Diaz et al.~\cite{DiazST02} (see also~\cite{CurticapeanDM17}). Finally, we output $\#\homs{H^\ast_p}{G}\mod p$; correctness follows from Lemma~\ref{lem:quotient_reduction}.
\end{proof}

\begin{remark}[A quasipolynomial-time algorithm]
In fact, Dalmau and Jonsson~\cite{DalmauJ04} show that $\#\homsprob(\mathcal{H})$ can be solved in polynomial time (but with a worse exponent of $|V(G)|$) whenever the treewidth of $\mathcal{H}$ is bounded. Using the quasipolynomial time algorithm for $\textsc{GI}$ due to Babai~\cite{Babai16} as explained in Remark~\ref{rem:quasipoly}, we can compute the $p$-reduced quotient in quasipolynomial time in $|V(H)|$ and then invoke the aforementioned algorithm of Dalmau and Jonsson. Consequently, for any fixed prime $p$, and for any class of graphs $\mathcal{H}$ for which $\mathcal{H}^\ast_p$ has bounded treewidth, we obtain an algorithm for $\#_p\homsprob(\mathcal{H})$ running in quasipolynomial time.\lipicsEnd
\end{remark}

Next, we show that that $\#_p\homsprob(\mathcal{H})$ is intractable whenever the treewidth of the $p$-reduced coefficients is unbounded. For convenience, we restate the formal theorem.

%\begin{theorem}[Hardness]
%Let $\mathcal{H}$ be a computable class of graphs. If the treewidth of $\mathcal{H}^\ast_p$ is unbounded, then $\#_p\homsprob(\mathcal{H})$ is $\pW{1}$-hard and, assuming rETH, cannot be solved in time
%\[f(|H|) \cdot |G|^{o(\mathsf{tw}(H^\ast_p)/\log \mathsf{tw}(H^\ast_p))} \]
%for any function $f$.
%\end{theorem}
\mainhomshard*
\begin{proof}
We reduce from $\#_p\cphomsprob(\mathcal{H}^\ast_p)$ which proves the claim by Lemma~\ref{lem:hardness_basis_mod}; note that $\mathcal{H}^\ast_p$ is recursively enumerable as $\mathcal{H}$ is computable.

Let $H^\ast_p$ and $G$ be an instance of $\#_p\cphomsprob(\mathcal{H}^\ast_p)$, that is, $H^\ast_p$ is the $p$-reduced quotient of a graph $H\in \mathcal{H}$, and $G$ is an $H^\ast_p$-coloured graph with colouring $c$. Observe first that we can find $H$ in time only depending on $H^\ast_p$. For the reduction, we consider \emph{colourful} homomorphisms as an intermediate notion: We write $\cfhoms{H^\ast_p}{G}$ for the set of all homomorphisms $\varphi$ from $H^\ast_p$ to $G$ such that $c(\varphi(V(H^\ast_p)))=V(H^\ast_p)$, that is, each colour is met (exactly) once. It is known that
\begin{equation}\label{eq:cf_cp}
\#\cfhoms{H^\ast_p}{G} = \#\auts{H^\ast_p}\cdot \#\cphoms{H^\ast_p}{G}    
\end{equation}
and that, by inclusion-exclusion,
\begin{equation}\label{eq:cf_incl_excl}
 \#\cfhoms{H^\ast_p}{G} = \sum_{J\subseteq V(H^\ast_p)} (-1)^{|J|} \cdot \#\homs{H^\ast_p}{G-J}\,,   
\end{equation}
where $G-J$ is the graph obtained from $G$ be deleting all vertices coloured (by $c$) with a colour in $J$. We refer the reader for instance to~\cite[Lemma~2.51 and~2.52]{Roth19} for proofs of (\ref{eq:cf_cp}) and (\ref{eq:cf_incl_excl}).

Now recall that, by Lemma~\ref{lem:quotient_reduction}, for all graphs $G'$ we have \[\#\homs{H^\ast_p}{G'} = \#\homs{H}{G'} \mod p\,.\]
In combination with (\ref{eq:cf_cp}) and (\ref{eq:cf_incl_excl}) we thus obtain
\begin{equation}\label{eq:main_hom_hardness}
    \#\auts{H^\ast_p}\cdot \#\cphoms{H^\ast_p}{G} = \sum_{J\subseteq V(H^\ast_p)} (-1)^{|J|} \cdot \#\homs{H}{G-J} \mod p\,.
\end{equation}
Next, we use the fact that $H^\ast_p$ has, by definition, no automorphisms of order $p$. Thus $\#\auts{H^\ast_p} \neq 0 \mod p$ and we can multiply (\ref{eq:main_hom_hardness}) by $\#_p\auts{H^\ast_p}^{-1}$. We obtain
\begin{equation*}
    \#_p\cphoms{H^\ast_p}{G} = \#_p\auts{H^\ast_p}^{-1} \sum_{J\subseteq V(H^\ast_p)} (-1)^{|J|} \cdot \#_p\homs{H}{G-J} \mod p\,,
\end{equation*}
which we can evaluate using our oracle for $\#_p\homsprob(\mathcal{H})$; recall from above that we can find $H$ in time only depending on $H^\ast_p$.
Finally, since the number of terms in the sum only depends on $H^\ast_p$, we obtain the desired parameterized Turing-reduction from $\#_p\homsprob(\mathcal{H}^\ast_p)$ yielding $\pW{1}$-hardness by Lemma~\ref{lem:hardness_basis_mod}. Furthermore, since each graph $G-J$ is of size at most~$|G|$, the conditional lower bound under rETH transfers from Lemma~\ref{lem:hardness_basis_mod} as well. This concludes our proof.
\end{proof}

\bibliography{main}

\end{document}